\documentclass[journal]{IEEEtran}
\ifodd 0
\newcommand{\rev}[1]{{\color{blue}#1}}
\else
\newcommand{\rev}[1]{#1}
\fi

\ifodd 0
\newcommand{\yu}[1]{{\color{blue}#1}}
\else
\newcommand{\yu}[1]{#1}
\fi

\ifodd 0
\newcommand{\aug}[1]{{\color{blue}#1}}
\else
\newcommand{\aug}[1]{#1}
\fi

\ifodd 0
\newcommand{\augg}[1]{{\color{blue}#1}}
\else
\newcommand{\augg}[1]{#1}
\fi

\ifodd 0
\newcommand{\auggg}[1]{{\color{blue}#1}}
\else
\newcommand{\auggg}[1]{#1}
\fi

\ifodd 0
\newcommand{\final}[1]{{\color{blue}#1}}
\else
\newcommand{\final}[1]{#1}
\fi

\ifodd 0
\newcommand{\res}[1]{{\color{blue}#1}}
\else
\newcommand{\res}[1]{#1}
\fi

\ifodd 0
\newcommand{\resf}[1]{{\color{blue}#1}}
\else
\newcommand{\resf}[1]{#1}
\fi

\IEEEoverridecommandlockouts

\ifCLASSINFOpdf
\else
\fi
\usepackage{amsmath, amssymb, amsthm, mathrsfs,bm}
\usepackage{graphicx}
\usepackage{array}
\usepackage{tabularx}
\usepackage{amssymb,amsmath,color,graphicx}
\usepackage{subfigure}
\usepackage{cite}
\usepackage{verbatim}
\usepackage[linesnumbered,ruled,vlined]{algorithm2e}
\usepackage{mathtools}
\usepackage{multirow}
\usepackage{setspace}

\usepackage{tikz,pgf}
\newtheorem{thm}{Theorem} 
 
\newtheorem{obs}{Observation}
\newtheorem{pps}{Proposition}

\newcolumntype{L}[1]{>{\raggedright\arraybackslash}p{#1}}
\newcolumntype{C}[1]{>{\centering\arraybackslash}p{#1}}
\newcolumntype{R}[1]{>{\raggedleft\arraybackslash}p{#1}}

\begin{document}
%\begin{spacing}{1.7}
\title{\res{Mobile Data Trading: Behavioral Economics Analysis and Algorithm Design}}


\author{\textcolor[rgb]{0,0,0}{Junlin Yu,~\IEEEmembership{Student Member,~IEEE}, Man Hon Cheung, Jianwei Huang,~\IEEEmembership{Fellow,~IEEE}, and H. Vincent Poor,~\IEEEmembership{Fellow,~IEEE}}
%\authorblockA{\IEEEauthorrefmark{1}Department of Information Engineering, The Chinese University of Hong Kong, Hong Kong, China\\}
%\authorblockA{\IEEEauthorrefmark{2}Department of Electrical Engineering, Princeton University, Princeton, NJ, USA\\}
%\authorblockA{Emails: \{yj112, mhcheung, jwhuang\}@ie.cuhk.edu.hk, poor@princeton.edu}
%}
\thanks{Manuscript received September 22, 2016; revised January 13, 2017; accepted January 26, 2017.
This work is in part supported by the General Research Funds (Project Number CUHK 14202814, 14206315, and 14219016) established under the University Grant Committee of the Hong Kong Special Administrative Region, China, and in part by the U.S. Army Research Office under Grant W911NF-16-1-0448 and U. S. National Science Foundation under Grant ECCS-1549881. Part of this paper was presented in \cite{wiopt15}.}
\thanks{J. Yu, M. H. Cheung, and J. Huang are with the Department of Information Engineering, the Chinese University of Hong Kong, Hong Kong, China; Emails: \{yj112, mhcheung, jwhuang\}@ie.cuhk.edu.hk. H. V. Poor is with the Department of Electrical Engineering, Princeton University, Princeton, NJ, USA; Email: poor@princeton.edu. \final{The authors would like to thank Chap Yin Liu and Fong Yuen Pang for their help and support in developing the mobile data trading app.}}
}


%Telephone: (800) 555--1212, Fax: (888) 555--1212}


\maketitle
\thispagestyle{empty}

\begin{abstract}

Motivated by the recently launched mobile data trading markets (e.g., China Mobile Hong Kong's 2nd exChange Market), in this paper we study the mobile data trading problem under the future data demand \emph{uncertainty}. \res{ We introduce a brokerage-based market, where sellers and buyers propose their selling and buying quantities, respectively, to the trading platform that matches the market supply and demand. }
%The platform is modelled to act as a broker\textcolor[rgb]{0,0,0}{, which} facilitates the trade by matching the supply and demand.
To understand the users' realistic trading behaviors, a prospect theory (PT) model from behavioral economics is proposed, which includes the widely adopted expected utility theory (EUT) as a special case. Although the PT modeling leads to a challenging non-convex optimization problem, the optimal solution can be characterized by exploiting the unimodal structure of the objective function.
%Building upon our analysis, we develop an Android mobile data trading app, which estimates the user's risk preference, and provides trading recommendations dynamically considering the latest market and usage information. 
\res{Building upon our analysis, we design an algorithm to help estimate the user's risk preference and provide trading recommendations dynamically, considering the latest market and usage information.}
It is shown in our simulation that the risk preferences have a significant impact on the user's decision and outcome: a risk-averse dominant user can guarantee a higher minimum profit in the trading, while a risk-seeking dominant user can achieve a higher maximum profit.
By comparing with the EUT benchmark, it is shown that a PT user with a low reference point is more willing to buy mobile data. Moreover, when the probability of high future data demand is low, a PT user is more willing to buy mobile data due to the probability distortion comparing with an EUT user.
\end{abstract}

%\noindent \small{\textbf{Index Terms}: Behavioral economics, prospect theory, expected utility theory, mobile data trading.}
\normalsize
%\begin{IEEEkeywords}
%Prospect Theory; spectrum trading; spectrum sensing; Expected Utility Theory; Sequential Optimization.
%\end{IEEEkeywords}

% \fi
\IEEEpeerreviewmaketitle

%intro 2: Dynamic spectrum access can be categorized under three approaches: dynamic exclusive use, open sharing and hierarchical access [ref]. Dynamic exclusive use allows licensees to sell and trade spectrum dynamically [ref]. Hierarchical access model opens licensed spectrum to secondary users while limiting the interference perceived by primary users (licensees) [ref]. Open sharing allows sharing among peer users as the basis for managing a spectral region [ref]. In his paper, we consider a secondary operator who achieve spectrum from a spectrum owner by dynamic spectrum sensing and leasing.

%\newpage
%=====================================================
\section{Introduction}
%=====================================================
\subsection{Background and Motivation}
With the increasing computation and communication capabilities of mobile devices, global mobile data traffic has been growing tremendously in the past few years \cite{eri,cisco}. 
  %\rev{(We should add some discussion on 2CM from CMHK.)}
\res{One way to alleviate the tension between the mobile data demand and the network capacity is to utilize the spectrum more efficiently, for example through spectrum sharing \cite{hanzz,7118253,6354285}.
Another way is to flatten the demand curve through pricing \cite{zhang2014time,ha2012tube,7337420,6848090,6849296}. More specifically}, the mobile service providers have been experimenting with several innovative pricing schemes, such as usage-based pricing, shared
data plans, and sponsored data pricing, to extract more revenue from the growing data while sustaining a good service quality to users.
However, the above mentioned schemes do not fully take advantage of the heterogeneous demands across all mobile users, and unused data in the monthly plan will be cleared at the end of the month. 
Recently, China Mobile Hong Kong (CMHK) launched the first 4G data trading platform in the world, called the 2nd exChange Market (2CM), which allows its users to trade their monthly 4G mobile data quota directly with each other.\footnote{\res{The three major mobile operators in China (China Mobile, China Unicom, and China Telecomm) now all support such a trading platform \cite{china1,china2,china3}.}} 
In this platform,  a seller can sell some of his remaining data quota of the current month on the platform with a desirable price set by himself. 
If a buyer wants to buy some data at the listed price, the platform will help complete the transaction and transfer \textcolor[rgb]{0,0,0}{the} proper data amount from the seller's quota to the buyer's quota of that month.

However, there is a shortcoming of the current one-sided 2CM mechanism. More specifically, 2CM is a sellers' market, where a buyer cannot list his desirable buying price and quantity. \yu{This means that a buyer needs to frequently check the platform to see whether the current (lowest) selling price is acceptable, while a seller does not know whether he can sell the data at his proposed price immediately. In other words, both buyers and sellers suffer from the incomplete information of this one-sided market.} 

\res{To improve the existing CMHK mechanism, we apply the widely used Walrasian auction in the stock markets \cite{1979, stoll1990stock}. In such a mechanism, both sellers and buyers can submit their selling and buying prices and quantities to the platform.}
The platform clears some transaction whenever the highest buying price among buyers is no smaller than lowest selling price among sellers. \yu{We are interested in understanding how a user should participate in such a market under the \emph{uncertainty} of his future data usage, given his remaining data quota of the current month and the current prices and quantities of other sellers and buyers.} More specifically, we would like to answer the following questions: \emph{(i) Should a user choose to be a seller or a buyer? (ii) How much should he sell or buy?} 

%To this end, we proposed a double-auction business model, in which the exchange market acts as a broker. The sellers and buyers submit their ask prices to the broker, and then the market choose the 
%as the market clear price.

\yu{\res{The key feature of the user's decision problem is the future data demand uncertainty, as there will be a \emph{satisfaction loss} if the user's realized demand exceeds his monthly data quota (after incorporating the results of data trading), and there will be a waste of money if the user's realized demand is less than his monthly data quota (if the user purchases too much data from the market)}. A typical approach of solving a user's decision problem with uncertainty is to maximize the user's expected utility, i.e., the expected utility theory (EUT) (e.g., \cite{schmeidler1989subjective}).}
%Then by comparing these utilities, the user can decide whether to be a seller or a buyer. 
Empirical evidences \cite{kahneman_pt79,tversky_ai92}, however, have shown that the EUT model can deviate from real world observations significantly due to the complicated psychological aspect of human decision-making.
Alternatively, researchers in behavioral economics have \textcolor[rgb]{0,0,0}{shown} that \emph{prospect theory} (PT), which establishes a more general theoretical model that includes EUT as a special case, provides a psychologically more accurate description of the decision making under uncertainty, and explains some human behaviors that seem to be illogical under EUT \cite{kahneman_pt79}.\footnote{\res{The expected utility theory (EUT), which is based on an axiomatic system, has an underlying assumption that decision makers are rational and risk-averse when facing uncertainties \cite{eutt}. PT is one of the most
widely used generalizations of EUT, as PT incorporates human emotions and psychology into the utility theory \cite{jin_bp08,ptr1,ptr2}.} }

More specifically, PT shows that a decision maker evaluates an outcome significantly differently from what people have commonly assumed in EUT in several aspects:
	(1) \emph{Impact of reference point}: A PT decision maker's evaluation is based on the \emph{relative} gains or losses comparing to a reference point, instead of the absolute values of the outcomes.  %but not on the magnitude of the utility of the option. 
  (2) \emph{The s-shaped asymmetric value function}: A PT decision maker tends to be \emph{risk-averse} when considering gains and \emph{risk-seeking} when considering losses. Furthermore, the PT decision maker is \emph{loss averse}, in the sense that he strongly prefers avoiding losses to achieving gains. 
  (3) \emph{Probability distortion}: A PT decision maker tends to \emph{overweigh} low probability events and \emph{underweigh} high probability events. 
  %prospect theory is a behavioral theory that can decribe these irrational behaviors in a mathematical model[ref]. The valuation in PT is significantly different from EUT in the following three ways. PT states that people evaluate the outcomes in terms of relative gains or losses regarding a reference point rather than the final asset position; People tend to be risk-averse, since they strongly prefer avoiding losses to acquiring gains; People tend to overreact to small probability events, but underreact to medium and large probabilities.  
As PT has been shown to be more accurate than EUT in predicting human behaviors \cite{barberis2001prospect, kahneman_pt79, tversky_ai92}, it has been applied to gain better understandings of \textcolor[rgb]{0,0,0}{financial markets} \cite{jin_bp08} and labor markets \cite{camerer2004behavioral}. \aug{However, there does not exist any PT-based studies in understanding the users' decisions in the mobile data trading market. }

\subsection{Contributions}\label{sec:con}
In this paper, we aim to understand a user's realistic trading behavior in a mobile data market, considering his future data demand uncertainty. 

\res{In the first part of the paper, we focus on deriving the optimal trading decision of a user based on his remaining quota and possible demand till the end of the billing cycle, without considering future possible tradings.}\footnote{For example, the billing cycle of a monthly data plan is a month. \res{In a more general case, a user may trade multiple times in the same billing cycle. We can model the user's decision problem in this general setting as a Dynamic Programing (DP) problem, which is much more challenging to solve. We leave this to the future work.}} 
\res{Specifically, we formulate the problem as a two-stage optimization
problem, where the user decides whether to be a seller or a buyer in Stage I (at a particularly given trading time), and then determines his selling quantity (as a seller) or buying quantity (as a buyer) in Stage II.}
Besides considering the optimal decision of a risk-neutral user in the EUT framework, we will also consider the impact of the user's risk preferences on the decision. 
To be more specific, a risk-seeking decision maker is aggressive and wants to achieve a high maximum profit even with the risk of a low minimum profit, while a risk-averse decision maker is conservative and wants to guarantee a satisfactory level of minimum profit.
The PT provides a comprehensive analytical framework for understanding the optimal decisions of different types of decision makers. However, the corresponding optimization is non-convex hence is challenging to solve.  
Nevertheless, by exploiting the unimodal structure in each sub-interval of the feasible set, we can obtain the globally optimal solution of the non-convex optimization problem.
We further discuss the practical insights by comparing the analysis under PT and EUT for the case with binary outcomes. 

\res{In the second part of the paper, we introduce an algorithm
for autonomous and adaptive data trading based on the theory developed in the first part. In such an algorithm, a user can trade multiple times during a billing cycle, with each trading decisions being made in a ``myopic'' fashion without considering the possible future trading opportunities. Since a user's risk preference will significantly impact the result of this algorithm, we design another algorithm to estimate the user's risk preference. We implement the algorithms on an Android app\footnote{Notice that the app is based on the real CMHK market, hence it is different from our theory in two aspects. First, there is no buyer's market, and the seller will make decision at the price slightly lower than the minimum selling price. Second, a user can make several decisions during a billing cycle, based on his current quota and future data demand uncertainty.} and evaluate our algorithm's
performances under different risk preferences through numerical examples.}

Our key contributions of this paper are summarized as follows:

\begin{itemize}
	\item \emph{Behavioral economics modeling of uncertainty}: We use prospect theory to model the user's trading behavior under future data demand uncertainty. We consider all three key characteristics of PT and derive key insights that characterize the optimal selling and buying decisions. 
	\item \emph{Characterization of the optimal trading solution}: \aug{Despite the non-convexity of the user's decision problem, we are able to obtain the globally optimal solution by exploiting the convexity and unimodality in different sub-intervals of the feasible set.} We further evaluate how different behavioral characteristics (i.e., reference point, probability distortion, and s-shaped valuation) affect this optimal decision.
	\item \emph{Engineering insights on risk preferences}: Comparing with the benchmark EUT result, we show that a PT user with a low reference point is more willing to buy mobile data and less willing to sell mobile data. Moreover, a PT user is even more willing to buy mobile data when the probability of high future data demand is small, mainly due to the probability distortion.
	\item \res{\emph{Evaluation of algorithms}: We evaluate the user's profit under our proposed algorithm numerically.} Based on this, we show that a risk-averse user can achieve the highest minimum profit, a risk-seeking user can achieve the highest maximum profit, and a risk-neutral user can achieve the highest average profit.
\end{itemize}

 Next we review the literature in Section II. In Section III, we formulate the user's utility functions under both EUT and PT. In Section IV, we compute the optimal user decision, and illustrate the insights through a special case of binary outcomes. \res{In Section V, we explain the implementation of our multi-trade algorithm on an Android app, which estimates the user's risk preferences and compute the optimal trading decisions accordingly. In Section VI, we numerically evaluate the user's optimal decision based on several model parameters, and compute the overall profit that our algorithm can achieve in a billing cycle under different risk preferences.} We conclude the paper in Section VII.

\section{Literature Review} \label{sec:review}

\subsection{Mobile Data Pricing}
\augg{Previous studies have focused on several different mobile data pricing schemes, such as usage-based, flat-rate, cap pricing, time-dependent pricing, location dependent pricing, and shared data plans\cite{zhang2014time,ha2012tube,7337420,6848090,6849296}.
%Authors in \cite{kesidis2008flat} and \cite{ford2012most} compared different pricing schemes for broadband networks, including three types of schemes: flat-rate scheme, usage based scheme, and cap scheme.
For example, Zhang \emph{et al.} in \cite{zhang2014time} studied the ISP's revenue maximization problem with three different pricing schemes, including flat-rate scheme, usage based scheme, and cap scheme.
Ha \emph{et al.} in \cite{ha2012tube} showed that time dependent pricing schemes can reduce the ISPs' need of over-provision network resources at peak times, hence can reduce the ISPs' operational costs. 
Ma \emph{et al.} in \cite{7337420} showed that time and location aware pricing for mobile data traffic can incentivize users to smooth traffic and reduce network congestion.
The authors in \cite{6848090} and \cite{6849296} showed that shared data plan can decrease the average unit usage cost for the users by allowing multiple users share the same pool of data quota.
In this paper, we consider the mobile data trading among all the users subscribing to the same service (such as 4G) from a mobile operator, hence takes advantage of the heterogeneous demands of a larger number of users in the mobile data market. }

\subsection{\augg{User Decisions in Communication Networks and Smart Grids under PT}}
The research of using PT to understand user decisions in communication networks and smart grids is at its infancy stage. Due to the complexity of modeling and analysis, all previous literature have only considered one or two of the three key features of PT in the modeling. Li \emph{et al.} in \cite{li_pi12,li2014users} and Yang \emph{et al.} in \cite{yang2014impact} compared the equilibrium strategies of a binary decision game among wireless network end-users under EUT and PT, where they considered a linear value function with the probability distortion.
Xiao \emph{et al.} in \cite{6895275} and Wang \emph{et al.} in \cite{wang2014integrating} characterized the unique Nash Equilibrium of an energy exchange game among microgrids under PT, where they considered a linear value function with the probability distortion. Yu \emph{et al.} in \cite{yu2014spectrum} studied a secondary wireless operator's spectrum investment problem, where they considered a linear probability distortion and s-shaped value function. \augg{To the best of our knowledge, this paper is the first work that studies a mobile data trading problem under PT, where we capture all three characteristics of PT when modeling and analyzing the problem}. As a result, we are able to \textcolor[rgb]{0,0,0}{gain} a more thorough understanding of the user's optimal decisions based on his specific risk preferences and derive more insights.

%%%%%%%%%%%%%%%%%%%%%%%%%%%%%%%%%%%%%%%%%%%%%%
\section{System Model}\label{sec:model}
%%%%%%%%%%%%%%%%%%%%%%%%%%%%%%%%%%%%%%%%%%%%%%
  In this paper, we consider the trading decision of a single user\footnote{Note that there are many users in the market, and a single user' decision will not significantly impact on the market. So from a single user's point of view, the market can be viewed as exogenously given and stochastically changing.} in a mobile data trading platform. 
  We first introduce the mobile data trading market in Section \ref{sec:III-A}.
  Then we discuss the user's profile in Section \ref{sec:III-B} and his risk preferences model in Section \ref{sec:III-C}.
  In Section \ref{sec:III-D}, we formulate the user's two-stage trading decision problem.

%%%%%%%%%%%%%%%%%%%%%%%%%%%%%%%%%%%%%%%%%
\subsection{Mobile Data Trading Market}\label{sec:III-A}
%%%%%%%%%%%%%%%%%%%%%%%%%%%%%%%%%%%%%%%%%%%
  We consider a two-sided mobile data trading platform as shown in Fig.~\ref{fig:market}. The \emph{seller's market} lists the sellers' proposed prices and the corresponding amount of data available for sale at each price.\footnote{\res{Under the continuous double auction mechanism, every user can make his trading decision (i.e., his role, price, and quantity) at any time. He will submit his decision to the platform as his bid. When the platform receives a user's bid, it will try to match the existing bids with this new bid, or keep
this new bid online if it cannot be matched immediately.}} \res{If a buyer wants to purchase some data quota immediately, he can choose to purchase at the \emph{minimum selling price} $\pi_s^{\min}$ in the seller's market.} In Fig.~\ref{fig:market}, we have $\pi_s^{\min} = \$20$.
 Similarly, the \emph{buyer's market} lists buyers' proposed prices and the corresponding amount of data demand at each price. \res{If a seller wants to sell some data quota immediately, he can choose to sell at the \emph{maximum buying price} $\pi_b^{\max}$.} In Fig.~\ref{fig:market}, we have $\pi_b^{\max} = \$16$.\footnote{\res{As evidenced in the real CMHK market, we assume that the quantity associated with the minimum selling price is large enough, such that a single buyer who wants to complete the trade immediately can simply consider a single price $\pi_s^{\min}$. Similar to the buying decision, we assume that the quantity associated with the maximum buying price is large enough such that a single seller who wants to sell his data immediately can simply consider a single price $\pi_b^{\max}$.}}
  
  \begin{figure}	
	\centering		
	\includegraphics[width=0.48\textwidth]{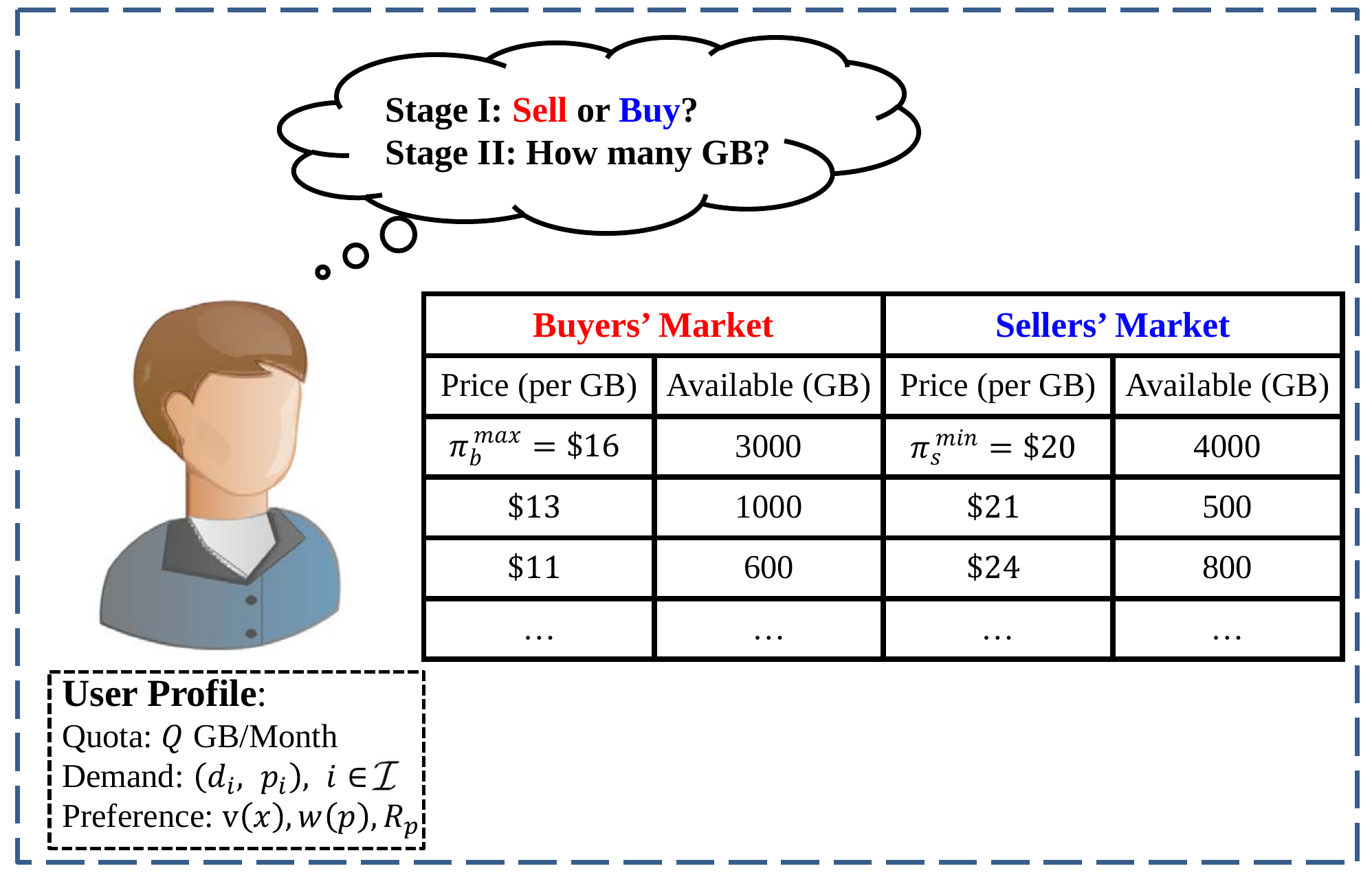}
	\caption{An example of the trading decision with the data trading platform.}
	\label{fig:market}
\end{figure}
%We consider a two sided mobile data trading market with a large number of users\footnote{Notice that the two sided market is slightly different with the seller market in real CMHK system, and we have introduced the advantages of the two sided market in Section I-A. In Section \ref{sec:implementation}, we will show the real CMHK market with $4740$ GB to be sold at the same price, which validates our assumption. Since the number of users in the market is large, a single user's choice will have a negligible impact on the market. }.
Note that in Fig.~\ref{fig:market}, the maximum buying price ($\pi_b^{\max}=\$16$) is lower than the minimum selling price ($\pi_s^{\min}=\$20$). 
This is because those selling offers with prices less than $\$16$ have already been cleared by the market, so are those buying requests with prices higher than $\$20$. 
%Based on the CMHK market, we assume that the quantity associated with the maximum buying price and the quantity associated with the minimum selling price are both large enough. 
%This means that for a single user who wants to complete the trade immediately, he only needs to consider the maximum buying price $\pi_b^{\max}$ and minimum selling price $\pi_s^{\min}$, and ignore all other prices\footnote{If a user does not need to complete the trade immediately, he may choose to list a selling price higher than $\pi_s^{\min}$ or list a buying price lower than $\pi_b^{\max}$. However, due to the fact that the selling quantity with minimum sell price is very large, it is rarely possible for the quantity with higher price to be realized.}. In particular, if a user chooses to be a seller in Stage I, his selling price in Stage II will be $\pi_s=\pi_b^{\max}$, so that he can sell the data immediately with the maximum price that some existing buyer can accept. Similarly, if a user chooses to be a buyer in Stage I, he will set his buying price in Stage II as $\pi_b=\pi_s^{\min}$, so that he can buy the data immediately with the minimum price that some existing seller can offer. Hence we will ignore the users' pricing decisions in the rest of the paper. 

%%%%%%%%%%%%%%%%%%%%%%%%%%%%%%%%%%%%%%%%%%%%%%
\subsection{User's Profile}\label{sec:III-B}
%%%%%%%%%%%%%%%%%%%%%%%%%%%%%%%%%%%%%%%%%%%%%%

\emph{Remaining Data Quota:} \res{For the analytical model in Sections III and IV, we assume that the user makes the trading decision without considering potential future tradings in the same billing cycle.}\footnote{\res{We conducted a survey with over 50 CMHK users, and found that users' decisions are not fully rational and are usually myopic due to bounded rationality \cite{6172627}. More specifically, around 60\% of the users trade only once during one billing cycle.}} We use $Q$ to denote his remaining data quota at the time of decision. For example, if the user subscribes to a data plan of 5 GB per month and he has consumed 2 GB so far, then $Q=3$ GB for the remaining time of the billing cycle.

\emph{Demand Uncertainty:} The user has an uncertainty regarding his future data demand from now till the end of the billing cycle.  
We assume that his future data demand $d$ follows a discrete distribution over the set of $I$ possible values, $\{d_i$: $i \in \mathcal{I}=\{1,\ldots,I\}$, $d_1<\ldots<d_I\}$, with the corresponding probability mass function $\mathbb{P}(d = d_i) = p_i$ with $\sum_{i=1}^I p_i = 1$.\footnote{\res{Mathematically, when we choose the number of possible realizations $I$ to be large enough, the discrete distribution can well approximate a continuous distribution \cite{zhengp}. }} To avoid the trivial case, we assume that $d_1<Q$ and $d_I>Q$. We further define $\hat{\imath}$ as the index that $d_{\hat{\imath}}<Q$ and $d_{\hat{\imath}+1}\geq Q$. 
%at the beginning of the billing cycle (i.e., month), and has a prediction of this month's total demand $d(t)$ at time $t$. By selling or buying data in the market, a user can change his remaining data quota $Q(t)$ (for the current month only). Notice that the user can make multiple trading decisions in a billing cycle, and his prediction of data demand is different at different stage of the billing cycle\footnote{For example, after watching a movie using cellular data at time $t$, a user's prediction $d(t)$ will increase.}. Hence, in Section \ref{sec:implementation}, we propose a dynamic algorithm, where quota $Q(t)$ is updated after every trade, and the prediction of demand $d(t)$ is updated according to the real-time usage. 

%In this section, we only consider the user's trading decision at a particular time $t$\footnote{For simplicity, we write $d(t)$ as $d$, and write $Q(t)$ as $Q$.}. We assume that the user's prediction on the total data consumption has a set $\mathcal{I} = \{1,\ldots,I\}$ of $I$ possible values $\{d_i$: $i \in \mathcal{I}$, $d_1<d_2<\ldots<d_I\}$ with the corresponding  probabilities {$\{p_i$: $i\in\mathcal{I}$, $\sum_{i=1}^I p_i = 1\}$, where $d_1<Q$ and $d_I>Q$\footnote{The analysis for the case where all $\{d_i$: $i = 1, ..., I\}$ are higher (or lower) than the monthly quota $Q$ is relatively trivial, hence is omitted here due to space limitations.}. We further define $\hat{\imath}$ as the index that $d_{\hat{\imath}}<Q$ and $d_{\hat{\imath}+1}\geq Q$. 

\emph{Satisfaction Loss:} \aug{The user's data plan has a two-part pricing tariff, where the user pays a fixed fee for the data consumption up to a monthly quota (5 GB in the previous example), and a linear high usage-based cost for any extra data consumption. Such a pricing model is widely used by major operators like AT\&T in US and CMHK in Hong Kong\cite{zhang2014time}. Specifically, the user needs to pay a price of $\kappa$ (\$/GB)\footnote{For example, for a 4G CMHK user, $\kappa=60$.} if the user's future data demand $d$ exceeds his remaining data quota $Q$. We define the \emph{satisfaction loss} of the user as the additional payment (which is a non-positive term) for exceeding the monthly quota:}
\begin{equation}\label{eq:ly}
     L(y)=\left\{
    \begin{aligned}
    &0,  &   &\text{ if }y\geq 0,\\
    &\kappa y, &   &\text{ if }y<0,\\
    \end{aligned}
    \right.
\end{equation}
\noindent where $y<0$ means that the quota is exceeded. Without data trading, $y=Q-d$.

%The linear coefficient $\kappa\geq0$ represents the usage-based pricing imposed by the operator.\footnote{}

%%%%%%%%%%%%%%%%%%%%%%%%%%%%%%%%%%%%%%%%%%%%%%%%%%%%%%
\subsection{Risk Preferences}\label{sec:III-C}
%%%%%%%%%%%%%%%%%%%%%%%%%%%%%%%%%%%%%%%%%%%%%%%%%%%%%%

To model the user's data trading problem under future data demand uncertainty, we consider the following three features of PT, namely  reference point $R_p$, s-shaped \emph{value function} $v(x)$, and \emph{probability distortion function} $w(p)$ \cite{kahneman_pt79, prelec1991decision}.

%-------------------------------------------------
\subsubsection{Reference Point}
%-------------------------------------------------
The reference point $R_p$ indicates the user's physiological target of the outcome. The user considers an outcome a \emph{gain} if it is higher than the reference point, and a \emph{loss} if it is lower than the reference point. A high reference point means that the user is more likely to treat an outcome as a loss, and a low reference point means that he is more likely to treat an outcome as a gain. 
This will significantly affect the user's subjective valuation of the outcome, as we will explain next. 

%We note EUT is a special case of PT, with the parameter choice of $R_p=0$.\footnote{In fact, as long as the other parameters $\lambda=\beta=\mu=1$ (which will be introduced next), choosing a non-zero value of $R_p$ will just induce a constant shift of the EUT utilities, without affecting the optimal decision under EUT.} 
%-------------------------------------------------
\subsubsection{S-shaped Asymmetrical Value Function}
%-------------------------------------------------
\begin{figure}[t]
\setlength{\belowcaptionskip}{-5mm}
%\begin{tabular}{cc}   
\begin{minipage}{0.48\linewidth}
  \centerline{\includegraphics[width=4cm]{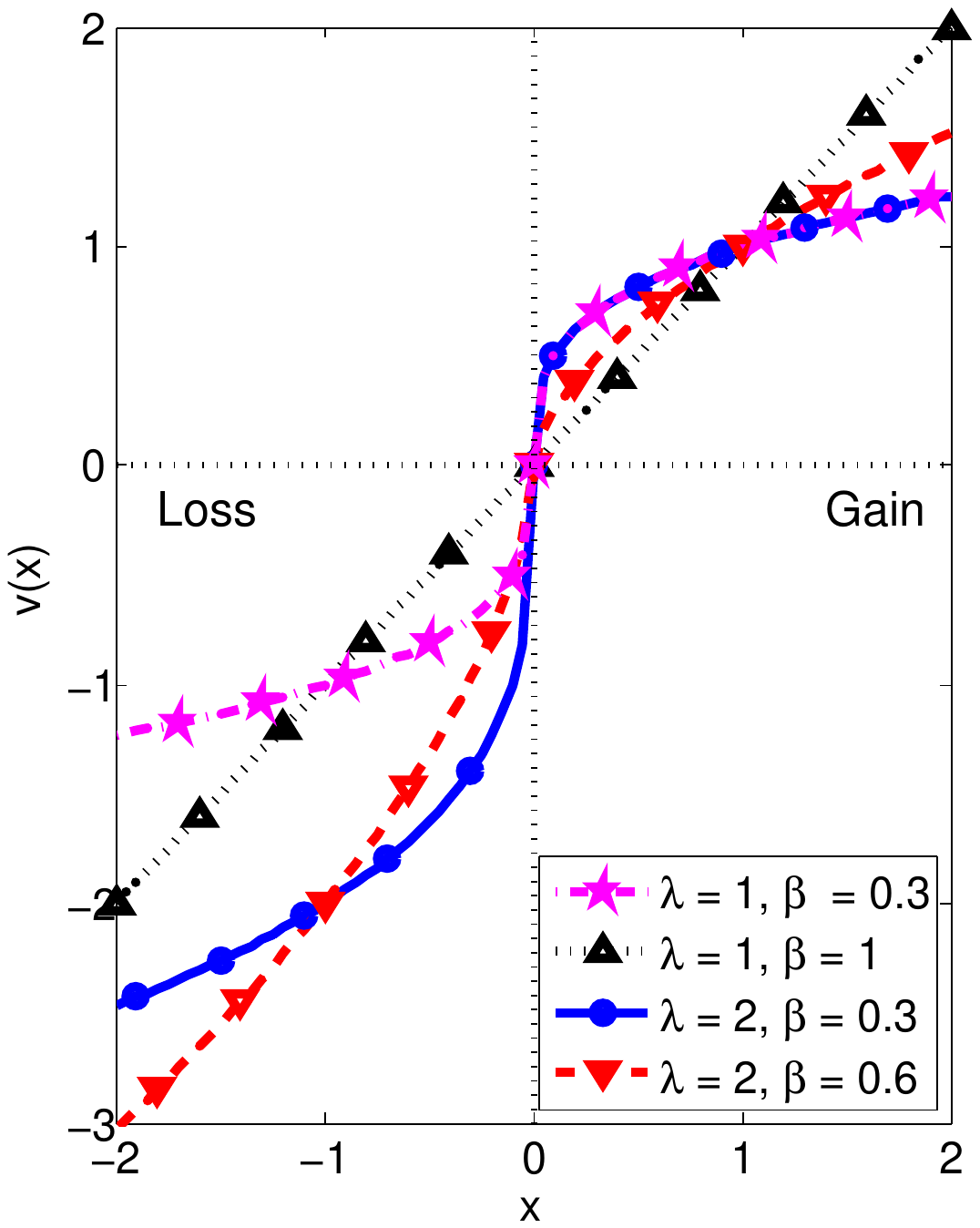}}
  \centerline{(a) $v(x)$}
\end{minipage}
\hfill
\begin{minipage}{.48\linewidth}
  \centerline{\includegraphics[width=4.5cm]{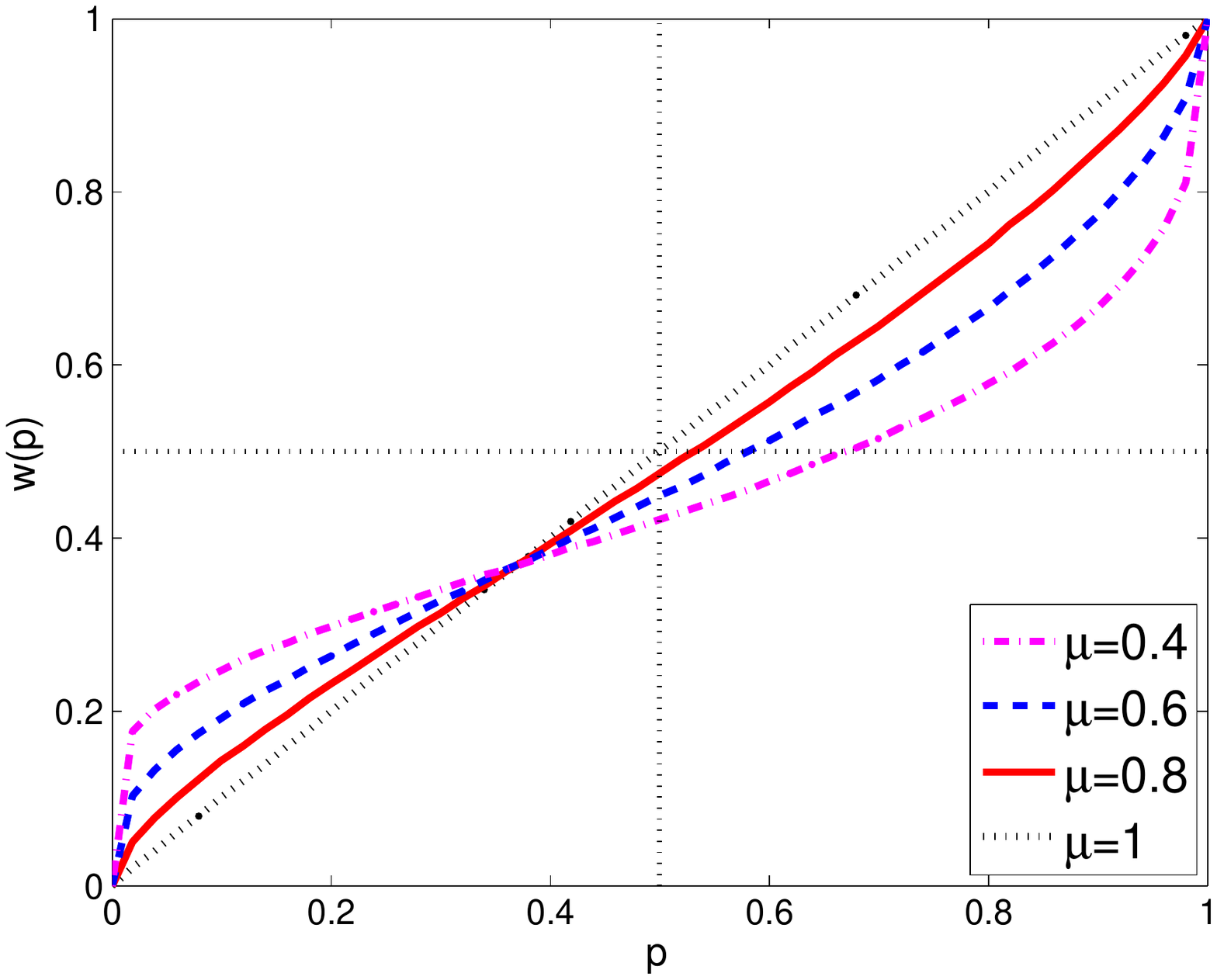}}
  \centerline{(b) $w(p)$}
\end{minipage}
\vfill
\caption{The s-shaped asymmetrical value function $v(x)$ and the probability distortion function $w(p)$ in PT.}
\label{fig:valuepdfun}
\end{figure}

Fig.~\ref{fig:valuepdfun}(a) illustrates the value function $v(x)$, which maps an objective outcome $x$ to the user's subjective valuation $v(x)$. Notice that all the outcomes are measured relatively to the reference point $R_p$, which is normalized to $x=0$ in the figure. \textcolor[rgb]{0,0,0}{Behavioral} studies show that the function $v(x)$ is s-shaped, which is concave in the gain region (i.e., $x>0$, when the outcome is larger than the reference point) and convex in the loss region (i.e., $x<0$, when the outcome is smaller than the reference point). Moreover, the impact of loss is larger than the gain, i.e., $|v(-x)|>v(x)$ for any $x>0$.\footnote{\aug{To better understand the s-shaped value function, consider the following two lottery settings. Lottery A1: 50\% to win \$200, and 50\% to win \$0; Lottery A2: 100\% to win \$100. Experimental results \cite{kahneman_pt79,tversky_ai92} showed that most people prefer Lottery A2 to A1. The result reflects that people are risk-averse in gains (i.e., $\beta<1$). Next we further consider another two lottery settings. Lottery B1: 50\% to win \$100, and 50\% to loss \$100; Lottery B2: 100\% to win \$0. Experimental results \cite{kahneman_pt79,tversky_ai92} showed that most people prefer Lottery B2 to B1. The result reflects that people are loss averse (i.e., $\lambda>1$).}} A commonly used value function in the PT literature is \cite{kahneman_pt79}
\begin{equation} \label{equ:valuefunction}
     v(x)=\left\{
    \begin{aligned}
    &x^\beta,  &   &\text{ if }x\geq 0,\\
    &-\lambda(-x)^\beta, &   &\text{ if }x<0,\\
    \end{aligned}
    \right.\\
\end{equation}
\noindent where $0<\beta\leq1$ and $\lambda\geq1$. \augg{Here $\beta$ is the \textbf{risk parameter}, where a \emph{smaller} $\beta$ means that the value function is more concave in the gain region, hence the user is more \emph{risk-averse in gains}. Meanwhile, a \emph{smaller} $\beta$ also means that the value function is more convex in the loss region, hence the user is more \emph{risk-seeking in losses}. Under a high reference point $R_p$, the user is more likely to encounter losses, hence a smaller $\beta$ means that the user is more \emph{risk-seeking dominant}. Under a low reference point $R_p$, however, a smaller $\beta$ means that the user is more \emph{risk-averse dominant}.} The valuation of the loss region is further characterized by the \textbf{loss penalty parameter} $\lambda$, where a \emph{larger} $\lambda$ indicates that the user is more \emph{loss averse}.

We note that the value function in EUT is a special case of PT, with the parameter choices $\lambda= \beta = 1$, and the value function becomes a linear function of $v(x) = x$. In this case, the choice of reference point only leads to a constant shift of the value function without affecting the user's decision. Without loss of generality, we will choose $R_p=0$ for the EUT case.

%-------------------------------------------------
\subsubsection{Probability Distortion}
%-------------------------------------------------
Fig.~\ref{fig:valuepdfun}(b) illustrates the probability distortion function $w(p)$, which captures humans' psychological over-weighting of low probability events and under-weighting of high probability events \cite{kahneman_pt79}. A commonly used probability distortion function is \cite{prelec1991decision}
\begin{equation}\label{equ:pdfunction}
	w(p) = \exp(-(-\ln p)^\mu), \;0<\mu\leq1,
\end{equation}
\noindent where $p$ is the \emph{objective} probability of an outcome and $w(p)$ is the corresponding \emph{subjective} probability. Here $\mu$ is the \textbf{probability distortion parameter}, which reveals how a person's subjective evaluation distorts the objective probability. A \emph{smaller} $\mu$ means a \emph{larger} distortion.\footnote{\aug{To better understand PT, consider the following lottery settings. Lottery A1: 1\% to win \$99, and 99\% to loss \$1; Lottery A2: 100\% to win \$0. Experimental results \cite{kahneman_pt79,tversky_ai92} showed that most people prefer Lottery A1 to A2. The result reflects that people will have a subjective probability distortion of small probability events (i.e., $\mu<1$). }}

When $\mu=1$, we have $w(p) = p$, which refers to the case of EUT without probability distortion. 

%%%%%%%%%%%%%%%%%%%%%%%%%%%%%%%%%%%%%%%%%%%%%%%%%%%%%%
\subsection{Two-Stage Decision Problem}\label{sec:III-D}
%%%%%%%%%%%%%%%%%%%%%%%%%%%%%%%%%%%%%%%%%%%%%%%%%%%%%%
Next we derive the user's expected utilities of being a buyer and a seller, respectively, with the remaining data quota $Q$ and a probability distribution of the future data demand $d$. 

\begin{figure}	
	\centering		
	\includegraphics[width=0.48\textwidth]{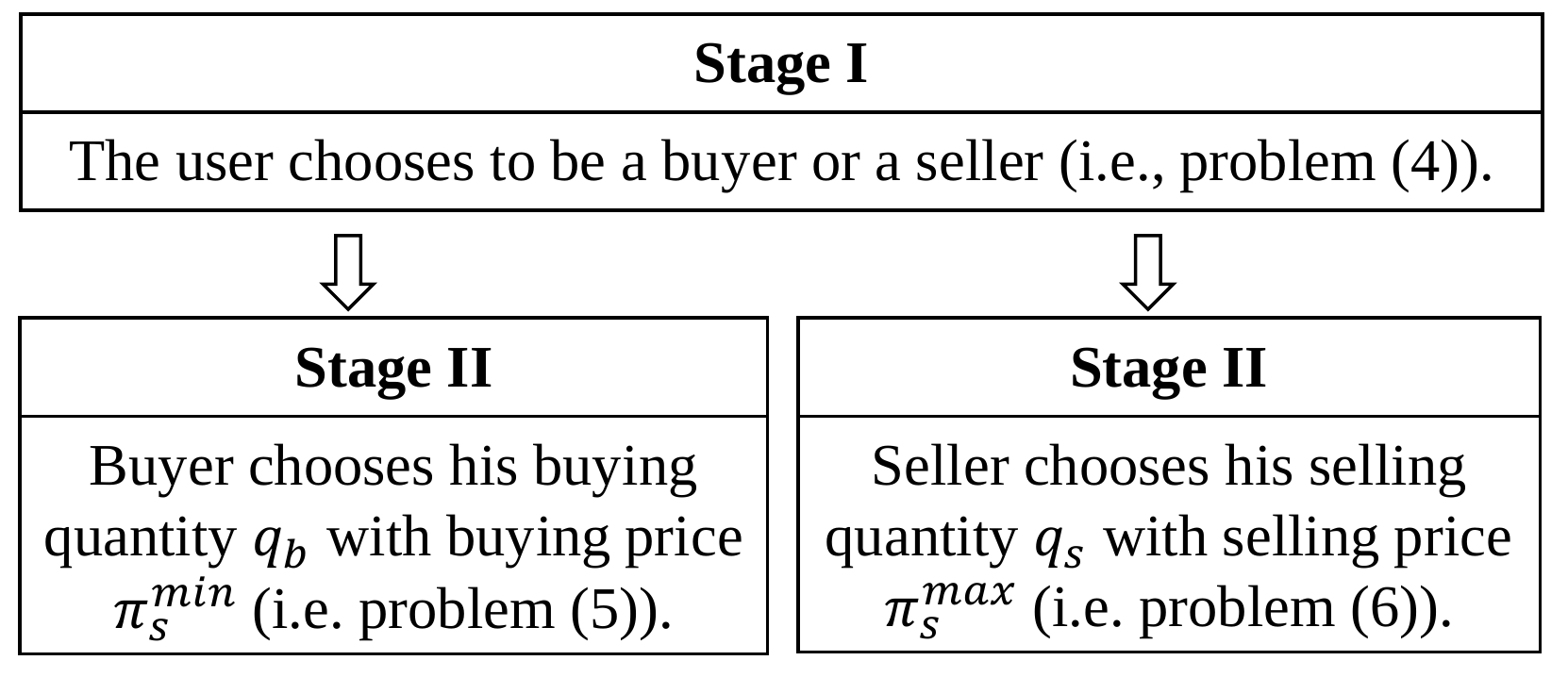}
	\caption{Two-Stage Optimization.}
	\label{fig:twostage}
\end{figure}
Fig.~\ref{fig:twostage} shows how each user makes the trading decision in two stages.\footnote{\res{The two-stage game model and the simultaneous seller-buyer decision model lead to the same result. We present it as a two-stage game for the ease of exposition.}} In Stage I, he decides whether to sell or to buy in the market. In Stage II, he decides the price and quantity as a seller or as a buyer, depending on his choice in Stage I.

%----------------------------------------------------
\subsubsection{Stage I's Problem}
%----------------------------------------------------
In Stage I, a user makes a decision $a\in \mathcal{A} = \{s,b\}$, where $s$ and $b$ correspond to being a seller and a buyer, respectively. 
We use $u(a)$ to denote the user' \emph{maximum utility} that can be achieved under the choice of $a$ (through the optimized decisions in  Stage II), as defined in (\ref{equ:stage2b}) and (\ref{equ:stage2s}). Then, the user's Stage I optimization problem is
\begin{align}
\max_{a\in\{s,b\}} \quad\quad &u(a). \label{equ:stage1}
\end{align}

%----------------------------------------------------
\subsubsection{Stage II's Problem}
%----------------------------------------------------
%In Stage II, a seller determines his offer $\{q_s,\pi_s\}$, which means that he is willing to sell $q_s$ (GBs) of data at a unit price of $\pi_s$ (dollars per GB). A buyer determines his bid $\{q_b,\pi_b\}$, which means that he is willing to buy a total of $q_b$ of data at a unit price of $\pi_b$. As we have discussed in Section III-A, a buyer will choose $\pi_b=\pi_s^{\min}$ and a seller will choose $\pi_s=\pi_b^{\max}$ if he wants to complete the trade.

A buyer in Stage II needs to decide his buying quantity $q_b$, given the minimum selling price $\pi_s^{\min}$ as discussed in Section \ref{sec:III-A}. Thus, the buyer's problem is to maximize his expected utility:\footnote{For all the optimization problems discussed in this paper, we will consider the three features of PT as discussed in Section \ref{sec:III-C}, and EUT is a special case under proper parameter choices. We will not repeat this point later on.}
\begin{align}
u(b)=~\max_{q_b\geq0} ~ U(b,q_b)&=\sum_{i=1}^I w(p_{i})v(-\pi_s^{\min} q_b\notag\\&+L(Q+q_b-d_{i})-R_p),\label{equ:stage2b}
\end{align}
\noindent where $\pi_s^{\min} q_b$ is the cost for buying the data at the price $\pi_s^{\min}$, and $L(Q+q_b-d_{i})$ is the satisfaction loss after trading if the future data demand is $d_i$. 

\aug{On the other hand, a seller in Stage II needs to decide his selling quantity $q_s$, given the maximum buying price $\pi_b^{\max}$:}
\begin{align}
u(s)~=~\max_{q_s\geq 0} ~ U(s,q_s)&=\sum_{i=1}^I w(p_{i})v(\pi_b^{\max} q_s\notag\\
&+L(Q-q_s-d_{i})-R_p),\label{equ:stage2s}
\end{align}
\noindent where $\pi_b^{\max} q_s$ is the revenue obtained from selling the data at the price $\pi_b^{\max}$, and $L(Q-q_s-d_{i})$ is the satisfaction loss after trading if the future data demand is $d_i$.

In the next section, we will solve the user's two-stage optimal trading problems (\ref{equ:stage1}), (\ref{equ:stage2b}), and (\ref{equ:stage2s}) by backward induction.

%For buyer $j$, he plays this two-stage game by maximizing his utility function with the constraint that the total quantity to buy is less than total quantity to sold in this market. For seller $i$, he plays this two-stage game by maximizing his utility function with the constraint that his selling quantity is less than his data usage that entitled from monthly service.
%
%Hence, for the $j$-th buyer, he plays the two-stage game by solving an optimization problem as follows:
%\begin{align}
%\max_{s_j}\quad \quad &U_j(s_j)\\\notag
%s.t. \quad \quad &0\leq \sum_j s_j\leq \sum_i s_i\\\notag
%\end{align}
%and for the $i$-th seller, he plays the two-stage game by solving Sn optimization problem as follows:
%\begin{align}
%\max_{s_i,\pi}\quad \quad &U_i(s_i)\\\notag
%s.t. \quad \quad &0\leq s_i\leq s_i^m\\\notag
%&\pi\geq0\\\notag
%\end{align}

\section{Solving The Two-stage Optimization Problem}\label{sec:solve}
In this section, we first derive the user's optimal selling or buying decision in Stage II. Then, we consider whether the user chooses to be a seller or a buyer in Stage I by comparing his maximum achievable utilities under both cases.

Problems (\ref{equ:stage2b}) and (\ref{equ:stage2s}) are challenging analytically due to the non-convexity of the s-shaped value function $v(x)$, especially under an arbitrary reference point. To obtain clear engineering insights, we focus on two choices of reference points in the following analysis: 
\begin{itemize}
\item \emph{High reference point} $R_p=0$: \aug{It reflects the user's expectation of \aug{observing} the lowest possible demand level $d_1$ hence having no excessive demand. }
\item \emph{Low reference point} $R_p=\kappa(Q-d_I)<0$: It reflects the user's expectation of observing the highest possible demand level $d_I$ and paying for the corresponding excessive demand (without trading).
\end{itemize} 

\augg{\res{The high reference point refers to the best case scenario without trading, while the low reference point refers to the worst case scenario without trading. Best case and worst case scenarios are widely used concepts in risk management \cite{artzner1999coherent}, and are frequently used as benchmarks for evaluating investment performances \cite{jin_bp08}.} For a particular given outcome, it is more likely to be considered as a gain under $R_p = \kappa(Q-d_I)$ than under $R_p=0$. }

%We find the global optimal solution by partitioning the whole interval into several sub-intervals, and exploiting the convex or the special \emph{unimodal} structure in each sub-interval. 
\aug{To get around the non-convexity issue of problems (\ref{equ:stage2b}) and (\ref{equ:stage2s}), we partition the whole feasible range of the decision variable into several sub-intervals based on the piece-wise linearity of the satisfaction loss function $L(y)$ in (\ref{eq:ly}), such that the objective function in each sub-interval is either convex or unimodal. We then compute the unique optimal solution by confining the problem to each sub-interval, and finally identify the global optimum by comparing the optimal objective function values of all sub-intervals.}
%We will show the detailed method in the following two subsections. 

\aug{In order to understand the impact of the risk parameters on the optimal trading decisions, we further consider a special case with binary possible demand $I=2$, in which case we are able to characterize the user's optimal decision in closed-form.}

%%%%%%%%%%%%%%%%%%%%%%%%%%%%%%%%%%%%%%%%%%%%%%%%%%%%%%%%%%%%%%%%%%%%%%%%%%%%%
%===========================================
\subsection{Stage II: Solving Buyer's Problem (\ref{equ:stage2b})}
%===========================================
%%%%%%%%%%%%%%%%%%%%%%%%%%%%%%%%%%%%%%%%%%%%%%%%%%%%%%%%%%%%%%%%%%%%%%%%%%%%%

\subsubsection{General Case of $I \geq 2$}
%The buyer's problem is shown in (\ref{equ:stage2b}). As we have mentioned before, EUT is a special case of PT under $\lambda=\beta=\mu= 1$ and $R_p=0$\footnote{Due to space limitation, we include the analysis of EUT case our technical report [ref].}. 
The way of solving problem (\ref{equ:stage2b}) will depend on the choice of reference point. \res{Under the high reference point $R_p=0$, we will partition the whole feasible range of $q_b$ into $I-\hat{\imath}+1$ sub-intervals based on $I$ possible realizations of $d_i$. We will show that $U(b,q_b)$ is convex in each sub-interval, which implies that the optimal $q_b^*$ for each sub-interval is one of the two boundary points.} \final{An example of $U(b,q_b)$ under $R_p=0$ is shown in Fig.~\ref{fig:examplesb}(a). In this example, we assume $I=3$, where $d_1<Q$ and $Q<d_2<d_3$, so that $\hat{\imath}=1$. As we can see from Fig.~\ref{fig:examplesb}(a), the feasible range of $q_b$ can be divided into three sub-intervals: $[0,d_2-Q]$, $[d_2-Q, d_3-Q]$, and $[d_3-Q,\infty)$. The function $U(b,q_b)$ is convex in each sub-interval, so that we can find the global optimal $q_b$ by comparing the function values at the boundary points of the sub-intervals (i.e., $U(b,0)$, $U(b,d_2-Q)$, and $U(b,d_3-Q)$).}

\res{Under the low reference point $R_p=\kappa(Q-d_I)$, we can show that $U(b,q_b)$ is a concave function in each sub-interval. Thus, the optimal $q_b^*$ for each sub-interval is either at one of the boundary points or at the critical point (where the first order derivative equals zero). As long as we obtain the optimal solution for each sub-interval, we can compute the globally optimal solution by comparing the $I-\hat{\imath}+1$ sub-intervals' optimal points.}

\begin{figure}[t]
\setlength{\belowcaptionskip}{-5mm}
%\begin{tabular}{cc}   
\begin{minipage}{0.48\linewidth}
  \centerline{\includegraphics[width=4.7cm]{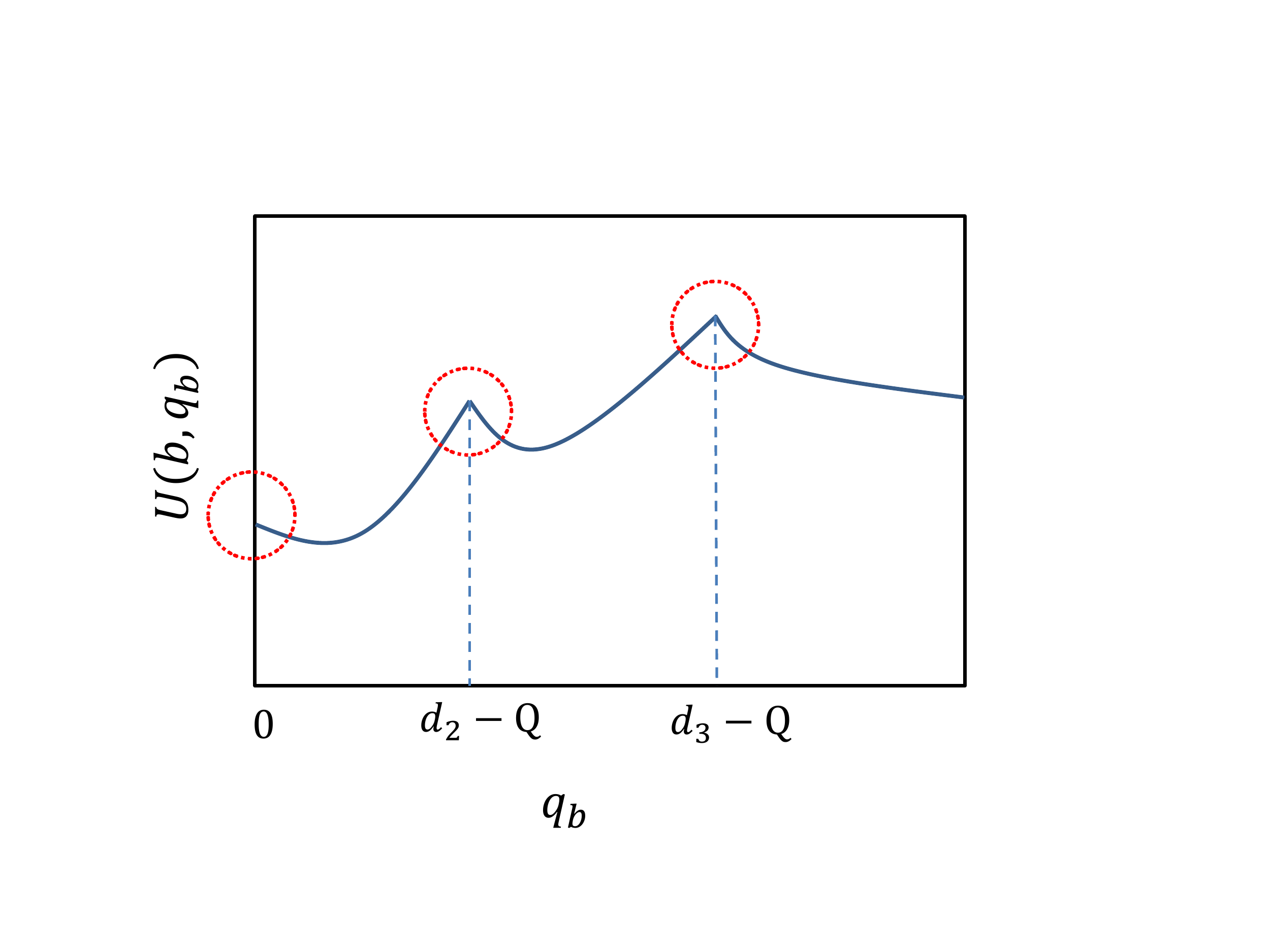}}
  \centerline{(a) $U(b,q_b)$}
\end{minipage}
\hfill
\begin{minipage}{.48\linewidth}
  \centerline{\includegraphics[width=4.7cm]{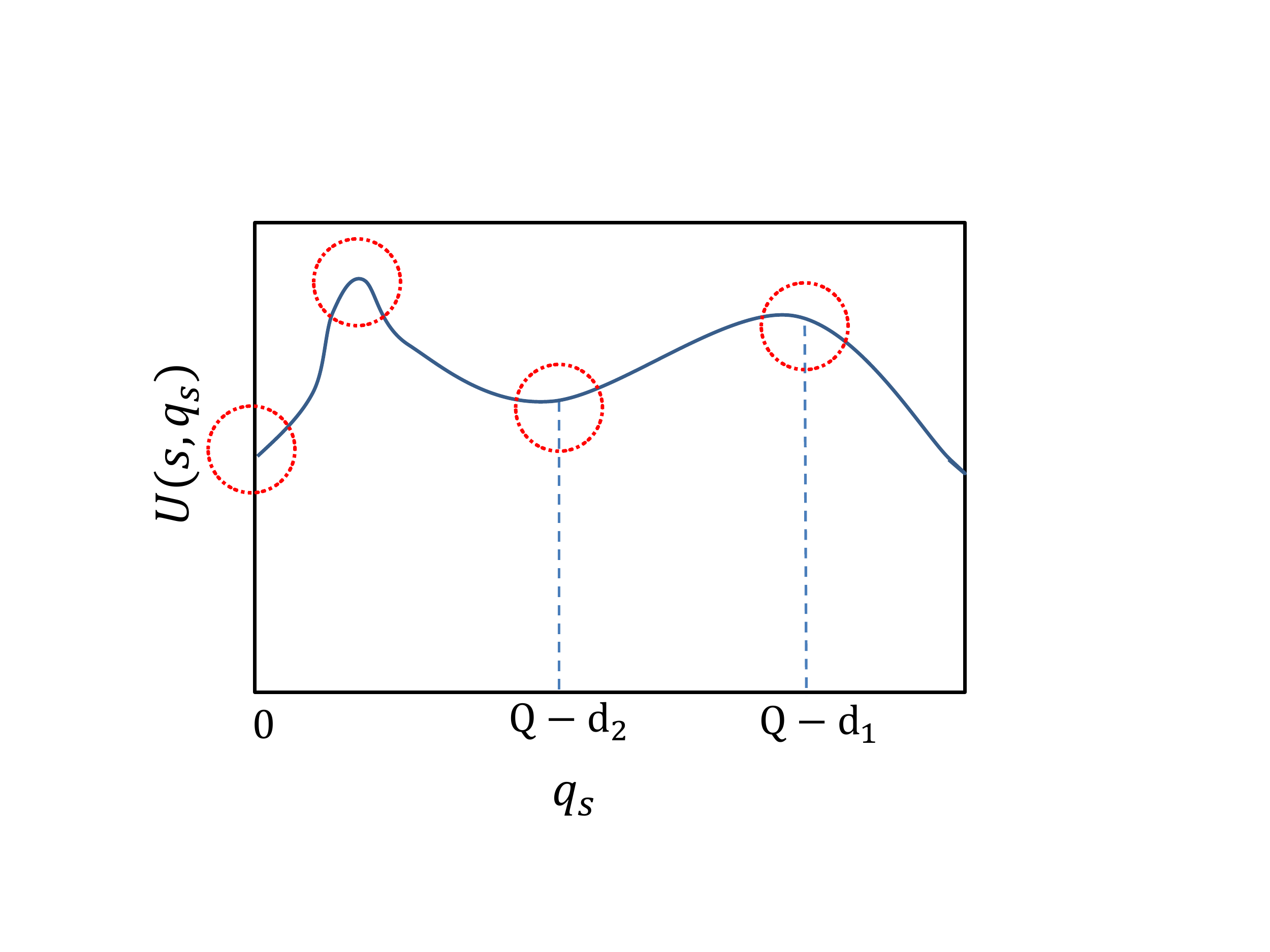}}
  \centerline{(b) $U(s,q_s)$}
\end{minipage}
\vfill
\caption{\final{Examples of (a) $U(b,q_b)$ in (\ref{equ:stage2b}) under $R_p=0$ and (b) $U(s,q_s)$ in (\ref{equ:stage2s}) under $R_p=0$.}}
\label{fig:examplesb}
\end{figure}
\auggg{Before introducing the following theorem, we first define
\begin{equation}
    \mathcal{X}_b=\{q_b:\frac{\partial U(b,q_b)}{\partial q_b}=0\text{ under } R_p=\kappa(Q-d_I)\},
\end{equation}
\noindent which is the set of critical points (the points at which the first order derivative of the user's utility equals zero). We can also prove that there are at most $I$ critical points in the whole feasible range (i.e., $|\mathcal{X}_b|\leq I$) in Appendix A.}
\begin{thm}\label{thm:buyerthm}
The buyer's optimal buying quantity by solving problem (\ref{equ:stage2b}) under the high reference point $R_p = 0$ is
\begin{equation}\label{thm:buyerh}
q_b^*=\arg\max_{q_b\in\{Q-d_i,i=\hat{\imath}+1,\ldots,I\}\cup\{0\}} \{U(b,q_b)\},
\end{equation}
\noindent and that under the low reference point $R_p=\kappa(Q-d_I)$ is
\begin{equation}\label{thm:buyerl}
q_b^*=\arg\max_{q_b\in\{Q-d_i,i=\hat{\imath}+1,\ldots,I\}\cup\mathcal{X}_b\cup\{0\}} \{U(b,q_b)\}.
\end{equation}
\end{thm}
The proof of Theorem \ref{thm:buyerthm} is given in Appendix A.

Next, we show the impact of the value function, probability distortion, and reference point in the special case of binary outcomes.

\subsubsection{Special Case of $I=2$}\label{sec:specialcase}
To better illustrate the impact of various parameters on the buyer's optimal decision, we next consider the buyer's optimization problem with $I=2$ possible demands. More specifically, there are two possible realizations of the future data demand: $d_1=d_l$ and $d_2=d_h$, with $0<d_l<Q<d_h$. The probability of observing a high demand $d_h$ is $p$, and the probability of observing low demand $d_l$ is $1-p$. 

\aug{We first define \textbf{buyer's threshold price} under different reference points. As we will show in Theorem \ref{thm:buyer2}, the optimal buying amount equals $d_h-Q$ when the minimum selling price $\pi_s^{\min}$ is below the buyer's threshold price: 
%\vspace{2mm}
\begin{align}\label{equ:pib}
&\bar{\pi}_b^{EUT} \triangleq \kappa p, \quad\bar{\pi}_b^{PTh} \triangleq \kappa\left[\frac{w(p)}{w(p)+w(1-p)}\right]^{\frac{1}{\beta}}, \notag\\
&\bar{\pi}_b^{PTl} \triangleq \frac{\kappa w(p)}{w(1-p)+w(p)}.
\end{align}

\begin{thm}\label{thm:buyer2}
The buyer's optimal buying solution by solving problem (\ref{equ:stage2b}) under EUT is 
\begin{align}
\label{eq:table_buyereut}
        q_{b}^*=\left\{
        \begin{array}{ll}
        d_{h}-Q, 
        &\text{ if }\pi_s^{\min}<\bar{\pi}_b^{EUT} ,\\
        0, 
        &\text{ if }\pi_s^{\min}\geq \bar{\pi}_b^{EUT}.
        \end{array}
        \right.
\end{align}
\noindent His optimal buying solution by solving problem (\ref{equ:stage2b}) under PT with high reference point $R_p = 0$ is
\begin{align}
\label{eq:table_buyer}
        q_{b}^*=\left\{
        \begin{array}{ll}
        d_{h}-Q, 
        &\text{ if }\pi_s^{\min}<\bar{\pi}_b^{PTh},\\
        0, 
        &\text{ if }\pi_s^{\min}\geq \bar{\pi}_b^{PTh},
        \end{array}
        \right.
\end{align}
\noindent and that with low reference point $R_p=\kappa(Q-d_h)$ is
\begin{align}
\label{eq:table_buyerrp}
        &q_{b}^*  =  \notag\\
        &\left\{    
        \begin{array}{ll}
        d_{h}-Q,
        &\text{if }\pi_s^{\min} < \bar{\pi}_b^{PTl},\\
        0, 
        &\text{if }\pi_s^{\min} \geq \bar{\pi}_b^{PTl} ~\text{and}~ \beta=1,\\
        \frac{\kappa(Q-d_h)}{\left[ \frac{w( p )(\kappa - \pi_s^{\min})^{\beta}}{w( 1 - p )\pi_s^{\min}} \right]^{\frac{1}{\beta-1}}   +\pi_s^{\min}},
        &\text{if }\pi_s^{\min} \geq \bar{\pi}_b^{PTl}~
        \text{and} ~\beta <1.
        \end{array}
        \right.
\end{align}
\end{thm}
}
Theorem \ref{thm:buyer2} is a special case of Theorem \ref{thm:buyerthm}, and the proof is given in Appendix B. The result in (\ref{eq:table_buyereut}) follows directly from (\ref{eq:table_buyer}) and (\ref{eq:table_buyerrp}) by setting $\beta = \mu = \lambda = 1$.

In (\ref{eq:table_buyereut}) and (\ref{eq:table_buyer}), we observe that the optimal buying quantity is discontinuous at the buyer's threshold price. This is due to the linearity of utility function in the EUT case and the convexity of utility function in the PT case with $R_p=0$.  Details are given in Appendix B.
%The tables also shows that when the minimum selling price is lower than a threshold price, the buyer will buy an amount equal to the difference between the high demand and we define the three threshold prices as
%
%\begin{defi}
%The buyer will buy an amount equal to the difference between the high demand when the minimum selling price is lower than the threshold prices $ \bar{\pi}_b^{EUT} $, $ \bar{\pi}_b^{PTh} $ and $  \bar{\pi}_b^{PTl} $ in EUT case, PT case with $R_p=0$ and PT case with $R_p=\kappa(Q-d_h)$, respectively.
%\end{defi}
 
\aug{From Theorem \ref{thm:buyer2}, we have the following observations on the impact of reference point when we fix the probability distortion parameter $\mu=1$ (hence removing the impact of probability distortion).}

\begin{obs}
 (PT vs EUT under the high reference point) When $\mu=1$ and $R_p=0$, we have $  \bar{\pi}_b^{PTh} < \bar{\pi}_b^{EUT}$. This means that under a high reference point, a PT buyer is less willing to purchase mobile data than an EUT buyer.
\end{obs}
\begin{obs} 
 (PT vs EUT under the low reference point) When $\mu=1$ and $R_p=\kappa(Q-d_h)$, we have $  \bar{\pi}_b^{PTl} = \bar{\pi}_b^{EUT}$. However, the optimal buying quantity $q_b^*$ of the PT buyer in (\ref{eq:table_buyerrp}) is no smaller than that of the EUT buyer in (\ref{eq:table_buyereut}) under the same price $\pi_s^{\min}$. This means that under a low reference point, a PT buyer is more willing to purchase mobile data than an EUT buyer.
\end{obs}
\textbf{Notice that buying data reduces the risk that the future data demand exceeds the quota.} \aug{As we have mentioned in Section III-C, a smaller $\beta$ means that the buyer is more \emph{risk-seeking in losses} and more \emph{risk-averse in gains}. Under a high expectation (e.g., $R_p=0$), the buyer with a smaller $\beta$ (in the PT case) is more risk-seeking dominant and will not buy data. Under a low expectation (e.g., $R_p=\kappa(Q-d_h)<0$), the buyer is more risk-averse dominant, and will buy an amount equal to $d_h - Q$, which will completely eliminate the risk that the future data demand exceeds the updated quota $d_h$.  }

%%%%%%%%%%%%%%%%%%%%%%%%%%%%%%%%%%%%%%%%%%%%%%%%%%%%%%%%%%%%%%%%%%%%%%%%%%%%%
%===========================================
\subsection{Stage II: Solving Seller's Problem (\ref{equ:stage2s})}
%===========================================
%%%%%%%%%%%%%%%%%%%%%%%%%%%%%%%%%%%%%%%%%%%%%%%%%%%%%%%%%%%%%%%%%%%%%%%%%%%%%
\subsubsection{General Case of $I \geq 2$}
%The seller's problem is shown in (\ref{equ:stage2s}). As we have mentioned, EUT is a special case of PT under $\lambda=\beta=\mu= 1$ and $R_p=0$\footnote{Due to space limitation, we include the analysis of EUT case our technical report [ref].}. 
To solve problem (\ref{equ:stage2s}) under both $R_p = 0$ and $R_p = \kappa(Q-d_I)$, we partition the whole interval of $q_s$ into $\hat{\imath}+1$ sub-intervals. We show that $U(s,q_s)$ has a special \emph{unimodal} structure in each sub-interval. Since the first
order derivative of a unimodal function will cross zero at most once in each sub-interval, and thus the optimal $q_s^*$ for each sub-interval is either at one of the boundary points or at the critical point (where the first order derivative equals zero). Then, by comparing the $\hat{\imath}+1$ optimal points, we can find the global optimal solution. \final{An example of $U(s,q_s)$ under $R_p=0$ is shown in Fig.~\ref{fig:examplesb}(b). In this example, we assume $I=3$, where $d_1<d_2<Q$ and $Q<d_3$, so that $\hat{\imath}=2$. As we can see from Fig.~\ref{fig:examplesb}(b), the feasible range of $q_s$ can be divided into three sub-intervals: $[0,Q-d_2]$, $[Q-d_2, Q-d_1]$, and $[Q-d_1,\infty)$. The function $U(s,q_s)$ is unimodal in each sub-interval, so that we can find the global optimal $q_s^*$ by comparing the boundary function values of the sub-intervals  (i.e., $U(s,0)$, $U(s,Q-d_2)$, and $U(s,Q-d_1)$) and the function values of critical points if they exist (i.e., $U(s,q_s), q_s\in\{q_s:\frac{\partial U(s,q_s)}{\partial q_s}=0\text{ under }  R_p=0\}$).}
%Before introducing the following theorem, we first define $g_j(q_s)$ (see in [ref] in Appendix) as the indicator functions, and the optimal $q_s^*$ is based on the indicators.

\auggg{Before introducing the following theorem, we first define
\begin{align}
 &\mathcal{X}_{sh} = \{q_s: \frac{\partial U(s,q_s)}{\partial q_s}=0\text{ under }  R_p=0\} \notag\\
 &\text{and   }\notag\\
 &\mathcal{X}_{sl} = \{q_s: \frac{\partial U(s,q_s)}{\partial q_s}=0\text{ under }  R_p=\kappa(Q-d_I)\},
\end{align}
\noindent which are the sets of critical points (the points at which the first order derivatives of the user's utility equal zero). We can also prove that there are at most $I$ critical points in the whole feasible range under both $R_p=0$ and $R_p=\kappa(Q-d_I)$ (i.e., $|\mathcal{X}_{sh}|\leq I$ and $|\mathcal{X}_{sh}|\leq I$) in Appendix C.}

\begin{thm}\label{thm:sellerthm}
The seller's optimal selling quantity $q_s^*$ of problem (\ref{equ:stage2s}) under PT with the high reference point $R_p = 0$ is
\begin{equation} \label{thm:sellerh}
q_s^*=\arg\max_{q_s\in\{Q-d_i,i=1,\ldots,\hat{\imath}\}\cup\mathcal{X}_{sh}\cup\{0\}} \{U(s,q_s)\},
\end{equation}
\noindent and that with the low reference point $R_p = \kappa(Q-d_I)$ is
\begin{equation}\label{thm:sellerl}
q_s^*=\arg\max_{q_s\in\{Q-d_i,i=1,\ldots,\hat{\imath}\}\cup\mathcal{X}_{sl}\cup\{0\}} \{U(s,q_s)\}.
\end{equation}
\end{thm}

The proof of Theorem \ref{thm:sellerthm} is given in Appendix B. Next, we show the impact of the value function, probability distortion, and reference point for the special case of binary outcomes.

\subsubsection{Special Case of $I=2$}
To better illustrate the insights, we next consider the seller's optimization problem with $I = 2$ possible demands. 

We first define \textbf{seller's threshold price} under different risk preferences. As we will show in Theorem \ref{thm:seller2}, the optimal selling amount equals $Q-d_l$ when the maximum buying price $\pi_b^{\max}$ is above the seller's threshold price. The seller's threshold prices $\bar{\pi}_s^{EUT}$, ${\bar{\pi}_s^{PTh}}$, and $\bar{\pi}_s^{PTl}$ are the unique\footnote{The proof of the uniqueness is in Appendix D.} solutions of the following three equations:
\begin{equation}\label{equ:piseut}
\bar{\pi}_s^{EUT} = \kappa p ,
\end{equation}
\begin{equation}\label{equ:pispth}
\frac{\lambda(\kappa -\bar{\pi}_s^{PTh} )^\beta w(p)}{({\bar{\pi}_s^{PTh}})^\beta w( 1 - p )} \left(  1 + \frac{\kappa (d_{h} - Q)}{( \kappa  - \bar{\pi}_s^{PTh} )(Q - d_{l})} \right)^{\beta - 1}    = 1 \\,
\end{equation}
\begin{align}\label{equ:pisptl}
&w(1 - p)\{[(\bar{\pi}_s^{PTl} - \kappa)Q+\kappa d_h-\bar{\pi}_s^{PTl} d_l]^{\beta} - [\kappa(d_h - Q)]^{\beta}\}\notag\\
&=\lambda w(p)[(\kappa-\bar{\pi}_s^{PTl})(Q-d_l)]^{\beta}.
\end{align}
\begin{thm}\label{thm:seller2}
The seller's optimal selling quantity in problem (\ref{equ:stage2s}) under EUT is 
\begin{align}
\label{eq:table_sellereut}
        q_s^*=\left\{
        \begin{array}{ll}
        Q-d_l, 
        &\text{ if }\pi_b^{\max}>\bar{\pi}_s^{EUT},\\
        0, 
        &\text{ if }\pi_b^{\max}\leq \bar{\pi}_s^{EUT}.
        \end{array}
        \right.
\end{align}
\noindent His optimal selling quantity in problem (\ref{equ:stage2s}) under PT with the high reference point $R_p = 0$ is
\begin{align}
\label{eq:table_seller}
        &q_s^*  =  \notag\\
        &\left\{   
        \begin{array}{ll}
        Q-d_l, 
        &\text{if }\pi_b^{\max} > \bar{\pi}_s^{PTh},\\
        0, 
        &\text{if }\pi_b^{\max} \leq \bar{\pi}_s^{PTh} ~\text{and} ~\beta =1,\\
         \frac{\frac{\kappa }{\kappa -\pi_b^{\max}}\left(d_{h} - Q\right)}{\left(  \frac{w(1 - p){\pi_b^{\max}}^{\beta}}{w(p)\lambda(\kappa  - \pi_b^{\max})^{\beta}}  \right)^{\frac{1}{\beta-1}}   - 1}, 
        &\text{if }\pi_b^{\max} \leq \bar{\pi}_s^{PTh}~ \text{and}~\beta<1,
        \end{array}
        \right.
\end{align}
\noindent and that with the low reference point $R_p=\kappa(Q-d_h)$ is
\begin{align}
\label{eq:table_sellerrp}
        q_s^* = \left\{
        \begin{array}{ll}
         Q - d_l, 
        &\text{if }\pi_b^{\max}>\bar{\pi}_s^{PTh},\\
         0, 
        &\text{if }\pi_b^{\max}\leq\bar{\pi}_s^{PTh}.
        \end{array}
        \right.
\end{align}
\end{thm}
Theorem \ref{thm:seller2} is a special case of Theorem \ref{thm:sellerthm}, and the proof of Theorem \ref{thm:seller2} is given in Appendix D. The result in (\ref{eq:table_sellereut}) follows directly from (\ref{eq:table_seller}) and (\ref{eq:table_sellerrp}) by setting $\beta=\mu=\lambda=1$.

In (\ref{eq:table_sellereut}) and (\ref{eq:table_sellerrp}), we observe that the optimal selling quantity $q_s^*$ is discontinuous at the seller's threshold price. This is due to the linearity of utility function in the EUT case and the unimodality of utility function in the PT case with $R_p=\kappa(Q-d_h)$. Details are given in Appendix D. 

\aug{From Theorem \ref{thm:seller2}, we have the following observations on the impact of reference point when we fix the probability distortion parameter $\mu=1$.}

\begin{obs}
(PT vs EUT under the high reference point) When $\mu=1$ and $R_p=0$, we have $  \bar{\pi}_s^{PTh} < \bar{\pi}_s^{EUT}$. This means that under a high reference point, a PT seller is more willing to sell mobile data than an EUT seller.

\end{obs}
\begin{obs}
(PT vs EUT under the low reference point) When $\mu=1$ and $R_p=\kappa(Q-d_h)$, we have $\bar{\pi}_s^{PTl} > \bar{\pi}_s^{EUT}$. This means that under a low reference point, a PT seller is less willing to sell mobile data than an EUT seller.
\end{obs}
\aug{\textbf{Contrary to buying data, selling data increases the risk that the future data demand exceeds the quota}. Under a high expectation (e.g., $R_p=0$), the seller with a smaller $\beta$ is more \emph{risk-seeking dominant} and will sell a large amount ($Q-d_l$). Under a low expectation (e.g., $R_p=\kappa(Q-d_h)<0$), the seller with a smaller $\beta$ is more \emph{risk-averse dominant} and will not sell data.  }

\final{In Stage I, the user decides whether to be a seller or a buyer by comparing the maximum utilities that he can achieve in both cases.}

\section{\res{Implementation of Mobile Data Trading}}\label{sec:implementation}
\rev{Building upon our theoretical analysis in Sections III and IV, here we consider several issues related to the practical implementations. We first discuss the mobile data trading algorithm that allows a user to trade multiple times during a billing cycle to adjust his trading decision in Section \ref{sec:alg2}.} Then we introduce a practical algorithm to estimate the user's risk preferences in Section \ref{sec:alg1}. 

\begin{figure}
    \centering
    \includegraphics[width=0.48\textwidth]{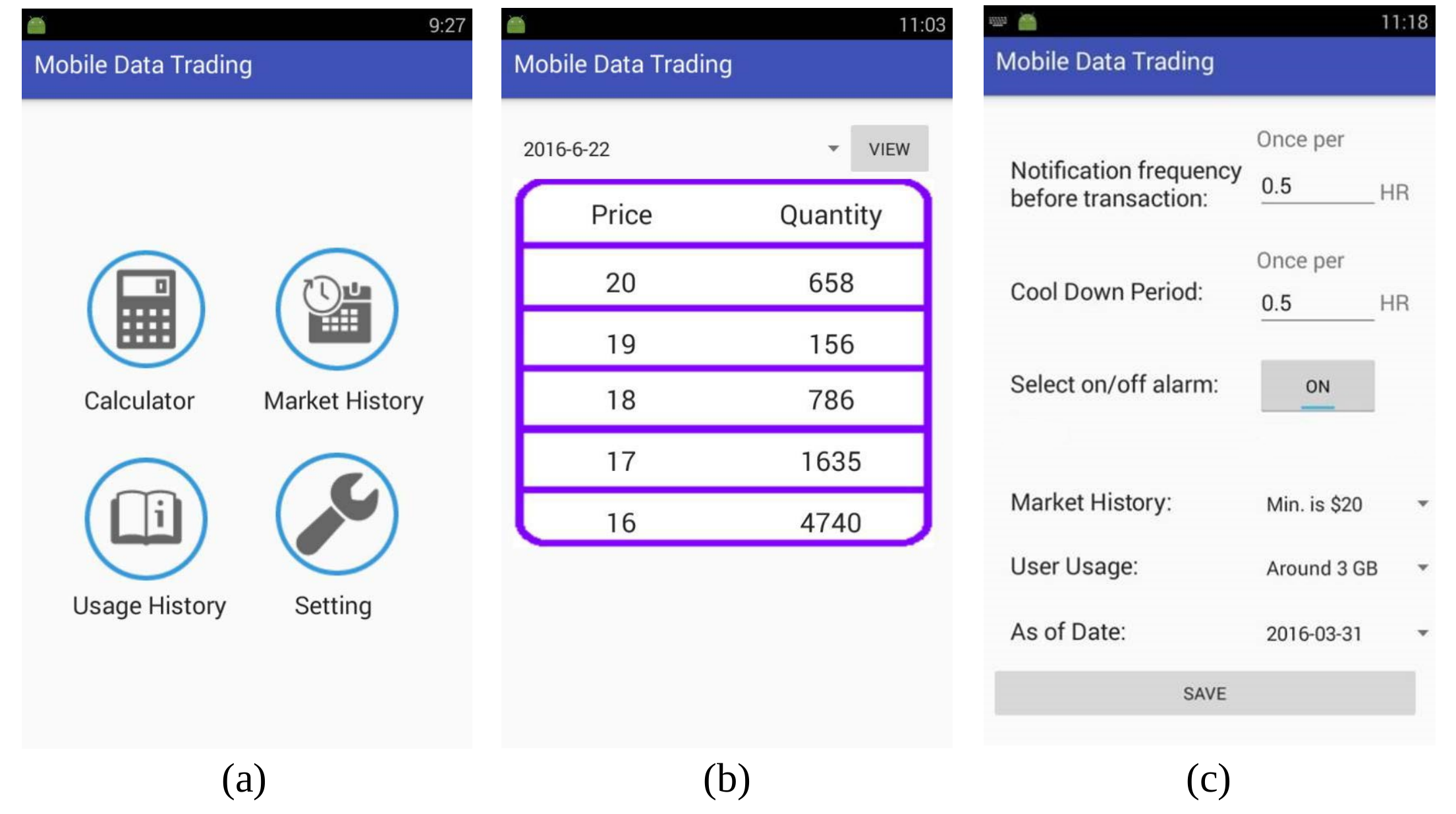}
   				\caption{Screenshots of the app: (a) Homepage, (b) Market Information, and (c) Settings.}
    \label{fig:appfigure1}
\end{figure}
Since a user's prediction of his future data demand may not be accurate, the user may want to make multiple data trading decisions as the time passes by. \res{Hence, we design a mobile data trading algorithm to facilitate a user to make smart decisions over time in a semi-automatic fashion, which reduces the user's need of frequently checking the market prices and estimating the future data demand. The mobile data trading algorithm relies on Algorithm 1 (discussed in Section V-B) to estimate the user's risk preferences, and can provide trading suggestions at any time based on the current market price, the user's current usage, and his risk preferences. Our algorithm is implemented as an Android app, the interface of which is shown in Fig.~\ref{fig:appfigure1}.} Fig.~\ref{fig:appfigure1}(a) shows the homepage screen of the Android app, which involves four areas: calculator, market history, usage history, and setting. Fig.~\ref{fig:appfigure1}(b) shows the current CMHK market information, which includes the selling prices and quantities.  Fig.~\ref{fig:appfigure1}(c) shows various system parameters that can be changed by the user, such as the trading notification frequency.\footnote{\aug{The user may not want to be disturbed by frequent notifications. He can adjust this by either turning off the notification alarm, or reduce the notification frequency to a low level, e.g., once per 24 hour.}}
\subsection{\res{Mobile Data Trading Algorithm Design}}\label{sec:alg2}

\begin{figure}
    \centering
    \includegraphics[width=0.48\textwidth]{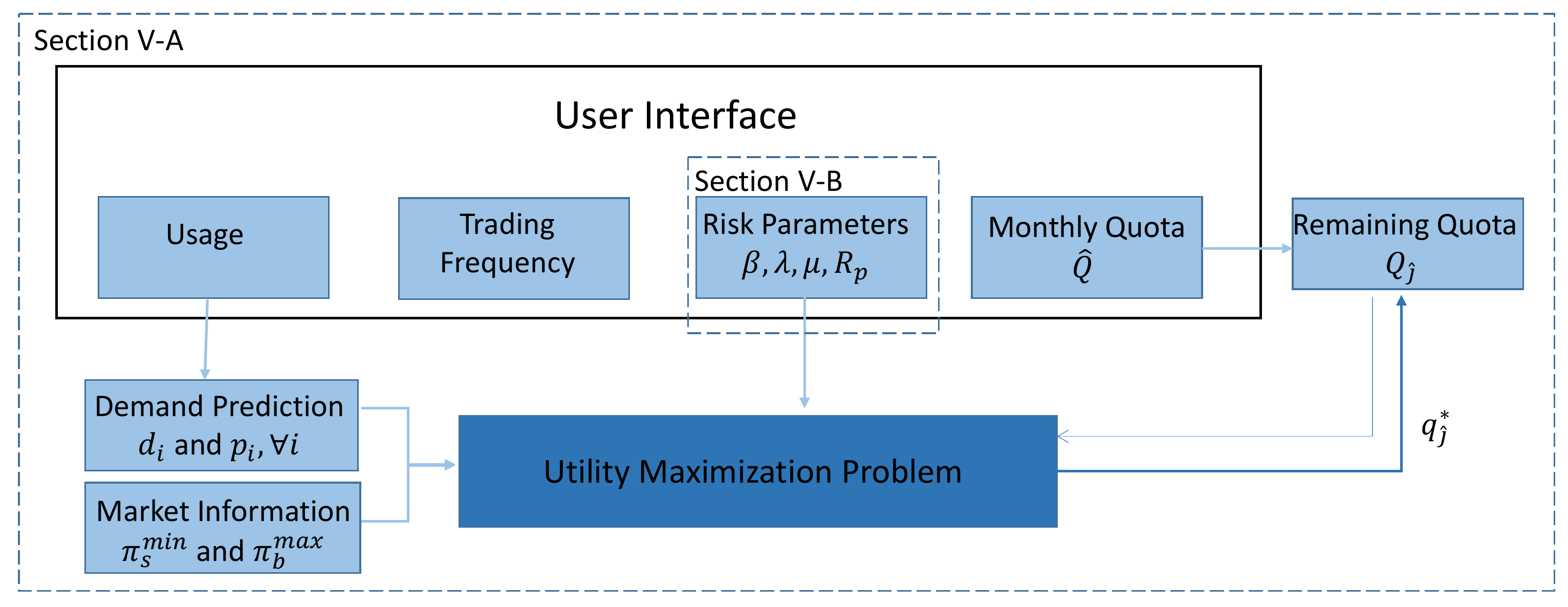}
    \caption{\res{Flowchart of the algorithm (Android app).}}
    \label{fig:appfigure}
\end{figure}

Fig.~\ref{fig:appfigure} illustrates the key function modules of the mobile app: 
\begin{itemize}
    \item Market Information: The app retrieves the CMHK mobile data trading market information, in order to determine the minimum selling price $\pi_s^{\min}$ and the maximum buying price $\pi_b^{\max}$.\footnote{Recall from Section \ref{sec:con} that there is no buyer's market in the actual CMHK platform, so $\pi_b^{\max}$ in problem (\ref{equ:stage2s}) is not well defined. To address this issue, we note that although different sellers can set different prices in the CMHK market, the system will always try to satisfy the buyers' demands with the lowest selling price. Based on the fact that the selling quantity at the minimum selling price is often very large (e.g., 4740 GB in Fig.~\ref{fig:appfigure1}(b) on June 22, 2016), the seller is not able to sell his data at a price higher than the minimum selling price, so we can assume that the maximum buying price is the same as minimum selling price, i.e., $\pi_b^{\max}=\pi_s^{\min}$.} 
    \item Trading Frequency: We assume that the app will make $T$ trading decision in a billing cycle\footnote{By default setting, the app will send every trading suggestion as a notification. The user can change the notification frequency as shown in Fig.~\ref{fig:appfigure1}(c).}. In the following discussions, for the ease of exposition we assume the trading frequency is once a day, i.e., the user makes $T=30$ trading decisions during a monthly billing cycle.\footnote{The optimal trading decision may be not to sell or buy any data, i.e., skipping some of the trading opportunities.}
    \item Usage: The app records the user's usage everyday. We denote the \emph{actual} data usage of day $\hat{\jmath}$ in month $\hat{m}$ as $\delta_{m,j}$, where month $\hat{m}$ has $T_{\hat{m}}$ days.
    \item Demand Prediction: We use an adaptive model for the user's future data demand prediction. Specifically, assuming that we are on day $\hat{\jmath}$ of month $\hat{m}$, we aim to estimate the distribution of the future data demand of the remaining time (i.e., from day $\hat{\jmath}$ to day $T_{\hat{m}}$) of month $\hat{m}$ by considering the previous $I$ month's (denoted as month $\hat{m}-1,...,\hat{m}-I$) data usage during the same time period (i.e., from day $\hat{\jmath}+1$ to day $T_{\hat{m}}$). 
    The \emph{predicted} data usage of the last $T_{\hat{m}}-\hat{\jmath}$ days in month $\hat{m}$ for $i\in \mathcal{I}$ is
        \begin{align}
        d_i:=\sum\limits_{j=\hat{\jmath}+1}^{T_{\hat{m}-i}} \delta_{\hat{m}-i,j}.\label{eq:24yu}
    \end{align}
    We will use them to predict the future data demand of the rest of month $\hat{m}$ with an equal probability. That is, $p_i = 1/I$ for $i\in\mathcal{I}$.  
    \item Quota: The remaining quota from day $\hat{\jmath}$ to the end of month $\hat{m}$ is $Q_{\hat{\jmath}}$, which corresponds to $Q$ in Sections III and IV. The value of $Q_{\hat{\jmath}}$ is an input to the utility maximization problem in (\ref{equ:stage2b}) and (\ref{equ:stage2s}), which is updated every day as follows:
        \begin{align}
        Q_{\hat{\jmath}}:=\left\{
        \begin{array}{ll}
        Q_{\hat{\jmath}-1}+q^*_{\hat{\jmath}-1}-\delta_{\hat{m},\hat{\jmath}-1}, 
        &\text{ if }\hat{\jmath}\geq 2,\\
        \hat{Q}, 
        &\text{ if }\hat{\jmath}=1.
        \end{array}
        \right.\label{eq:25wei}
    \end{align}
    Here $q^*_{{\hat{\jmath}}-1}$ is the trading quantity on day ${\hat{\jmath}}-1$ that we will discuss below, and it can be zero if no trading happens on that day.  Thus, the first line of (\ref{eq:25wei}) means that the quota is updated based on the trading quantity $q_{\hat{j}-1}^*$ and usage $\delta_{\hat{m},\hat{\jmath}-1}$, while the second line means the initialization of the month quota to $\hat{Q}$ on the first day of the month.
    \item Risk Parameters: The risk parameters include the value function parameters $\beta$ and $\lambda$ in (\ref{equ:valuefunction}), the probability distortion parameter $\mu$ in (\ref{equ:pdfunction}), and the reference point $R_p$ in (\ref{equ:stage2b}) and (\ref{equ:stage2s}).
     \item Utility Maximization Problem: The app solves problem (\ref{equ:stage1}), which involves solving (\ref{equ:stage2b}) and (\ref{equ:stage2s}), based on the market information ($\pi_b^{\max}$ and $\pi_s^{\min}$), the user's current quota $Q_{\hat{\jmath}}$, the risk parameters, usage, and future data demand prediction in (\ref{eq:24yu}). The output of the utility maximization problem on day $\hat{\jmath}$ is the optimal buying or selling quantity $q_{\hat{\jmath}}^*$, which in turn will update the quota as in (\ref{eq:25wei}). Note that a positive $q_j^*$ means the optimal buying quantity (i.e., output of (\ref{equ:stage2b})), while a negative $q_j^*$ means the optimal selling quantity (i.e., the output of (\ref{equ:stage2s}) multiplied by $(-1)$).
\end{itemize}

 \begin{algorithm}[t]
\caption{Estimation of Value Function Parameters $\lambda$ and $\beta$}
%\KwIn{$r_i$, $Backgrd(T_i)$=${T_1,T_2,\ldots ,T_n}$ and similarity threshold $\theta_r$}
%\KwOut{$con(r_i)$}
%$j:= 1$; $//$ The current day is denoted by $j$.\\
%$Q_1=Q$\;
%Calculate $q_1$ by solving (5)\;
\textbf{Input}: Quota ($Q_{\hat{\jmath}}$), risk parameters ($\mu$, $R_p$), usage($\delta_{\hat{m}-i,j},i=0,\ldots,I,j=1,\ldots,T_{\hat{m}-i}$), market information ($\pi_s^{\min}$, $\pi_b^{\max}$).\\
\For{$\hat{\jmath}=1$ \text{to} $T_{\hat{m}}$} 
{
\For{$i=1$ \text{to} $I$}{
$d_{i}:= \sum\limits_{j=\hat{\jmath}+1}^{T_{\hat{m}-i}} \delta_{\hat{m}-i,j}$ and $p_i:=1/I$\\
}
$d^{\min}:=\arg\min_{i\in\mathcal{I}}d_{i}$ and $d^{\max}:=\arg\max_{i\in\mathcal{I}}d_{i}$\\
\res{Users input the indifference prices ($\pi_{ind}^b$ and $\pi_{ind}^s$)\\
Substitute $d^{\min}$, $d^{\max}$, $\pi_{ind}^b$, and $\pi_{ind}^s$ into (25) and (26), and solve them for $\lambda$ and $\beta$\\}
Set $\lambda_{\hat{\jmath}} := \lambda$ and $\beta_{\hat{\jmath}} := \beta$\\
Update $Q_{\hat{\jmath}}$ according to (24)\\
}
\textbf{Output}: $\lambda:=\frac{\sum_{\hat{\jmath}=1}^{T_{\hat{m}}}\lambda_{\hat{\jmath}}}{T_{\hat{m}}}$ and $\beta:=\frac{\sum_{\hat{\jmath}=1}^{T_{\hat{m}}} \beta_{\hat{\jmath}}}{T_{\hat{m}}}$\\
%return $con(r_i)$\;
\end{algorithm}

The detailed algorithm for computing the data trading decisions with user's specific risk preferences is shown in Appendix E.

\subsection{Risk Parameter Estimation}\label{sec:alg1}
Since the trading decision is user-dependent, we need to estimate each user's specific risk preferences.
%As mentioned in Section \ref{sec:III-C}, in PT we have value function and probability distortion function in (\ref{equ:valuefunction}) and (\ref{equ:pdfunction}) that will affect user's risk preferences. 
In particular, we want to estimate the user's value function parameters $\lambda$ and $\beta$ in (\ref{equ:valuefunction}), which are problem-specific.\footnote{Here, we assume that the probability distortion parameter $\mu$ and the reference point $R_p$ are known. Note that probability distortion reflects the user's weighting effect on small and large probabilities, which does not depend on a specific problem. Hence, we will assume that the user already knows his probability distortion parameter $\mu$ through some other applications, surveys, or his previous investment decisions, e.g., the surveys in \cite{kahneman_pt79,tversky_ai92}. We can also calculate parameter $\mu$ through some surveys related to lotteries. For example, ``Which Lottery do you prefer? A: 100\% to win \$2; B: 99\% to win \$0, and 1\% to win \$100''. However, there is very few literature on estimating user's reference point, because a varying reference point will make the indifference equations in (\ref{eq:ind1}) and (\ref{eq:ind2}) complicated and unsolvable. In our paper, we have shown the impact of $R_p$ in our analysis on two different values in Section IV. For simplicity, we assume the high reference point $R_p=0$ in the app design and simulations in Sections V and VI-B.} For example, the parameters in making financial investments and enjoying entertainment may be quite different even for the same user.

Algorithm 1 presents the pseudo code of our algorithm to estimate the user's value function parameters $\lambda$ and $\beta$.
The basic idea is to solve the two indifference equations below \cite{kahneman_pt79,tversky_ai92}  for $\lambda$ and $\beta$ in the value function in (\ref{equ:valuefunction}).
%\begin{spacing}{1}
\begin{align}
&\sum_{i=1}^I w(p_{i})v\left(-\pi_{ind}^b (d_{I}-Q_{\hat{\jmath}})+L(d_{I}-d_{i})\right)\notag\\
&=\sum_{i=1}^I w(p_{i})v\left(L(Q_{\hat{\jmath}}-d_{i})\right), \text{ and } \label{eq:ind1} \\
&\sum_{i=1}^I w(p_{i})v\left(\pi_{ind}^s (Q_{\hat{\jmath}}-d_{1})+L(d_{1}-d_{i})\right)\notag\\
&=\sum_{i=1}^I w(p_{i})v\left(L(Q_{\hat{\jmath}}-d_{i})\right).\label{eq:ind2}
\end{align}
%\end{spacing}
%\vspace{2mm}
\noindent Here, $\pi_{ind}^b$ and $\pi_{ind}^s$ are the user's indifference prices, where $\pi_{ind}^b$ corresponds to the price below which he is willing buy data at $d^{\max}-Q_{\hat{\jmath}}$, and $\pi_{ind}^s$ is the price above which he is willing to sell data at $Q_{\hat{\jmath}}-d^{\min}$, where $d^{\max}$ and $d^{\min}$ are defined in line 5.

In Appendix F, we establish that every pair of indifference equations (\ref{eq:ind1}) and (\ref{eq:ind2}) of Algorithm 1 has a unique solution of $\lambda$ and $\beta$.

\augg{When estimating the indifference price, the user may not have an exact value in mind. Hence, to improve the estimation accuracy of Algorithm 1, we have an estimation period of $T_{\hat{m}}$ days (line 2 to \res{9}), and have $T_{\hat{m}}$ pairs of difference equations with different demand predictions. Then, we choose the average values among the solutions of the equations (line \res{10}). }

\section{Performance Evaluation}\label{sec:simulation}
In Section \ref{sec:simulation1}, we first illustrate the impact of the PT model parameters on the user's optimal decision of a single trading in the billing cycle. \res{Then we evaluate the performance of our algorithm by numerically simulating the case of making multiple decisions in a billing cycle in Section \ref{sec:simulation2}. }

The simulations illustrate the following insights for a PT user' optimal trading decision \aug{(by comparing with an EUT user)}: (i) risk-seeking dominant under a high reference point: \aug{Without considering the effect of probability distortion, a PT buyer is risk-seeking and is less willing to buy mobile data and more willing to sell mobile data than an EUT buyer}. (ii) Probability distortion: For the case of binary demand realizations, when the probability of high demand is \aug{small}, a PT buyer is risk-averse and is more willing to buy mobile data. On the other hand, when the probability of high demand is \aug{large}, a PT buyer is risk-seeking and is less willing to buy mobile data. \aug{(iii) Profit: A PT user achieves a lower average profit than an EUT user. However, a risk-seeking dominant user can achieve a higher maximum profit, while a risk-averse dominant user can guarantee a higher minimum profit.}

\begin{figure*}[t]
\begin{minipage}{0.32\linewidth}
	\centering
					\includegraphics[width = 2.17in]{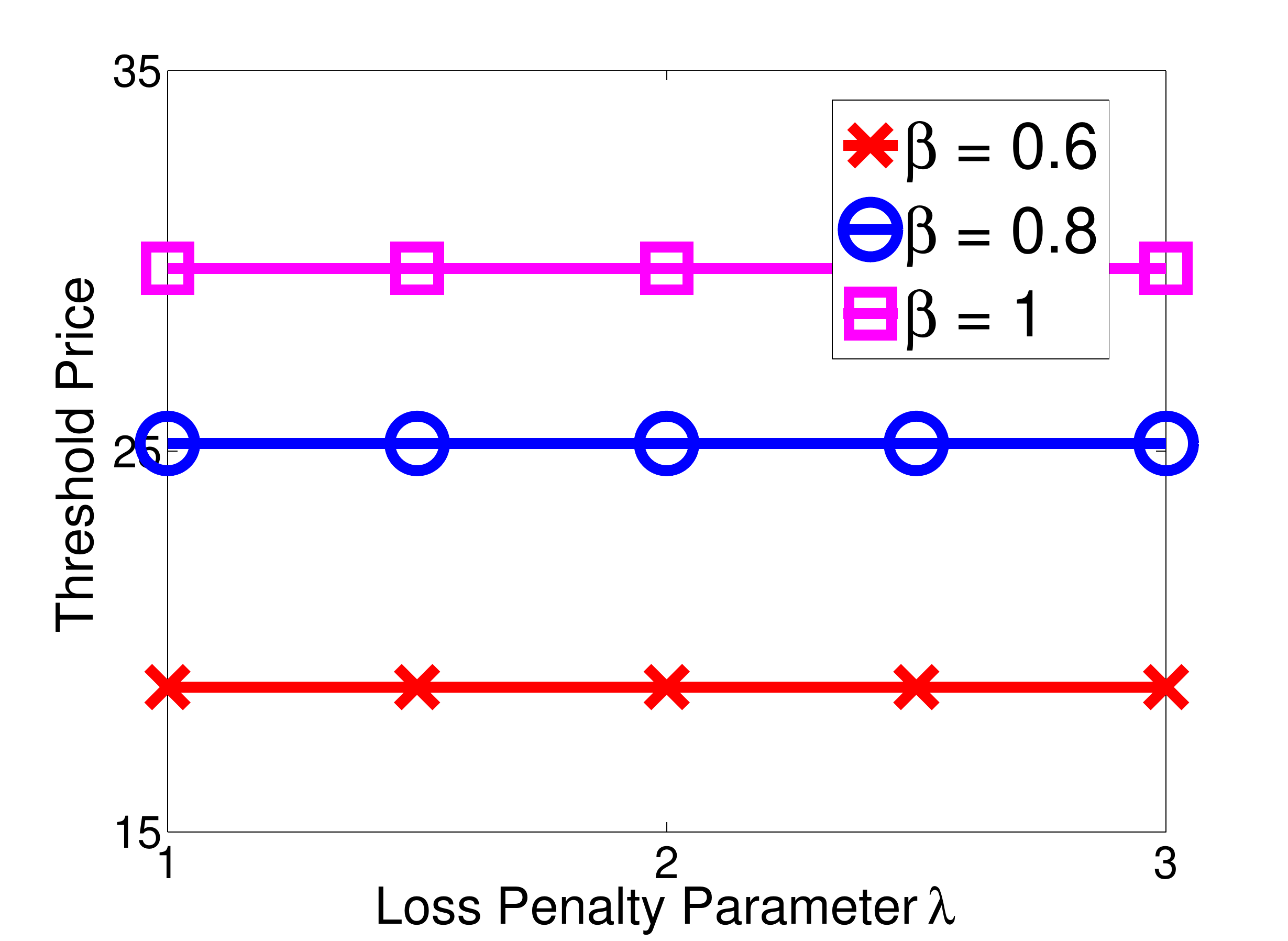}
	  \caption{Buyer's threshold price $\bar{\pi}_b^{PTh}$ versus loss penalty parameter $\lambda$ with different $\beta$.}
	  \label{fig:6}
\end{minipage}
\hfill
\begin{minipage}{0.32\linewidth}
\vspace{1.5mm}
	\centering
					\includegraphics[width = 2.18in]{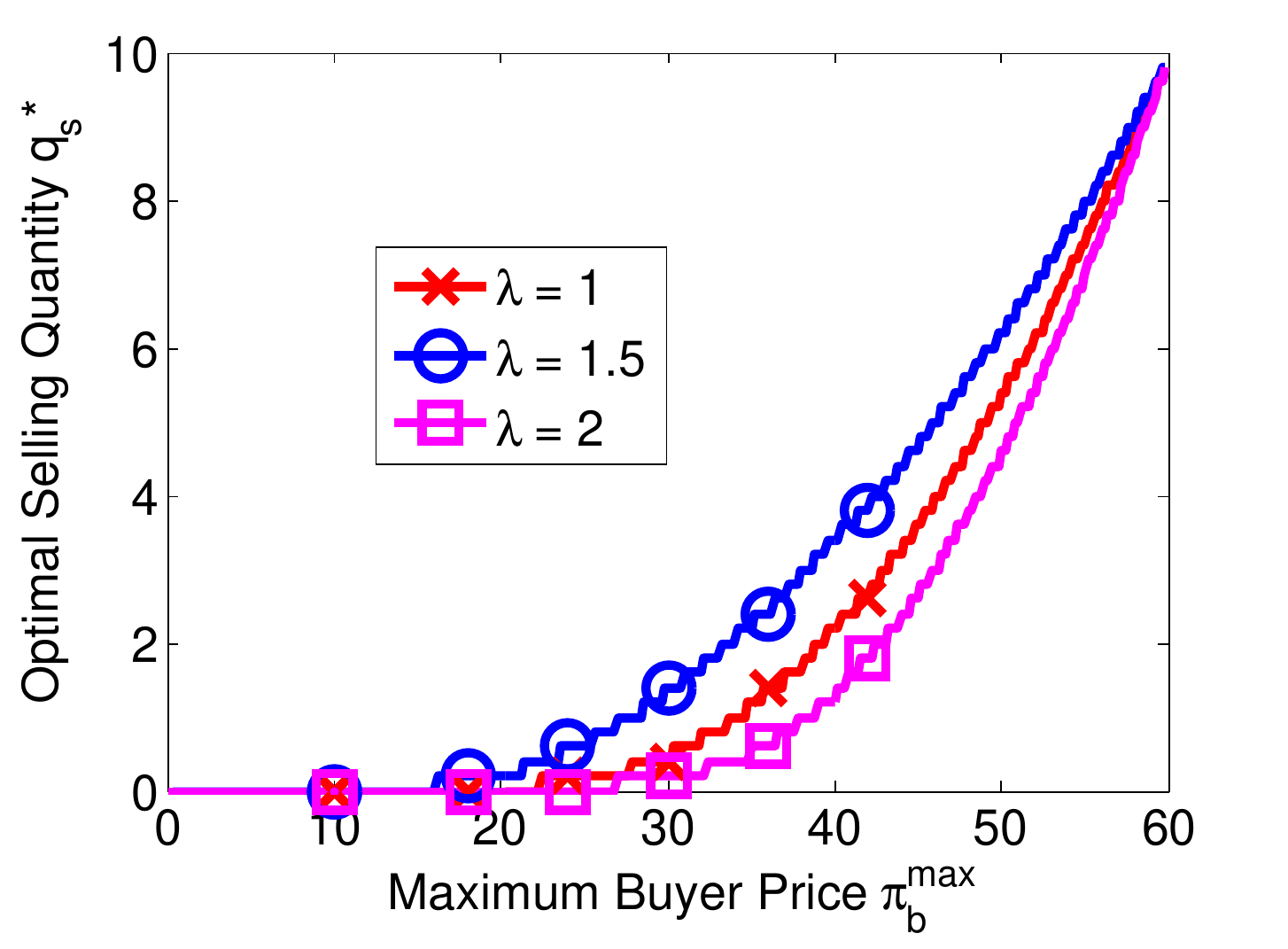}
					\vspace{-1.6mm}
				\caption{Seller's selling quantity $q_s^*$ versus maximum buying price $\pi_b^{\max}$ with different $\lambda$.}
				\label{fig:8}
\end{minipage}
\hfill
\begin{minipage}{0.32\linewidth}
\vspace{2mm}
	\centering
					\includegraphics[width = 2.17in]{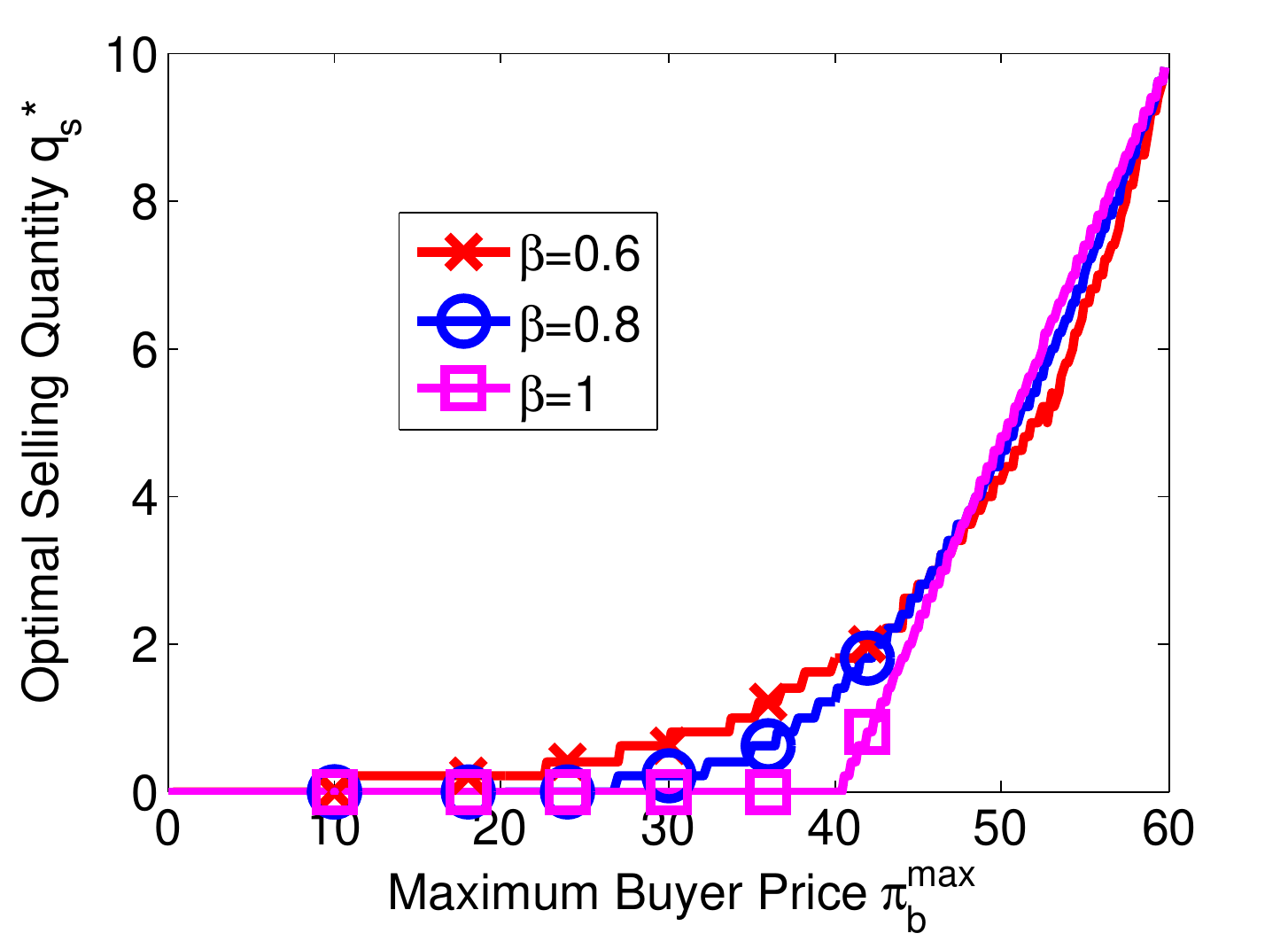}
					\vspace{-1.3mm}
				\caption{Seller's selling quantity $q_s^*$ versus maximum buying price $\pi_b^{\max}$ with different $\beta$.}
				\label{fig:9}
				\end{minipage}
\end{figure*}

\subsection{Impact of PT Model Parameters}\label{sec:simulation1}
\resf{In this subsection, we illustrate the impact of the PT model parameters ($\lambda$, $\beta$, and $\mu$) and market parameters ($\pi_s^{min}$ and $\pi_b^{max}$) on the user's optimal decision with $I=20$ possible outcomes in Figs. \ref{fig:6}, \ref{fig:8}, and \ref{fig:9}, and then illustrate the impact of the demand uncertainty parameter ($p$) with binary outcomes ($I=2$) in Fig. \ref{fig:7}.} Due to space limitations, we will only consider the high reference point $R_p=0$ for the PT case. 

\textbf{Impact of the loss penalty parameter $\lambda$ and the risk parameter $\beta$ on a buyer's threshold price $\bar{\pi}_b^{PTh}$:} \res{Here we assume $\mu = 1$ and $p_1=p_2=\ldots=p_{20}=0.05$. Fig.~\ref{fig:6} shows that the buyer threshold price $\bar{\pi}_b^{PTh}$ is increasing in $\beta$ for a fixed value of $\lambda$, and does not change in $\lambda$ for a fixed value of $\beta$. Note that a higher threshold price means that the buyer is more willing to buy mobile data.}
This is because under the high reference point $R_p=0$, the buyer will consider any possible outcome as a loss. In this case, a smaller $\beta$ means that the user is more risk-seeking in losses, so he does not need to purchase mobile data to reduce the risk that the future data demand exceeds the quota. Meanwhile, \aug{\textcolor[rgb]{0,0,0}{notice that} $\lambda$ \textcolor[rgb]{0,0,0}{is used for differentiate the value function in the loss region and gain region in (\ref{equ:valuefunction}).}} \textcolor[rgb]{0,0,0}{As the user will never encounter a gain in this case, the threshold price is independent of $\lambda$.}

\res{\textbf{Impact of the loss penalty parameter $\lambda$ and the risk parameter $\beta$ on a seller's optimal selling quantity $q_s^*$:} Fig.~\ref{fig:8} illustrates how the seller's selling quantity $q_s^*$ changes with the maximum buying price $\pi_b^{\max}$ and $\lambda$. Here we assume that $p_1=p_2=\ldots=p_{20}=0.05$, $\mu = 1$, and $\beta = 0.8$. Fig.~\ref{fig:8} shows that $q_s^*$ increases in $\pi_b^{\max}$. This is because as $\pi_b^{\max}$ increases, the seller gains more revenue from the trade, hence he wants to sell more.}
Fig.~\ref{fig:8} also shows that under the same value of $\pi_b^{\max}$, $q_s^*$ is non-increasing in $\lambda$. This is because, as $\lambda$ increases, the seller becomes more loss averse, hence \textcolor[rgb]{0,0,0}{he} will sell less in order to avoid a heavy loss when the future data demand is high.

\res{Fig.~\ref{fig:9} illustrates how the seller's selling quantity $q_s^*$ changes with the maximum buying price $\pi_b^{\max}$ and $\beta$. Here we assume that $\mu = 1$ and $\lambda = 2$. Fig.~\ref{fig:9} shows that $q_s^\ast$ is decreasing in $\beta$ under a small $\pi_b^{\max}$, and is increasing in $\beta$ under a large $\pi_b^{\max}$. This is because under the high reference point $R_p=0$, the seller will encounter either a small gain or a large loss. In this case, a smaller $\beta$ means that the user is more risk-averse dominant, hence becomes more willing to sell mobile data. However, when $\pi_b^{\max}$ is large, the seller will encounter a large gain from selling data.   In this case, a smaller $\beta$ means that the user is more risk-seeking dominant, hence becomes less willing to sell mobile data.}

\textbf{Impact of the probability distortion parameter $\mu$ on a buyer's threshold price $\bar{\pi}_b^{PTh}$ in (\ref{equ:pib}):} \res{To illustrate the impact of the probability distortion parameter, we assume binary outcomes with $I=2$.} Fig.~\ref{fig:7} considers three different probabilities of high demand: high ($p=0.8$), medium ($p=0.5$), and low ($p=0.2$). Here we assume $\beta = 0.8$ and $\lambda = 2$. We can see that $\bar{\pi}_b^{PTh}$ decreases in $\mu$ when $p=0.2$, is independent of $\mu$ when $p=0.5$, and increases in $\mu$ when $p=0.8$. As a smaller $\mu$ means that the buyer will overweigh the low probability more, he becomes more risk-averse (i.e., $\bar{\pi}_b^{PTh}$ decreases) when $p$ is small. Similarly, since a smaller $\mu$ means that the buyer will underweigh the high probability more, he is more risk-seeking (i.e., $\bar{\pi}_b^{PTh}$ increases) when the $p$ is large. 

\begin{figure*}[t]
\begin{minipage}{0.32\linewidth}
\centering
					\includegraphics[width = 2.17in]{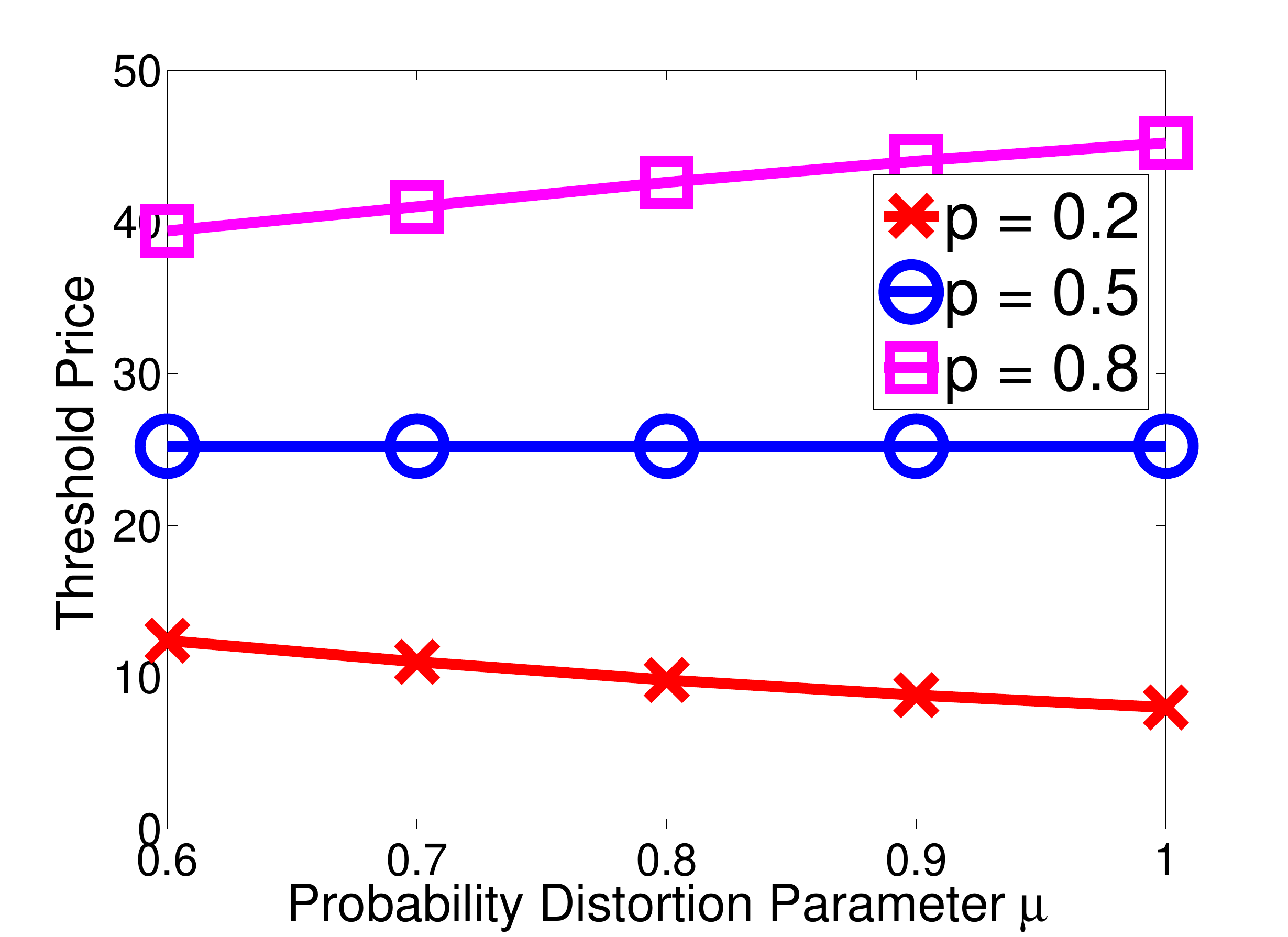}
  \caption{Buyer's threshold price $\bar{\pi}_b^{PTh}$ versus probability distortion parameter $\mu$ with different $p$.}
  \label{fig:7}
				\end{minipage}
\hfill
\begin{minipage}{0.32\linewidth}
\vspace{1.8mm}
	\centering
					\includegraphics[width = 1.9in]{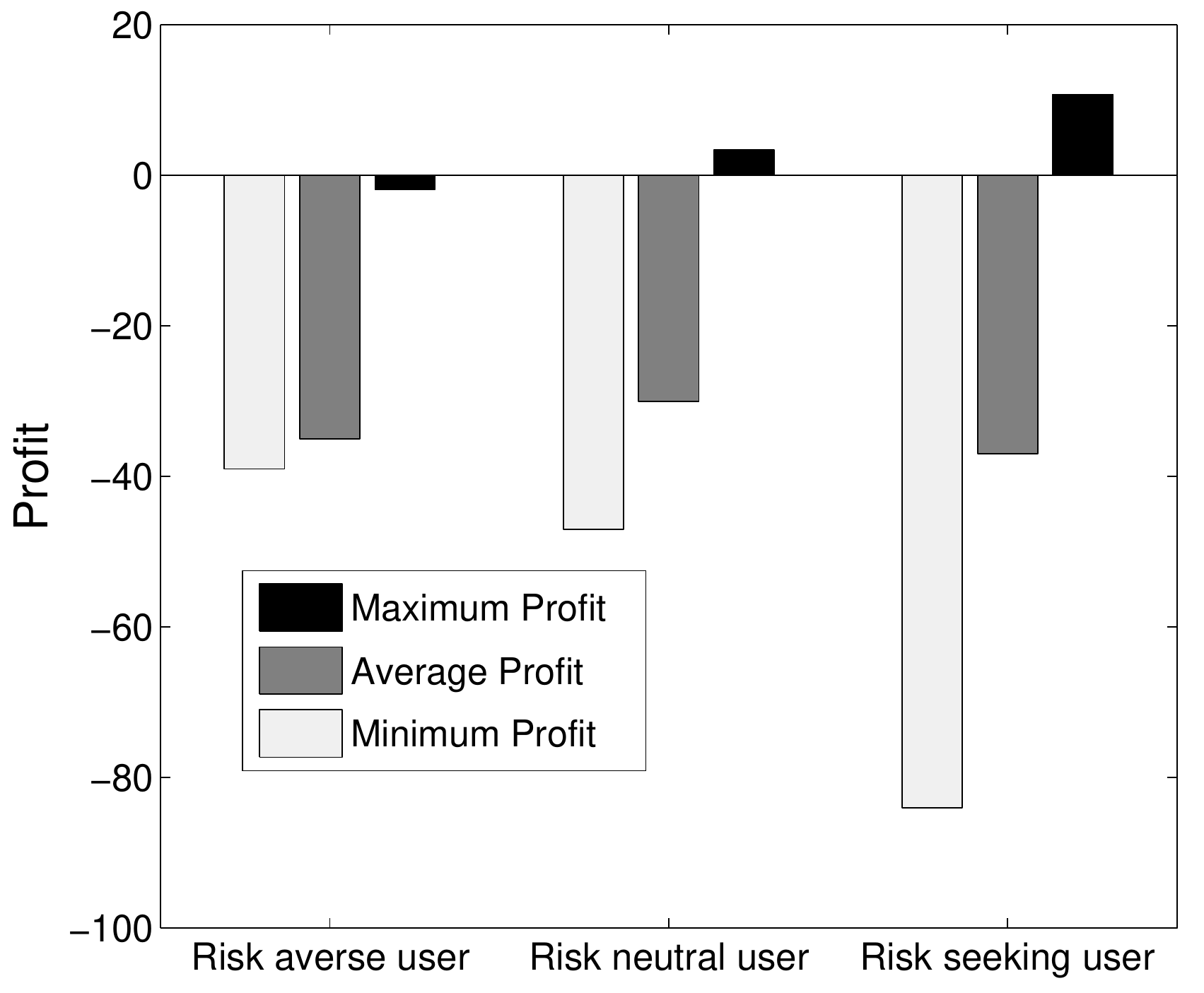}
			\vspace{-1.3mm}
				\caption{The user profit with different risk preferences.}
				\label{fig:82}
\end{minipage}
\hfill
\begin{minipage}{0.32\linewidth}
\vspace{2mm}
	\centering
					\includegraphics[width = 2.08in]{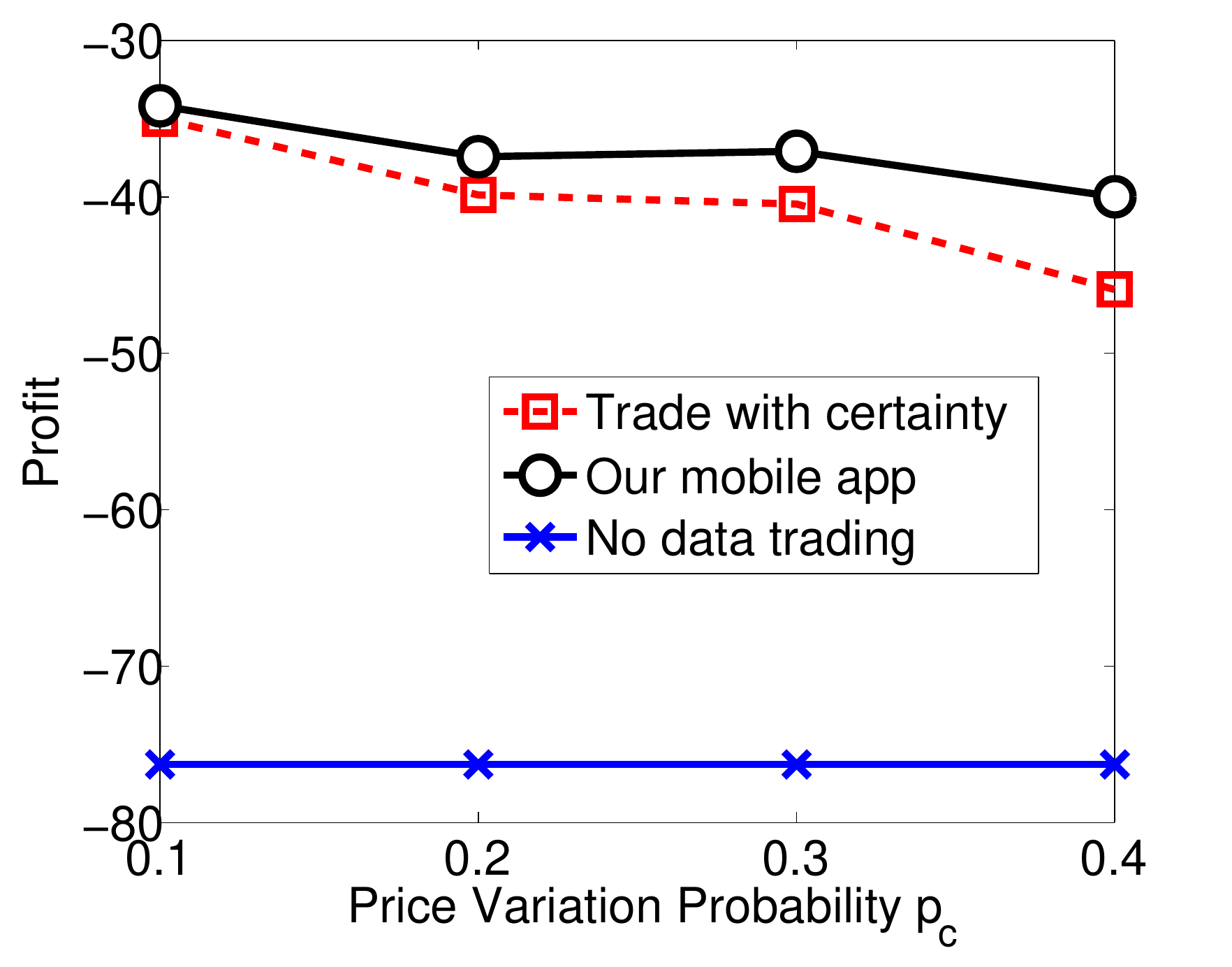}
			\vspace{-2mm}
				\caption{\res{The user profit with different price variations under the uniformly distributed usage}.}
				\label{fig:81}
				\end{minipage}
\end{figure*}

\subsection{\res{Evaluation of The Mobile Data Trading Algorithm}}\label{sec:simulation2}
\res{We then evaluate the total profit generated by our algorithm's trading decisions (introduced in Section \ref{sec:implementation}) in a billing cycle.} For each simulation, we consider a billing cycle of $T=30$ time slots. In the simulation settings, we assume that \aug{across two consecutive time slots}, the prices $\pi_b^{\max}$ and $\pi_s^{\min}$ increase by one unit (i.e., dollar) with probability $p_{c}$, decrease by one unit with probability $p_{c}$, or remain unchanged with probability $1-2p_c$. The changes of $\pi_b^{\max}$ and $\pi_s^{\min}$ are independent. We set the monthly quota $\hat{Q}=2$ GB, and randomly generate the previous $I$ months' total demand $d_i$ (defined in (\ref{eq:24yu})) with a mean value of $2$ GB\footnote{\res{In our simulation, we generate both uniformly distributed and normally distributed (with standard deviation $1/3$) demands.}}. The algorithm calculates the trading decision in every time slot based on the user's risk preferences under the high reference point $R_p=0$. Specifically, we define the profit\footnote{\augg{The profit may be negative, meaning that the total revenue due to selling data is lower than the payment due to buying data plus the payment due to satisfaction loss.}} $P_{\hat{m}}$ of month $\hat{m}$ as
\auggg{\begin{equation}
P_{\hat{m}}=\sum_{\hat{\jmath}=1}^{T_{\hat{m}}} -q_{\hat{\jmath}}^* \pi_{\hat{\jmath}} - L(Q+\sum_{{\hat{\jmath}}=1}^{T_{\hat{m}}}q_{\hat{\jmath}}^* - \sum_{\hat{\jmath}=1}^{T_{\hat{m}}} d_{\hat{m},\hat{\jmath}}), \label{eq:yu28}
\end{equation}
\noindent which consists of two parts: the net revenue due to selling or buying data, and the payment due to satisfaction loss. In (\ref{eq:yu28}), a positive $q_{\hat{\jmath}}^*$ means that the user buys data quota in day $\hat{\jmath}$, while a negative $q_{\hat{\jmath}}^*$ means that the user sells data quota in day $\hat{\jmath}$.}

By repeatedly running the simulation for 1000 billing cycles with randomly generated demands and prices, we first evaluate the impact of risk preferences on the maximum profit, minimum profit, and the average profit. \res{We then compare the average profit achieved by the algorithm implemented by our mobile app and several other benchmark strategies under different price variations with different price variations. In the first benchmark strategy \emph{``trade with certainty''}, we assume that the user is not willing to trade when he has uncertainty. This means that he will only trade once near the end of his billing cycle, when he knows the exact value of his monthly usage. In the second benchmark strategy \emph{``no data trading''}, the user does not trade at all. We compare these three strategies under the uniformly distributed usage.}

In Fig.~\ref{fig:82},  we assume $p_c=0.1$, and plot the profit of the user with different risk preferences. The risk parameters of different users are: (a) risk-averse dominant user: $\beta=1$, $\lambda = 2$; (b) risk-neutral user (EUT user): $\beta = \lambda =1 $; (c) risk-seeking dominant user: $\beta = 0.8$, $\lambda = 1$.\footnote{ A larger $\lambda$ indicates that the user is more loss averse, hence is more risk-averse. Since we have assumed a high reference point $R_p=0$, a smaller $\beta$ means the user is more risk-seeking dominant.} Since an EUT operator makes decision only by maximizing expected profit, we can see from Fig.~\ref{fig:82} that he can achieve the highest average profit. On the other hand, a PT operator makes decision by taking into account both the expected profit and its risk preferences. More specifically, although both risk-seeking dominant and risk-averse dominant PT users achieve a lower expected profit comparing to an EUT user, the risk-seeking dominant user can earn a higher maximum possible profit, while the risk-averse dominant user can guarantee a higher minimum possible profit. This is because the risk-seeking dominant user trades more quota, hence earns more when the price change is profitable, and loses more when the price change is unprofitable.

%\final{In Fig.~\ref{fig:81}, we plot the profit of the risk-neutral user with different values of the price variation probability $p_c$. Fig.~\ref{fig:81} shows that when $p_c=0.1$, our mobile app's profit is around $5\%$ higher than trade-once strategy. When $p_c=0.4$, our mobile app's profit is around $15\%$ higher than trade-once strategy. This is because our mobile app will suggest users to help buy when the price is low and sell when the price is high, hence takes better advantage of the price variation.}
\res{In Fig.~\ref{fig:81}, we plot the profit of the risk-neutral user with different values of the price variation probability $p_c$ under the uniformly distributed usage. Fig.~\ref{fig:81} shows that the gap between the profits generated by ``our mobile app'' and the ``trade with certainty'' strategies increases with $p_c$, e.g., the gap at $p_c=0.4$ is $500\%$ larger than the gap at $p_c=0.1$. 
This is because our mobile app suggests users buy when the price is low and sell when the price is high, hence takes advantage of the price variation. Comparing with the ``no data trading'' strategy, the user significantly benefits from the data trading market (i.e., reduces his net payment by 50\%).}

\section{Conclusion}
In this paper, we have considered a mobile data trading market that is motivated by the CMHK's 2CM platform. We have analyzed the optimal trading decision of a single user under a large market regime. 
We have compared and contrasted the user's optimal decisions under prospect theory (PT) and expected utility theory (EUT), and have highlighted several key insights. Comparing with an EUT user, a PT user with a high reference point is less willing to buy mobile data and more willing to sell mobile data. Moreover, when the probability of high demand is low, a PT user is more willing to buy mobile data comparing with an EUT user. On the other hand, when the probability of high demand is high, a PT user is less willing to buy mobile data. \res{In addition, we have designed a mobile data trading algorithm to recommend multiple trading decisions based on the user's current usage and risk preferences.} Our results suggested that a  risk  averse dominant user  can  achieve  the  highest  minimum  profit,  a  risk-seeking dominant user can achieve the highest maximum profit, while a risk-neutral user can achieve the highest average profit.

This study demonstrated that a more realistic behavioral modeling based on PT can shed important insights in understanding user's behavior on mobile data trading. \res{In the future work, we will use our app to collect data from the real market to help us understand users' real behaviors in data trading, and study the trading decision equilibria among all the market users and consider the service provider's data plan optimization. It is also interesting to study how the market competition among the service providers and the user-initiated data plan trading (such as that proposed in \cite{7565653}) affect the users' decisions, and how the data trading market affect the operator's other data plans (e.g., shared data plan).}

%%%%%%%%%%%%%%%%%%%%%%%%%%%%%%%%%%%%%%%%%%%%%%%%%%%%%%%%%%%%%%%%%
\bibliographystyle{IEEEtran}
\bibliography{IEEEabrv,mybibfile}
%%%%%%%%%%%%%%%%%%%%%%%%%%%%%%%%%%%%%%%%%%%%%%%%%%%%%%%%%%%%%%%%%%
% Note that IEEE does not put floats in the very first column - or typically
% anywhere on the first page for that matter. Also, in-text middle ("here")
% positioning is not used. Most IEEE journals/conferences use top floats
% exclusively. Note that, LaTeX2e, unlike IEEE journals/conferences, places
% footnotes above bottom floats. This can be corrected via the \fnbelowfloat
% command of the stfloats package.

%\section{Conclusion}
%The conclusion goes here. this is more of the conclusion


\appendix

%-----------------------------------------------
\subsection{Proof of Theorem 1}
%-------------------------------------------------
We divide the feasible interval of buying quantity $q_b$ into $I-\hat{\imath}+1$ sub-intervals, $[0,d_{\hat{\imath}+1}-Q],\ldots, [d_I-Q,\infty)$, and analyze the optimal buying quantity $q_b^*$ that maximizes $U(b,q_b)$ within each sub-interval. Such a division is based on (\ref{eq:ly}) that $L(Q+q_b-d_i)=0$ when $q_b\geq d_i-Q$.

\noindent\emph{1) Buyer's Problem in (\ref{equ:stage2b}) Under PT with $R_p =0$:}

\textbf{Case I}: $q_b\in[d_{j-1}-Q, d_j-Q], ~j=\hat{\imath}+1, \ldots, I$. In this case, the satisfaction loss under low demands is $L(Q+q_b-d_{i})=0$ for $i =1,\ldots,j-1$, and the satisfaction loss under high demands is $L(Q+q_b-d_i)=\kappa(Q+q_b-d_i)$ for $i=j,\ldots,I$. Thus, from (\ref{equ:valuefunction}) and (\ref{equ:stage2b}), we obtain
\begin{align}
U(b,q_b)  =&  \sum_{i=1}^{j-1}w(p_i)v[-\pi q_b]  \notag\\
&+   \sum_{i=j}^I w(p_i)v[\kappa(Q - d_i) - \pi q_b + \kappa q_b]\notag\\
=&- \lambda \sum_{i=1}^{j-1}w(p_i){(\pi q_b)}^{\beta} \notag\\
&- \lambda \sum_{i=j}^{I}w(p_i){[(\pi  -  \kappa)q_b +\kappa(d_i - Q)]}^{\beta}.
\end{align}

The second order partial derivative of $U(b,q_b)$ with respect to $q_b$ is
\begin{align}
&\frac{\partial U^2(q_b )}{\partial q_b^2}  =   -  \sum_{i=1}^{j-1}w( p_i )( \beta - 1 )\lambda\beta {\pi}^{\beta} q_b^{\beta-2}\notag\\
&-   \sum_{i=j}^I w( p_{i} ){(\pi - \kappa)}^2\lambda\beta( \beta - 1 ) {[( \pi - \kappa )q_b + \kappa( d_i - Q )]}^{\beta -2}  \notag\\
&> 0,
\end{align}
\noindent which implies that $U(b,q_b)$ is a convex function in $q_b$, and the optimal solution must lie at one of the boundary points. 
Hence, $q_b^*=\arg\max_{q_b\in\{d_i-Q,i=j-1,j\}} \{U(b,q_b)\}$.

\textbf{Case II}: $q_b\in[0, d_{\hat{\imath}+1} - Q]$. In this case, the satisfaction loss under low demands is $L(Q+q_b-d_{i})=0$ for $i<\hat{\imath}$, and the satisfaction loss under high demands is $L(Q+q_b-d_i)=\kappa(Q+q_b-d_i)$ for $i\geq\hat{\imath}$. Thus, from (\ref{equ:valuefunction}) and (\ref{equ:stage2b}), we obtain
\begin{align}
U(b,q_b)  =&  \sum_{i=1}^{\hat{\imath}-1}w(p_i)v(-\pi q_b)  \notag\\
&+  \sum_{i = \hat{\imath}}^I
v[-\pi q_b  +  L(Q  - d_i  +  q_b)].
\end{align}

Similar to Case I, we can show that $U(b,q_b)$ is a convex function in $q_b$, and the optimal solution must lie at one of the boundary points. 
Hence, we can obtain that $q_b^{*} = \arg\max_{q_b\in\{0,d_{\hat{\imath}+1}-Q\}} \{U(b,q_b)\}$.

\textbf{Case III}: $q_b\in[d_I - Q,\infty]$. In this case, the satisfaction loss is $L(Q+q_b-d_{i})=0$. Thus, the expected utility from (\ref{equ:stage2b}) is decreasing in $q_b$, and the optimal $q_b$ in this range is $q_b^*=d_I-Q$.

Combining the analysis in Cases I, and III, we have
$q_b^*=\arg\max_{q_b\in\{Q-d_i,i=\hat{\imath}+1,\ldots,I\}\cup\{0\}} \{U(b,q_b)\}$.

\noindent\emph{2) Buyer's Problem in (\ref{equ:stage2b}) Under PT with $R_p =\kappa(Q - d_I)$:}
\textbf{Case I}: $q_b\in[d_{j-1}-Q,d_j-Q], ~j=\hat{\imath}+1, \ldots, I$. In this case, the satisfaction loss under low demands is $L(Q+q_b-d_{i})=0$ for $i=1,\ldots,j-1$, and the satisfaction loss under high demands is $L(Q+q_b-d_l)=\kappa(Q+q_b-d_j)$ for $i=j,\ldots,I$. Thus, from (2) and (5), we obtain
\begin{align}
&U(b,q_b) = \sum_{i=1}^{j-1}w(p_i)V[-\pi q_b - \kappa(Q - d_{I})] \notag\\
&+   \sum_{i=j}^I w(p_i)V[\kappa(Q-d_i) - \pi q_b + \kappa q_b  -  \kappa(Q - d_I)]\notag\\
&=\sum_{i=1}^{j-1}w(p_i){[-\pi q_b + \kappa(d_I-Q)]}^{\beta} \notag\\
&+\sum_{i=j}^{I}w(p_i){[(\kappa- \pi)q_b +\kappa(d_I-d_i)]}^{\beta}.
\end{align}

The second order partial derivative of $U(b,q_b)$ with respect to $q_b$ is
\begin{align}
&\frac{\partial U^2(q_b)}{\partial q_b^2} = \sum_{i=1}^{j-1}w(p_i){(-\pi)}^2(\beta - 1)\beta  {[  \kappa (d_I - Q)-\pi q_b]}^{\beta-2} \notag\\
&+  \sum_{i=j}^I w(p_{i}){(\kappa - \pi)}^2(\beta - 1)\beta {[(\kappa - \pi)q_b + \kappa(d_I - d_i)]}^{\beta -2}\notag\\
 &<0,
\end{align}
which implies that $U(b,q_b)$ is a concave function in $q_b$, and the optimal solution must lie at a critical point where $U'(b,q_b)=0$ (if such a point exists in the sub-interval) or one of the boundary points (if a critical point does not exist in the sub-interval). 
Hence, $q_b^*=\arg\max_{q_b\in\{Q-d_i,i\in\mathcal{I}\}\cup\mathcal{X}_b} \{U(b,q_b)\}$.

\textbf{Case II}: $q_b \in[0, d_{\hat{\imath}}-Q]$. In this case, the satisfaction loss under low demands is $L(Q+q_b-d_{i})=0$ for $i<\hat{\imath}$, and the satisfaction loss under high demands is $L(Q+q_b-d_l)=\kappa(Q+q_b-d_i)$ for $i\geq\hat{\imath}$. Thus, from (5) and (2), we obtain
\begin{align}
&U(b,q_b) = \sum_{i=1}^{\hat{\imath}-1}w(p_i)V[-\pi q_b - \kappa(Q  -  d_{I})] \notag\\
&+ \sum_{i=\hat{\imath}}^I w(p_i)V[\kappa(Q - d_i) - \pi q_b + \kappa q_b  -  \kappa(Q - d_I)].
\end{align}

Similar to Case I, we know that $U(b,q_b)$ is a concave function in $q_b$, which implies that $U(b,q_b)$ is a concave function in $q_b$, and the optimal solution must lie at a critical point where $U'(b,q_b)=0$ (if such a point exists in the sub-interval) or one of the boundary points (if a critical point does not exist in the sub-interval).  Hence, $q_b^*=\arg\max_{q_b\in\{0,d_{\hat{\imath}}-Q\}\cup\mathcal{X}_b} \{U(b,q_b)\}$.

\textbf{Case III}: $q_b\in[d_I - Q,\infty]$. In this case, the satisfaction loss is $L(Q+q_b-d_{i})=0$. Thus, the expected utility from (5) is decreasing in $q_b$, and the optimal $q_b$ in this range is $q_b^*=d_I-Q$.

Combining the analysis in Case I, Case II, and Case III, we can obtain that
 $q_b^*=\arg\max_{q_b\in\{Q-d_i,i\in\mathcal{I}\}\cup\mathcal{X}_b} \{U(b,q_b)\}$.

%%%%%%%%%%%%%%%%%%%%%%%%%%%%%%%%%%%%%%%%%%%%%%%%%%%%%%%%
\subsection{Proof of Theorem 2}

For all three cases, we divide the feasible interval of buying quantity $q_b$ into two sub-intervals, $[0,d_h-Q]$ and $[d_h-Q,\infty)$, and analyze the optimal buying quantity $q_b^*$ that maximizes $U(b,q_b)$ within each sub-interval. Such a division is due to the fact that the satisfaction loss $L(Q+q_b-d_h)=0$ when $q_b\in [d_h-Q,\infty)$.
\subsubsection{Buyer's Problem Under EUT}
\begin{itemize}
	\item Case I: $q_b\in [0,d_h-Q]$. In this case, from (1), the satisfaction loss under low demand is $L(Q+q_b-d_l)=0$, and the satisfaction loss under high demand is $L(Q+q_b-d_h)=\kappa(Q+q_b-d_h)$. The expected utility from (5) is  
	\begin{align}
	U(b,q_b) =(\kappa p-\pi_s^{\min})q_b+\kappa p(Q-d_h),
	\end{align}
	which is a linear function in $q_b$. It is increasing in $q_b$ when $\pi_s^{\min}<\kappa p$, and decreasing in $q_b$ when $\pi_s^{\min}>\kappa p$. The optimal buying quantity is then $q_b^*=d_{h}-Q$ when $\pi_s^{\min}<\kappa p$, and $q_b^*=0$ when $\pi_s^{\min}> \kappa p$.  When $\pi_s^{\min}=\kappa p$, the utility is independent of $q_b$. Without loss of generality, we assume that $q_b^*=0$ when $\pi_s^{\min}=\kappa p$.
	\item Case II: $q_b\in [d_h-Q,\infty)$. In this case, the satisfaction loss under both low demand and high demand equals to $0$, and the utility $U(b,q_b)=-\pi_s^{\min} q_b$. Since the utility function $U(b,q_b)$ is linearly decreasing in $q_b$, we have $q_b^*=d_{h}-Q$ in this case.
\end{itemize}
Combing the above analysis, we obtain (11).

\subsubsection{Buyer's Problem Under PT with $R_p=0$ }
\begin{itemize}
	\item Case I: $q_b\in [0,d_h-Q]$. In this case, from (1), the satisfaction loss under low demand is $L(Q+q_b-d_l)=0$, and the satisfaction loss under high demand is $L(Q+q_b-d_h)=\kappa(Q+q_b-d_h)$. The expected utility from (5) is  
		\begin{align}
	U(b,q_b)=&-\lambda(\pi_s^{\min} q_b-\kappa (Q+q_b-d_{h}))^{\beta}w(p)\notag\\
	&-\lambda(\pi_s^{\min} q_b)^{\beta}w(1-p).
	\end{align}
	The second order partial derivative of $U(b,q_b)$ with respect to $q_b$ is 
%\begin{equation}
%\begin{split}
\begin{align}
   &\frac{\partial U^2 (b, q_b )}{\partial^2 q_b}    =  - \lambda\beta(\beta  -  1) \Bigl[( \pi_s^{\min} )^\beta(q_b)^{\beta}w( 1  -  p ) \notag\\
   &+ ( \pi_s^{\min}  - \kappa  )^2[( \pi_s^{\min}  - \kappa  )q_b  - \kappa ( Q -  d_{h} )]^{\beta - 2}w( p ) \Bigr]\notag\\
   &>0,
\end{align}
%\end{split}
%\end{equation}
which implies that $U(b, q_b)$ is a convex function in $q_b$, and the optimal solution must lie at one of the boundary points.\footnote{In the case $\beta =1$ and $U(b, 0)=U(b, d_h-Q)$, we will choose $q_b^*=0$ without loss of generality.} Hence $q_b^*= d_h-Q$ if $U(b, 0)<U(b, d_h-Q)$, and $q_b^*=0$ if $U(b, 0)\geq U(b, d_h-Q)$.
	
	\item Case II: $q_b\in [d_h-Q,\infty)$. In this case, the satisfaction loss under both low demand and high demand equals to $0$, and the expected utility is 
	\begin{align}
	U(b,q_b)=-\lambda[w(p)+w(1-p)](\pi_s^{\min} q_b)^{\beta}.
	\end{align} 
Since the first order partial derivative ${\partial U(b,q_b)}/{\partial q_b}<0$, $U(b,q_b)$ is a decreasing function of $q_b$, and $q_b^*=d_{h}-Q$ in this case.
\end{itemize}

Combing the above analysis, we obtain (12).

\subsubsection{Buyer's Problem Under PT with $R_p =\kappa(Q-d_h)$}

\begin{itemize}
	\item Case I: $q_b\in [0,d_h-Q]$. In this case, the satisfaction loss under low demand is $L(Q+q_b-d_l)=0$, and that under high demand is $L(Q+q_b-d_h)=\kappa(Q+q_b-d_h)$. The expected utility is 
		\begin{align}
	U(b,q_b)=&-\lambda(\pi_s^{\min} q_b+\kappa(d_h-Q))^{\beta}w(1-p)
	\notag\\
	&-\lambda((\kappa-\pi_s^{\min}) q_b)^{\beta}w(p).
	\end{align}
The second order partial derivative of $U(b,q_b)$ with respect to $q_b$ is 
	\begin{align}
	&\frac{\partial U^2(b,q_b)}{\partial^2 q_b} =\beta(\beta - 1)\{w(p)(\kappa - \pi_s^{\min})^{\beta}q_b^{\beta-2}\notag\\
	&+ w(1 - p)(\pi_s^{\min})^2[\kappa (d_{h} - Q) - \pi_s^{\min}q_b]^{\beta - 2}\} \notag\\
	&<0,
	\end{align}
	so $U(b, q_b)$ is a strictly concave function of $q_b$. As a result, the optimal solution $q_b^*$ satisfies the first order condition or lies at one of the boundary points. 
	
	We consider the first order partial derivative of $U(b,q_b)$ with respect to $q_b$:
		\begin{align}
	&\frac{\partial U(b,q_b)}{\partial q_b}=\beta\{w(p)(\kappa-\pi_s^{\min})^{\beta}q_b^{\beta-1}\notag\\
	&+ w(1 - p)(\pi_s^{\min})[\kappa (d_{h} - Q) - \pi_s^{\min}q_b]^{\beta-1}\}.
	\end{align}
	\begin{itemize}
	\item If $\beta = 1$, ${\partial U(b,q_b)}/{\partial q_b}$ is independent of $q_b$. When $\pi_s^{\min}<\frac{\kappa w(p)}{w(p)+w(1-p)}$, ${\partial U(b,q_b)}/{\partial q_b}>0$, so $q_b^*=d_{h}-Q$. When $\pi_s^{\min}\geq \frac{\kappa w(p)}{w(p)+w(1-p)}$, ${\partial U(b,q_b)}/{\partial q_b}\leq0$, so $q_b^*=0$. 
	\item If $0<\beta<1$, solving ${\partial U(b,q_b)}/{\partial q_b}=0$, we have $\tilde{q}_b=\frac{\kappa(Q-d_h)}{\left[\frac{w(p)(\kappa-\pi_s^{\min})^{\beta}}{w(1-p)\pi_s^{\min}}\right]^{\frac{1}{\beta-1}} + \pi_s^{\min}}>0$. If $\tilde{q}_b<d_h-Q$, then the optimal solution $q_b^\ast = \tilde{q}_b$. Otherwise, $q_b^\ast = d_h-Q$. 
	\end{itemize}
	\item Case II: $q_b\in [d_h-Q,\infty)$. In this case, the satisfaction losses under both low demand and high demand equal to $0$, and the expected utility is 
	\begin{align}
	U(b,q_b)=-\lambda[w(p)+w(1 - p)](\pi_s^{\min} q_b+\kappa(d_h - Q))^{\beta}.
	\end{align} 
Since the first order partial derivative ${\partial U(b,q_b)}/{\partial q_b}<0$, the utility function $U(b,q_b)$ is a decreasing function of $q_b$, so $q_b^*=d_{h}-Q$ in this case.
	\end{itemize}

Combing the above analysis, we obtain (13).

%------------------------------------------------------------
\subsection{Proof of Theorem 3}
%-------------------------------------------------------------
In the proof, we divide the feasible interval of selling quantity $q_s$ into $\hat{\imath}+1$ sub-intervals, $[0,Q-d_{\hat{\imath}}],\ldots, [Q-d_2,Q-d_1], [Q-d_1,\infty)$, and analyze the optimal buying quantity $q_s^*$ that maximizes $U(b,q_s)$ within each sub-interval. Such a division is based on the fact that the satisfaction loss $L(Q+q_s-d_i)=0$ when $q_s\geq d_i-Q$.

\noindent\emph{1) Seller's Problem in (\ref{equ:stage2s}) Under PT with $R_p =0$:}

\textbf{Case I}: $q_s\in[Q-d_{j},Q-d_{j-1}], ~j=2, \ldots, \hat{\imath}$. In this case, the satisfaction loss under low demands is $L(Q-q_s-d_{i})=0$ for $i =1,\ldots,j$, and the satisfaction loss under high demands is $L(Q-q_s-d_i)=\kappa(Q-q_s-d_i)$ for $i=j+1,\ldots,I$. The expected utility from (\ref{equ:stage2s}) is
\begin{align}
U&(s,q_s) = \sum_{i=1}^I w(p_i)v[\pi q_s+L(Q-q_s-d_i)]\notag\\
=&\sum_{i=1}^j w(p_i)v[\pi q_s] + \sum_{i=j+1}^I w(p_i)v[(\pi-\kappa)q_s+\kappa(Q-d_i)]\notag\\
=&\sum_{i=1}^j w(p_i){(\pi q_s)}^{\beta} +\sum_{i=j+1}^{\hat{\imath}}w(p_i)v[(\pi - \kappa)q_s  +  \kappa(Q - d_i)] \notag\\
&-  \sum_{i=\hat{\imath}+1}^{I}w(p_i)(\lambda){[(\kappa - \pi)q_s  + \kappa(d_i - Q)]}^{\beta}.\label{eq:32}
\end{align}

Since the sign of function $v[(\pi-\kappa)q_s+\kappa(Q-d_i)]$ depends on the value of $q_s$, $Q-d_j$, $Q-d_{j-1}$, $\frac{\kappa(Q-d_{j+1})}{\kappa-\pi}$, $\frac{\kappa(Q-d_{\hat{\imath}})}{\kappa-\pi}$, we divide the discussion in the following six cases:

(a) $\frac{\kappa(Q-d_{j+1})}{k-\pi}<Q-d_j$. In this case, $q_s\in(\frac{\kappa(Q-d_{j+1})}{\kappa-\pi},Q-d_{j-1}]$, and we have $v[(\pi-\kappa)q_s+\kappa(Q-d_i)]<0, ~\forall i = j+1,\ldots,\hat{\imath}$. 
The first order derivative 
\begin{align}
&\frac{\partial U(s,q_s)}{\partial q_s} = \sum_{i=1}^{\hat{\imath}}w(p_i){\pi}^{\beta} \beta q_s^{\beta-1}\notag\\
&-    \sum_{i=\hat{\imath}+1}^I    w(p_i)\lambda(\kappa - \pi)\beta{[(\kappa - \pi)q_s  + \kappa(d_i - Q)]}^{\beta-1}\notag\\
&= \sum_{i=1}^jw(p_1){\pi}^{\beta}\beta q_s^{\beta-1}\bigg[1 -   \sum_{i=j+1}^{I}  \frac{w(p_i)}{\sum_{i=1}^j w(p_i)}\lambda\beta(\frac{k}{\pi} \notag\\
&- 1){\left[\frac{k}{\pi} - 1 + \frac{k(d_i - Q)}{\pi q_s}\right]}^{\beta-1}\bigg],
\end{align}
where $$\sum_{i=j+1}^{I}\frac{w(p_i)}{\sum_{i=1}^jw(p_i)}(-\lambda)\beta(\frac{k}{\pi}-1){\left[\frac{k}{\pi}-1 +\frac{k(d_i-Q)}{\pi q_s}\right]}^{\beta-1}$$ is increasing in $q_s$. This indicates that the equation $U'(s,q_s)=0$ has at most one solution $q_s$. Hence, we can find that the optimal $q_s^*$ is the solution of $U'(s,q_s)=0$, or at one of the boundary points $\kappa(Q-d_{j+1})/(\kappa-\pi)$ or $Q-d_{j-1}$.

(b) $\frac{\kappa(Q-d_{\hat{\imath}})}{k-\pi}<Q-d_j<\frac{\kappa(Q-d_{j+1})}{k-\pi}<Q-d_{j-1}$. In this case, without loss of generality, we assume $\frac{\kappa(Q-d_{\hat{m+1}})}{k-\pi}<Q-d_j<\frac{\kappa(Q-d_{\hat{m}})}{k-\pi}$. When $q_s\in [\frac{\kappa(Q-d_{m+1})}{\kappa-\pi},\frac{\kappa(Q-d_{m})}{\kappa-\pi}], m=j+1,\ldots,\hat{m}$, we have $v[(\pi-\kappa)q_s+\kappa(Q-d_i)]\leq0$ for $i = m+1,\ldots,\hat{\imath}$, and $v[(\pi-\kappa)q_s+\kappa(Q-d_i)]\geq0$ for $i = j+1,\ldots,m$. From (\ref{eq:32}), we obtain
\begin{align}
&\frac{\partial U(s,q_s)}{\partial q_s}=\sum_{i=1}^jw(p_i){\pi}^\beta\beta q_s^{\beta-1} \notag\\
&+   \sum_{i=j+1}^m  w(p_i)\beta(\pi - \kappa)[(\pi - \kappa)q_s + \kappa(Q - d_i)]^{\beta-1}\notag\\
&-    \sum_{i=m+1}^I    w(p_i)\beta(\kappa - \pi)\lambda[(\kappa - \pi) q_s - \kappa(Q - d_i)]^{\beta-1}\notag\\
&=\sum_{i=1}^jw(p_i){\pi}^\beta\beta q_s^{\beta-1}\Bigg[1 \notag\\
&+   \sum_{i=j+1}^m  \frac{w(p_i)}{\sum\limits_{i=1}^j  w(p_i)}(1 - \frac{\kappa}{\pi}) \left[1 - \frac{\kappa}{\pi} + \frac{\kappa(Q - d_i)}{\pi q_s}\right]^{\beta-1}
\notag\\
&-     \sum_{i=m+1}^I   \frac{w(p_i)}{\sum\limits_{i=1}^j  w(p_i)}(\frac{\kappa}{\pi}  -  1)\lambda \left[\frac{\kappa}{\pi}  -  1  -  \frac{\kappa(Q - d_i)}{\pi q_s}\right]^{\beta-1}  \Bigg],
\end{align}
which follows a unimodal structure, and we can find that the optimal $q_s^*$ must lie at the critical point (i.e., the solution of $U'(s,q_s)=0$), or at one of the boundaries $\frac{\kappa(Q-d_{m+1})}{\kappa-\pi}$ or $\frac{\kappa(Q-d_{m})}{\kappa-\pi}$.
When $q_s\in(\frac{\kappa(Q-d_{j+1})}{\kappa-\pi},Q-d_{j-1}]$, we have $v[(\pi-\kappa)q_s+\kappa(Q-d_i)]<0, ~\forall i = j+1,\ldots,\hat{\imath}$. 
The first order derivative 
\begin{align}
&\frac{\partial U(s,q_s)}{\partial q_s} = \sum_{i=1}^{\hat{\imath}}w(p_i){\pi}^{\beta} \beta q_s^{\beta-1} \notag\\
&-   \sum_{i=\hat{\imath}+1}^I   w(p_i)\lambda(\kappa - \pi)\beta{[(\kappa - \pi)q_s  + \kappa(d_i - Q)]}^{\beta-1}\notag\\
&= \sum_{i=1}^jw(p_1){\pi}^{\beta}\beta q_s^{\beta-1}\bigg[-   \sum_{i=j+1}^{I} \frac{w(p_i)}{\sum_{i=1}^jw(p_i)}\lambda\beta(\frac{k}{\pi}\notag\\
&- 1){\left[\frac{k}{\pi}-1 +\frac{k(d_i - Q)}{\pi q_s}\right]}^{\beta-1} + 1\bigg],
\end{align}
where $$\sum_{i=j+1}^{I}\frac{w(p_i)}{\sum_{i=1}^jw(p_i)}(-\lambda)\beta(\frac{k}{\pi}-1){\left[\frac{k}{\pi}-1 +\frac{k(d_i-Q)}{\pi q_s}\right]}^{\beta-1}$$ is increasing in $q_s$. This indicates that the equation $U'(s,q_s)=0$ has at most one solution $q_s$. Hence, we can find that the optimal $q_s^*$ is the solution of $U'(s,q_s)=0$, or at one of the boundary points $\kappa(Q-d_{j+1})/(\kappa-\pi)$ or $Q-d_{j-1}$.

(c) $Q-d_j<\frac{\kappa(Q-d_{\hat{\imath}})}{k-\pi}<Q-d_{j-1}<\frac{\kappa(Q-d_{j+1})}{k-\pi}$. In this case, without loss of generality, we assume $\frac{\kappa(Q-d_{\hat{m+1}})}{k-\pi}<Q-d_{j-1}<\frac{\kappa(Q-d_{\hat{m}})}{k-\pi}$. When $q_s\in[Q-d_{j},\frac{\kappa(Q-d_{\hat{\imath}})}{\kappa-\pi})$, we have $v[(\pi-\kappa)q_s+\kappa(Q-d_i)]>0, ~\forall i = j+1,\ldots,\hat{\imath}$. From (\ref{eq:32}), we obtain
\begin{align}
&\frac{\partial U(s,q_s)}{\partial q_s} = \sum_{i=1}^{\hat{\imath}}w(p_i){\pi}^{\beta} \beta q_s^{\beta-1} \notag\\
&+  \sum_{i=\hat{\imath}+1}^I (-\lambda)(\kappa - \pi)\beta{[(\kappa - \pi)q_s  + \kappa(d_i - Q)]}^{\beta-1},
\end{align}
which also follows a unimodal structure, and we can find that the optimal $q_s^*$ must lie at the critical point (i.e., the solution of $U'(s,q_s)=0$), or at one of the boundaries $Q-d_{j}$ or $\frac{\kappa(Q-d_{\hat{\imath}})}{\kappa-\pi}$.
When $q_s\in [\frac{\kappa(Q-d_{m+1})}{\kappa-\pi},\frac{\kappa(Q-d_{m})}{\kappa-\pi}], m=\hat{m}+1,\ldots,\hat{\imath-1}$, we have $v[(\pi-\kappa)q_s+\kappa(Q-d_i)]\leq0$ for $i = m+1,\ldots,\hat{\imath}$, and $v[(\pi-\kappa)q_s+\kappa(Q-d_i)]\geq0$ for $i = j+1,\ldots,m$. From (\ref{eq:32}), we obtain
\begin{align}
&\frac{\partial U(s,q_s)}{\partial q_s}=\sum_{i=1}^jw(p_i){\pi}^\beta\beta q_s^{\beta-1} \notag\\
&+   \sum_{i=j+1}^m  w(p_i)\beta(\pi - \kappa)[(\pi - \kappa)q_s + \kappa(Q - d_i)]^{\beta-1}\notag\\
&-    \sum_{i=m+1}^I   w(p_i)\beta(\kappa - \pi)\lambda[(\kappa - \pi) q_s - \kappa(Q - d_i)]^{\beta-1}\notag\\
&=\sum_{i=1}^jw(p_i){\pi}^\beta\beta q_s^{\beta-1}\Bigg[1 + \sum_{i=j+1}^m\frac{w(p_i)}{\sum_{i=1}^jw(p_i)}(1 \notag\\
&- \frac{\kappa}{\pi})\left[1 - \frac{\kappa}{\pi}+\frac{\kappa(Q - d_i)}{\pi q_s}\right]^{\beta-1}
\notag\\
&-     \sum_{i=m+1}^I   \frac{w(p_i)}{\sum\limits_{i=1}^jw(p_i)}(\frac{\kappa}{\pi}  -  1)\lambda \left[\frac{\kappa}{\pi}  -  1  -  \frac{\kappa(Q - d_i)}{\pi q_s}\right]^{\beta-1}  \Bigg],
\end{align}
which follows a unimodal structure, and we can find that the optimal $q_s^*$ must lie at the critical point (i.e., the solution of $U'(s,q_s)=0$), or at one of the boundaries $\frac{\kappa(Q-d_{m+1})}{\kappa-\pi}$ or $\frac{\kappa(Q-d_{m})}{\kappa-\pi}$.

(d) $Q-d_j<\frac{\kappa(Q-d_{\hat{\imath}})}{k-\pi}<\frac{\kappa(Q-d_{j+1})}{k-\pi}<Q-d_{j-1}$. When $q_s\in(\frac{\kappa(Q-d_{j+1})}{\kappa-\pi},Q-d_{j-1}]$, we have $v[(\pi-\kappa)q_s+\kappa(Q-d_i)]<0, ~\forall i = j+1,\ldots,\hat{\imath}$. 
The first order derivative 
\begin{align}
&\frac{\partial U(s,q_s)}{\partial q_s} = \sum_{i=1}^{\hat{\imath}}w(p_i){\pi}^{\beta} \beta q_s^{\beta-1} \notag\\
&-   \sum_{i=\hat{\imath}+1}^I   w(p_i)\lambda(\kappa - \pi)\beta{[(\kappa - \pi)q_s + \kappa(d_i - Q)]}^{\beta-1}\notag\\
&= \sum_{i=1}^jw(p_1){\pi}^{\beta}\beta q_s^{\beta-1}\bigg[1 -    \sum_{i=j+1}^{I} \frac{w(p_i)}{\sum_{i=1}^jw(p_i)}\lambda\beta(\frac{k}{\pi}\notag\\
&-1)\left[\frac{k}{\pi}
\notag\-1 +\frac{k(d_i - Q)}{\pi q_s}\right]^{\beta-1}\bigg],
\end{align}
where $$\sum_{i=j+1}^{I}\frac{w(p_i)}{\sum_{i=1}^jw(p_i)}(-\lambda)\beta(\frac{k}{\pi}-1){\left[\frac{k}{\pi}-1 +\frac{k(d_i-Q)}{\pi q_s}\right]}^{\beta-1}$$ is increasing in $q_s$. This indicates that the equation $U'(s,q_s)=0$ has at most one solution $q_s$. Hence, we can find that the optimal $q_s^*$ is the solution of $U'(s,q_s)=0$, or at one of the boundary points $\kappa(Q-d_{j+1})/(\kappa-\pi)$ or $Q-d_{j-1}$.
When $q_s\in [\frac{\kappa(Q-d_{m+1})}{\kappa-\pi},\frac{\kappa(Q-d_{m})}{\kappa-\pi}], m=j+1,\ldots,\hat{\imath}-1$, we have $v[(\pi-\kappa)q_s+\kappa(Q-d_i)]\leq0$ for $i = m+1,\ldots,\hat{\imath}$, and $v[(\pi-\kappa)q_s+\kappa(Q-d_i)]\geq0$ for $i = j+1,\ldots,m$. From (\ref{eq:32}), we obtain
\begin{align}
&\frac{\partial U(s,q_s)}{\partial q_s}=\sum_{i=1}^jw(p_i){\pi}^\beta\beta q_s^{\beta-1}\notag\\
&+  \sum_{i=j+1}^m w(p_i)\beta(\pi - \kappa)[(\pi - \kappa)q_s + \kappa(Q - d_i)]^{\beta-1}\notag\\
&-    \sum_{i=m+1}^I   w(p_i)\beta(\kappa - \pi)\lambda[(\kappa - \pi) q_s - \kappa(Q - d_i)]^{\beta-1}\notag\\
&=\sum_{i=1}^jw(p_i){\pi}^\beta\beta q_s^{\beta-1}\Bigg[1 \notag\\
&+   \sum_{i=j+1}^m  \frac{w(p_i)}{\sum_{i=1}^jw(p_i)}(1  -  \frac{\kappa}{\pi}) \left[1  -  \frac{\kappa}{\pi}  +  \frac{\kappa(Q - d_i)}{\pi q_s}\right]^{\beta-1}
\notag\\
&+\sum_{i=m+1}^I\frac{w(p_i)}{\sum_{i=1}^jw(p_i)}(\frac{\kappa}{\pi} - 1)(-\lambda)\bigg[\frac{\kappa}{\pi} \notag\\
&- 1-\frac{\kappa(Q - d_i)}{\pi q_s}\bigg]^{\beta-1}\Bigg],
\end{align}
which follows a unimodal structure, and we can find that the optimal $q_s^*$ must lie at the critical point (i.e., the solution of $U'(s,q_s)=0$), or at one of the boundaries $\frac{\kappa(Q-d_{m+1})}{\kappa-\pi}$ or $\frac{\kappa(Q-d_{m})}{\kappa-\pi}$.
When $q_s\in[Q-d_{j},\frac{\kappa(Q-d_{\hat{\imath}})}{\kappa-\pi})$, we have $v[(\pi-\kappa)q_s+\kappa(Q-d_i)]>0, ~\forall i = j+1,\ldots,\hat{\imath}$. From (\ref{eq:32}), we obtain
\begin{align}
&\frac{\partial U(s,q_s)}{\partial q_s} = \sum_{i=1}^{\hat{\imath}}w(p_i){\pi}^{\beta} \beta q_s^{\beta-1} \notag\\
&+   \sum_{i=\hat{\imath}+1}^I(-\lambda)(\kappa - \pi)\beta{[(\kappa - \pi)q_s  + \kappa(d_i - Q)]}^{\beta-1},
\end{align}
which also follows a unimodal structure, and we can find that the optimal $q_s^*$ must lie at the critical point (i.e., the solution of $U'(s,q_s)=0$), or at one of the boundaries $Q-d_{j}$ or $\frac{\kappa(Q-d_{\hat{\imath}})}{\kappa-\pi}$.

(e) $\frac{\kappa(Q-d_{\hat{\imath}})}{k-\pi}<Q-d_j<Q-d_{j-1}<\frac{\kappa(Q-d_{j+1})}{k-\pi}$. In this case, without loss of generality, we assume $\frac{\kappa(Q-d_{\hat{m+1}})}{k-\pi}<Q-d_{j}<\frac{\kappa(Q-d_{\hat{m}})}{k-\pi}$, and $\frac{\kappa(Q-d_{\hat{l+1}})}{k-\pi}<Q-d_{j}<\frac{\kappa(Q-d_{\hat{l}})}{k-\pi}$.
When $q_s\in [\frac{\kappa(Q-d_{m+1})}{\kappa-\pi},\frac{\kappa(Q-d_{m})}{\kappa-\pi}], m=\hat{m+1},\ldots,\hat{l}$, we have $v[(\pi-\kappa)q_s+\kappa(Q-d_i)]\leq0$ for $i = m+1,\ldots,\hat{\imath}$, and $v[(\pi-\kappa)q_s+\kappa(Q-d_i)]\geq0$ for $i = j+1,\ldots,m$. From (\ref{eq:32}), we obtain
\begin{align}
&\frac{\partial U(s,q_s)}{\partial q_s}=\sum_{i=1}^jw(p_i){\pi}^\beta\beta q_s^{\beta-1} \notag\\
&+   \sum_{i=j+1}^m w(p_i)\beta(\pi - \kappa)[(\pi - \kappa)q_s + \kappa(Q - d_i)]^{\beta-1}\notag\\
&-     \sum_{i=m+1}^I  w(p_i)\beta(\kappa - \pi)\lambda[(\kappa - \pi) q_s - \kappa(Q - d_i)]^{\beta-1}\notag\\
&=\sum_{i=1}^jw(p_i){\pi}^\beta\beta q_s^{\beta-1}\Bigg[1 \notag\\
&+     \sum_{i=j+1}^m   \frac{w(p_i)}{\sum_{i=1}^jw(p_i)}(1 - \frac{\kappa}{\pi})  \left[1  -  \frac{\kappa}{\pi}  +  \frac{\kappa(Q - d_i)}{\pi q_s}\right]^{\beta-1}
\notag\\
&-       \sum_{i=m+1}^I   \frac{w(p_i)\lambda}{\sum_{i=1}^jw(p_i)}(\frac{\kappa}{\pi}  -  1)  \left[ \frac{\kappa}{\pi}  -  1  -  \frac{\kappa(Q - d_i)}{\pi q_s}\right]^{\beta-1}   \Bigg],
\end{align}
which follows a unimodal structure, and we can find that the optimal $q_s^*$ must lie at the critical point (i.e., the solution of $U'(s,q_s)=0$), or at one of the boundaries $\frac{\kappa(Q-d_{m+1})}{\kappa-\pi}$ or $\frac{\kappa(Q-d_{m})}{\kappa-\pi}$.

(f) $\frac{\kappa(Q-d_{\hat{\imath}})}{k-\pi}>Q-d_{j-1}$. In this case, $q_s\in[Q-d_{j},\frac{\kappa(Q-d_{\hat{\imath}})}{\kappa-\pi})$, and we have $v[(\pi-\kappa)q_s+\kappa(Q-d_i)]>0, ~\forall i = j+1,\ldots,\hat{\imath}$. From (\ref{eq:32}), we obtain
\begin{align}
&\frac{\partial U(s,q_s)}{\partial q_s} = \sum_{i=1}^{\hat{\imath}}w(p_i){\pi}^{\beta} \beta q_s^{\beta-1} 
\notag\\
&+    \sum_{i=\hat{\imath}+1}^I (-\lambda)(\kappa - \pi)\beta{[(\kappa - \pi)q_s  + \kappa(d_i - Q)]}^{\beta-1},
\end{align}
which also follows a unimodal structure, and we can find that the optimal $q_s^*$ must lie at the critical point (i.e., the solution of $U'(s,q_s)=0$), or at one of the boundaries $Q-d_{j}$ or $\frac{\kappa(Q-d_{\hat{\imath}})}{\kappa-\pi}$.

Hence, combining the above six cases, we have the optimal selling quantity $$q_s^*=\arg\max_{q_s\in\{Q-d_i,i=1,\ldots,\hat{\imath}\}\cup\mathcal{X}_{sh}} \{U(s,q_s)\}.$$

\textbf{Case II}: $q_s\in[Q-d_1,\infty)$. In this case, the expected utility from (\ref{equ:stage2s}) is decreasing in $q_s$, hence the optimal $q_s$ in this range is $q_s^*=Q-d_1$.

\textbf{Case III}: $q_s\in[0,Q-d_{\hat{\imath}}]$. In this case, the satisfaction loss under low demands is $L(Q-q_s-d_{i})=0$ for $i=1,\ldots,\hat{\imath}$, and the satisfaction loss under high demands is $L(Q-q_s-d_{i})=\kappa(Q-q_s-d_i)$ for $i=\hat{\imath}+1,\ldots,I$. The expected utility from (\ref{equ:stage2s}) is 
\begin{align}
&U(s,q_s)  = \sum_{i=1}^{\hat{\imath}} w(p_i)v[\pi q_s]  \notag\\
&+     \sum_{i=\hat{\imath}+1}^I   w(p_i)v[(\pi - \kappa)q_s + \kappa(Q - d_i)]\notag\\
&=\sum_{i=1}^{\hat{\imath}} w(p_i){(\pi q_s)}^{\beta} \notag\\
&-    \sum_{i=\hat{\imath}+1}^{I}w(p_i)\lambda{[(\kappa - \pi)q_s  + \kappa(d_i - Q)]}^{\beta}.
\end{align}
The first order derivative $\partial U(s,q_s)/\partial q_s$ follows a unimodal structure, and we can find that $q_s^*$ must lie at the critical point (i.e., the solution of $U'(s,q_s)=0$), or at one of the boundaries $0$ or $Q-d_{\hat{\imath}}$.

Combining the analysis in Cases I, II, and III, we have\\
$q_s^*=\arg\max_{q_s\in\{Q-d_i,i=1,\ldots,\hat{\imath}\}\cup\mathcal{X}_{sh}\cup\{0\}} \{U(s,q_s)\}$.

\noindent\emph{2) Seller's Problem in (\ref{equ:stage2s}) Under PT with $R_p =\kappa(Q - d_I)$:}
\textbf{Case I}: $q_s\in[Q-d_{j},Q-d_{j-1}], ~j=2,\ldots,\hat{\imath}$. In this case, the satisfaction loss under low demands is $L(Q+q_s-d_{i})=0$ for $i=1,\ldots,j-1$, and the satisfaction loss under high demands is $L(Q+q_s-d_l)=\kappa(Q+q_s-d_j)$ for $i=j,\ldots,I$. The expected utility from (6) is
\begin{align}
&U(s,q_s) = \sum_{i=1}^{j-1}w(p_i)V[-\pi q_s - \kappa(Q - d_{I})] \notag\\
&+  \sum_{i=j}^I w(p_i)V[\kappa(Q-d_i) - \pi q_s + \kappa q_s  -  \kappa(Q - d_I)]\notag\\
&=\sum_{i=1}^{j-1}w(p_i){[-\pi q_s + \kappa(d_I-Q)]}^{\beta}\notag\\
&+\sum_{i=j}^{I}w(p_i){[(\kappa- \pi)q_s +\kappa(d_I-d_i)]}^{\beta}.
\end{align}

The second order partial derivative of $U(s,q_s)$ with respect to $q_s$ is
\begin{align}
&\frac{\partial U^2(q_s)}{\partial q_s^2}  =  \sum_{i=1}^{j-1}w(p_i){(-\pi)}^2(\beta-1)\beta  {[  \kappa (d_I - Q)-\pi q_s  ]}^{\beta-2} \notag\\
&+   \sum_{i=j}^I w(p_{i}){(\kappa - \pi)}^2(\beta - 1)\beta {[(\kappa - \pi)q_s + \kappa(d_I - d_i)]}^{\beta -2} \notag\\
&< 0,
\end{align}
which implies that $U(s,q_s)$ is a concave function in $q_s$, and the optimal solution must lie at the critical point, where $U'(s,q_s)=0$, or one of the boundary points. 

Hence, $q_s^*=\arg\max_{q_s\in\{Q-d_i,i\in\mathcal{I}\}\cup\mathcal{X}_{sl}} \{U(s,q_s)\}$.

\textbf{Case II}: $q_s\in[Q-d_1,\infty)$. In this case, the expected utility from (6) is decreasing in $q_s$, hence the optimal $q_s$ in this range is $q_s^*=Q-d_1$.

\textbf{Case III}: $ q_s\in[0, Q-d_{\hat{\imath}}]$. In this case, the satisfaction loss under low demands is $L(Q+q_s-d_{i})=0$ for $i<\hat{\imath}$, and the satisfaction loss under high demands is $L(Q+q_s-d_l)=\kappa(Q+q_s-d_i)$ for $i\geq\hat{\imath}$. The expected utility from (6) is
\begin{align}
&U(s,q_s) = \sum_{i=1}^{\hat{\imath-1}}w(p_i)V[-\pi q_s - \kappa(Q - d_{I})] \notag\\
&+ \sum_{i=\hat{\imath}}^I w(p_i)V[\kappa(Q - d_i) - \pi q_s + \kappa q_s  -  \kappa(Q - d_I)].
\end{align}

Similar to Case I, we know that $U(s,q_s)$ is a concave function in $q_s$, and the optimal solution must lie at the critical point, where $U'(s,q_s)=0$, or one of the boundary points. 

Hence, $q_s^*=\arg\max_{q_s\in\{0,Q-d_{\hat{\imath}}\}\cup\mathcal{X}_{sh}\cup\{0\}} \{U(s,q_s)\}$.

Combining the analysis in Case I, Case II, and Case III, we have\\ 
 $q_s^*=\arg\max_{q_s\in\{Q-d_i,i\in\mathcal{I}\}\cup\mathcal{X}_{sh}\cup\{0\}} \{U(s,q_s)\}$.

%%%%%%%%%%%%%%%%%%%%%%%%%%%%%%%%%%%%%%%%%%%%%%%%%%%%%%%%%%%%%%%%%%%%%%%%%%%%
\subsection{Proof of Theorem 4}
In the proof of all three cases, we divide the feasible set of selling quantity $q_s$ into two subsets, $[0,Q-d_l]$ and $[Q-d_l,\infty)$, and analyze the optimal selling quantity $q_s^*$ that maximizes $U(s,q_s)$ within each subset. Such a division is due to the fact that the satisfaction loss $L(Q-q_s-d_l)=\kappa(Q-q_s-d_l)$ when $q_s\in [Q-d_l,\infty)$ which simplifies our analysis.

\subsubsection{Seller's Problem under EUT}
\begin{itemize}
	\item Case I: $q_s\in [0,Q-d_l]$. In this case, the satisfaction loss under low demand is $L(Q-q_s-d_l)=0$, and the satisfaction loss under high demand is $L(Q-q_s-d_h)=\kappa(Q-q_s-d_h)$. The expected utility is
	\begin{align}
	U(s,q_s) =(\pi_b^{\max}-\kappa p)q_s+\kappa p(Q-d_h),
	\end{align}
	which is a linear function in $q_s$, hence increasing in $q_s$ when $\pi_b^{\max}>\kappa p$ and decreasing in $q_s$, when $\pi_b^{\max}<\kappa p$. The optimal selling quantity is then $q_s^*=Q-d_{l}$ when $\pi_b^{\max}>\kappa p$ and $q_s^*=0$ when $\pi_b^{\max}<\kappa p$. When $\pi_b^{\max}=\kappa p$, the utility is independent of $q_s$. Without loss of generality, we assume that $q_s^*=0$ when $\pi_b^{\max}=\kappa p$.
	\item Case II: $q_s\in [Q-d_l,\infty)$. In this case, the satisfaction loss under low demand is $L(Q-q_s-d_l)=\kappa(Q-q_s-d_l)$, and the satisfaction loss under high demand is $L(Q-q_s-d_h)=\kappa(Q-q_s-d_h)$. The expected utility is  
\begin{align}
	U(s,q_s)=(\pi_b^{\max}-\kappa)q_s+\kappa Q-\kappa d_h p-\kappa d_l(1-p).
\end{align}	
 Since the utility function $U(s,q_s)$ is linearly decreasing in $q_s$, we have $q_s^*=Q-d_l$ in this case.
\end{itemize}
Combing the above analysis, we obtain (20).

\subsubsection{Seller's Problem under PT with $R_p=0$}
\begin{itemize}
\item Case I: $q_s\in [0,Q-d_l]$. In this case, the satisfaction loss under low demand is $L(Q-q_s-d_l)=0$, and the satisfaction loss under high demand is $L(Q-q_s-d_h)=\kappa(Q-q_s-d_h)$. The expected utility is  
\begin{align}
	U(s,q_s)=&-\lambda((\kappa-\pi_b^{\max}) q_s+\kappa (d_{h}-Q))^{\beta}w(p)\notag\\
	&+(\pi_b^{\max} q_s)^{\beta}w(1-p).
\end{align}

\begin{pps}
There is at most one local maximum point of $U(s,q_s)$ in the case $q_s\in [0,d_h-Q]$. When $1>\frac{\lambda(\kappa -\pi_b^{\max})^\beta w(p)}{{\pi_b^{\max}}^\beta w(1-p)}(1+\frac{\kappa (Q-d_h)}{( \kappa -\pi_b^{\max})(Q-d_l)})^{\beta-1}$, the local maximum point is at the right boundary point $q_s=d_h-Q$. When $1\leq\frac{\lambda(\kappa -\pi_b^{\max})^\beta w(p)}{{\pi_b^{\max}}^\beta w(1-p)}(1+\frac{\kappa (Q-d_h)}{( \kappa -\pi_b^{\max})(Q-d_l)})^{\beta-1}$, the local maximum point is at an interior point $q_s=\frac{\frac{\kappa }{\kappa -\pi_b^{\max}}(Q-d_h)}{(\frac{w(1-p){\pi_b^{\max}}^{\beta}}{w(p)\lambda(\kappa -\pi_b^{\max})^{\beta}})^{\frac{1}{\beta-1}}-1}$.
\end{pps}

\begin{proof}
The first order partial derivative of $U(s,q_s)$ with respect to $q_s$ is
	\begin{align}
	&\frac{\partial U(s,q_s)}{\partial q_s}=	\beta(\pi_b^{\max} q_s)^{\beta-1}\pi_b^{\max} w(1-p)
	\notag\\
	&-\lambda\beta( \kappa  - \pi_b^{\max})( ( \kappa  - \pi_b^{\max} ) q_s + \kappa (d_h - Q))^{\beta - 1}w(p ).
	\end{align}
	We define 
	\begin{align}
	g(q_s) = 1-\frac{w(p)\lambda(\kappa -\pi_b^{\max})^{\beta}}{w(1-p){\pi_b^{\max}}^{\beta}}(1+\frac{\kappa (Q-d_h)}{q_s})^{\beta-1},
	\end{align}
	and we can rewrite 
	$${\partial U(s, q_s )}/{\partial q_s}=w(1-p){\pi_b^{\max}}^{\beta}q_s^{\beta-1}\beta g(q_s),$$ 
	where $w(1-p){\pi_b^{\max}}^{\beta}q_s^{\beta-1}\beta>0$. The function $g(q_s)$ is a strictly decreasing function of $q_s$, which means the first order partial derivative  ${\partial U(s, q_s )}/{\partial q_s}$ will only be zero at most once, thus at most one local maximum point. 
		
	We then consider the two boundary points: $q_s=\epsilon$ and $q_s=Q-d_l-\epsilon$, with $\epsilon$ being a small positive number approaching zero (i.e. $\epsilon\rightarrow 0^+$).
	
	1) When $q_s=\epsilon$, we have
	\begin{align}
	\lim_{\epsilon\rightarrow 0^+}  g(\epsilon)  =  & \lim_{\epsilon\rightarrow 0^+}   1 \notag\\
	&- \frac{w(p)\lambda(\kappa  - \pi_b^{\max})^{\beta}}{w(1 - p){\pi_b^{\max}}^{\beta}}( 1 + \frac{\kappa (Q - d_l)}{\epsilon})^{\beta-1}\nonumber\\
	=&\infty.
	\end{align}
	2) When $q_s=Q-d_l-\epsilon$, we have:
	\begin{align}
	\lim_{\epsilon\rightarrow 0^+}   g(\epsilon)  &=    \lim_{\epsilon\rightarrow 0^+}   1 \notag\\
	&- \frac{w(p)\lambda(\kappa  - \pi_b^{\max})^{\beta}}{w(1 - p){\pi_b^{\max}}^{\beta}}( 1  +  \frac{\kappa (Q - d_l)}{Q - d_l - \epsilon})^{\beta-1}.
	\end{align}
	We can obtain 
	\begin{align}
&\lim_{\epsilon\rightarrow 0^+} g(Q-d_l-\epsilon) >0\notag\\ &\Leftrightarrow 1>\frac{\lambda(\kappa  - \pi_b^{\max})^\beta w(p)}{{\pi_b^{\max}}^\beta w(1 - p)}(1 \notag\\
&~~~~~~~~+ \frac{\kappa (Q-d_h)}{( \kappa  - \pi_b^{\max})(Q - d_l)})^{\beta-1}.
	\end{align}
	Since $g(q_s)$ and ${\partial U(s,q_s)}/{\partial q_s}$ have the same sign, from (63) we can obtain that 
	$${\partial U(s,Q-d_l-\epsilon)}/{\partial q_s}\geq0$$ when $$1\leq\frac{\lambda(\kappa -\pi_b^{\max})^\beta w(p)}{{\pi_b^{\max}}^\beta w(1-p)}(1+\frac{\kappa (Q-d_h)}{( \kappa -\pi_b^{\max})(Q-d_l)})^{\beta-1},$$ and $${\partial U(s,Q-d_l-\epsilon)}/{\partial q_s}<0$$ when $$1>\frac{\lambda(\kappa -\pi_b^{\max})^\beta w(p)}{{\pi_b^{\max}}^\beta w(1-p)}(1+\frac{\kappa (Q-d_h)}{( \kappa -\pi_b^{\max})(Q-d_l)})^{\beta-1}.$$ 
	For the same reason, from (61) we can obtain that $${\partial U(s,\epsilon)}/{\partial q_s}>0.$$ Thus ${\partial U(s, q_s )}/{\partial q_s}$ may be all positive within the interval $[0,Q-d_l]$ or first positive then negative within that interval based on the value of $\pi_b^{\max}$, $\mu$, $\beta$ and $\lambda$. 
	
	To sum up, the optimal selling amount $q_s^*$ in the case $q_s\in[0,Q-d_l]$ depends on the value of $\pi_b^{\max}$, $\mu$, $\beta$ and $\lambda$ as follows. (i) When $1\leq\frac{\lambda(\kappa -\pi_b^{\max})^\beta w(p)}{{\pi_b^{\max}}^\beta w(1-p)}(1+\frac{\kappa (Q-d_h)}{( \kappa -\pi_b^{\max})(Q-d_l)})^{\beta-1}$, the local maximum point is at an interior point $q_s^*=\frac{\frac{\kappa }{\kappa -\pi_b^{\max}}(Q-d_h)}{(\frac{w(1-p){\pi_b^{\max}}^{\beta}}{w(p)\lambda(\kappa -\pi_b^{\max})^{\beta}})^{\frac{1}{\beta-1}}-1}$. (ii) When $1>\frac{\lambda(\kappa -\pi_b^{\max})^\beta w(p)}{{\pi_b^{\max}}^\beta w(1-p)}(1+\frac{\kappa (Q-d_h)}{( \kappa -\pi_b^{\max})(Q-d_l)})^{\beta-1}$, the local maximum point is at the right boundary $q_s^*=Q-d_l$.

\end{proof}
	
\item Case II: $q_s\in [Q-d_l,\infty)$. In this case, the satisfaction loss under low demand is $L(Q-q_s-d_l)=\kappa(Q-q_s-d_l)$, and the satisfaction loss under high demand is $L(Q-q_s-d_h)=\kappa(Q-q_s-d_h)$. The expected utility is 
\begin{align}
	U(s,q_s)=&-\lambda((\kappa-\pi_b^{\max}) q_s+\kappa (d_{h}-Q))^{\beta}w(p)\notag\\
	&-\lambda((\kappa-\pi_b^{\max}) q_s+\kappa (d_{l}-Q))^{\beta}w(1-p).
\end{align}	
 Since the first order partial derivative ${\partial U(s,q_s)}/{\partial q_s}<0$, $U(s,q_s)$ is a decreasing function of $q_s$, and $q_s^*=Q-d_{l}$ in this case. 
\end{itemize}

Combing the above analysis, we obtain (21).

\subsubsection{Seller's Problem under PT with $R_p=\kappa(Q-d_h)$}
\begin{itemize}
\item Case I: $q_s\in [0,Q-d_l]$. In this case, the satisfaction loss under low demand is $L(Q-q_s-d_l)=0$, and the satisfaction loss under high demand is $L(Q-q_s-d_h)=\kappa(Q-q_s-d_h)$. The expected utility is  
\begin{align}
	U(s,q_s)=&-\lambda((\kappa-\pi_b^{\max}) q_s)^{\beta}w(p)\notag\\
	&+(\pi_b^{\max} q_s+\kappa (d_{h}-Q))^{\beta}w(1-p).
\end{align}

\begin{pps}
The maximum point of $U(s,q_s)$ in the case $q_s\in [0,Q-d_l]$ lies at one of the boundary points. When $U(s,0)\geq U(s,Q-d_l)$, the maximum point is at the left boundary point $q_s^*=0$. When $U(s,0)< U(s,Q-d_l)$, the maximum point is at the right boundary point $q_s^*=Q-d_l$.
\end{pps}
 
\begin{proof}
The first order partial derivative of $U(s,q_s)$ with respect to $q_s$ is
	\begin{align}
	&\frac{\partial U(s, q_s )}{\partial q_s}=-\lambda\beta( \kappa  - \pi_b^{\max})^{\beta} q_s^{\beta-1} w(p)\notag\\
	&+\beta\pi_b^{\max}[\pi_b^{\max} q_s + \kappa(d_h - Q)]^{\beta-1}w(1 - p).
	\end{align}
	We define 
	\begin{align}
	g(q_s)  =&\frac{w(1 - p){\pi_b^{\max}}}{\lambda w(p)}[\frac{\pi_b^{\max}}{\kappa - \pi_b^{\max}} + \frac{\kappa (d_h - Q)}{(\kappa - \pi_b^{\max})q_s}]^{\beta-1}\notag\\
	&- 1,
	\end{align}
	and we can rewrite $${\partial U(s, q_s )}/{\partial q_s}=\lambda\beta( \kappa  - \pi_b^{\max})^{\beta} q_s^{\beta-1} w(p) g(q_s),$$ where $$\lambda\beta( \kappa  - \pi_b^{\max})^{\beta} q_s^{\beta-1} w(p)>0.$$ The function $g(q_s)$ is a strictly increasing function of $q_s$, which means the first order partial derivative  ${\partial U(s,q_s)}/{\partial q_s}$ will only be zero at most once, and the point that satisfies the first order condition is a local minimum point. Hence the maximum point of $U(s,q_s)$ in the case $q_s\in [0,Q-d_l]$ lies at one of the boundary points. We compare the corresponding $U(s,q_s)$ with $q_s=0$ and $q_s=Q-d_l$ to find the optimal solution in this case.
	
	To sum up, the optimal selling amount $q_s^*$ in the case $q_s\in[0,Q-d_l]$ depends on the value of $\pi_b^{\max}$, $\mu$, $\beta$ and $\lambda$ as follows. (i) When $\lambda w(p)[(\kappa-\pi_b^{\max})(Q-d_l)]^{\beta}\geq w(1-p)\{[(\pi_b^{\max}-\kappa)Q+\kappa d_h-\pi_b^{\max} d_l]^{\beta}-[\kappa(d_h-Q)]^{\beta}\}$, the local maximum point is at the left boundary point $q_s=0$. (ii) When $\lambda w(p)[(\kappa-\pi_b^{\max})(Q-d_l)]^{\beta}<w(1-p)\{[(\pi_b^{\max}-\kappa)Q+\kappa d_h-\pi_b^{\max} d_l]^{\beta}-[\kappa(d_h-Q)]^{\beta}\}$, the local maximum point is at the right boundary point $q_s=Q-d_l$.

\end{proof}
	
\item Case II: $q_s\in [Q-d_l,\infty)$. In this case, the satisfaction loss under low demand is $L(Q-q_s-d_l)=\kappa(Q-q_s-d_l)$, and the satisfaction loss under high demand is $L(Q-q_s-d_h)=\kappa(Q-q_s-d_h)$. The expected utility is 
\begin{align}
	&U(s,q_s) = - \lambda((\kappa-\pi_b^{\max}) q_s)^{\beta}w(p)\notag\\
	& - \lambda((\kappa - \pi_b^{\max}) q_s + \kappa (d_h - d_{l} - 2Q))^{\beta}w(1 - p).
\end{align}	
 Since the first order partial derivative ${\partial U(s,q_s)}/{\partial q_s}<0$, the utility function $U(s,q_s)$ is a decreasing function of $q_s$, and $q_s^*=Q-d_{l}$ in this case. 
\end{itemize}

Combing the above analysis, we obtain (22).

%%%%%%%%%%%%%%%%%%%%%%%%%%%%%%%%%%%%%%%%%%%%%%%%%%%%%%%%%%%%%%%%%%%%%%%%%%%%
\subsection{Detailed Dynamic Data Trading Algorithm}
Next, we present an algorithm for computing the dynamic data trading decisions starting at day $\hat{\jmath}$. At the beginning of each day, the user obtains the demand prediction according to our sliding window method, and the remaining quota since yesterday's trade. The pseudo code is shown in Algorithm 2.

\begin{algorithm}
\caption{Dynamic Trading Decision}
\textbf{Input}: Remaining quota ($Q_{\hat{\jmath}}$), risk parameters ($\mu$, $\beta$, $\lambda$, $R_p$), demand prediction ($d_i$, $p_i$, $\forall i \in \mathcal{I}$), market information ($\pi_s^{\min}$, $\pi_b^{\max}$).\\ $//$ Initialize the risk parameters and personal profiles.\\
%\KwOut{$con(r_i)$}

\For{$j=J$ \text{to} $T$}
{
\For{$i=1$ \text{to} $I$}{
$d_{i}:= d_{i-1}-\delta_{\hat{m}-i,j-1}$; $//$ Update the prediction of demand based on current usage.\\
}
Calculate $q^*_b$ by solving Problem (\ref{equ:stage2b}) and $q^*_s$ by solving Problem (\ref{equ:stage2s})\;
\If{$U(q^*_b)>U(q^*_s)$}{$q_j^*:= q^*_b$ and $a^*:= b$;}
\Else{$q_j^*:= -q^*_s$ and $a^*:= s$;}
$Q_{j+1}:= Q_{j}+q_{j^*}-\delta_{\hat{m},\hat{\jmath}}$; $//$ Update the remained quota after every trade.\\
}
%return $con(r_i)$\;
\end{algorithm}
%\textbf{\renewcommand{\baselinestretch}{1.65}}

%%%%%%%%%%%%%%%%%%%%%%%%%%%%%%%%%%%%%%%%%%%%%%%%%%%%%%%%%%%%%%%%%%%%%%%%%%%%%%
\subsection{Proof of Unique solution by Indifference Equations (25) and (26)}
We show the unique solution of $\beta$ by the indifference equations (25) and (26) by the following theorem.
\begin{thm}
 Given $0\leq p\leq1$, $1<B<A<C$, the curve $f_A(\beta)=A^\beta$ and the curve $f_{BC}(\beta)=pB^{\beta}+(1-p)C^{\beta}$ have at most one intersection in the interval $\beta \in (0,1)$.
 \end{thm}
 \begin{proof}
 To show $f_A(\beta)$ and $f_{BC}(\beta)$ have at most one intersection, we need to show the solution $A(\beta)=(pB^\beta+(1-p)C^\beta)^{\frac{1}{\beta}}$ is monotone in $\beta$.
 
 We write 
 \begin{align}
 A(\beta)=(pB^\beta+(1-p)C^\beta)^{\frac{1}{\beta}}.
 \end{align}
 To find first order derivative of $A(\beta)$, we first write $$g(\beta)=\ln(A(\beta))=\frac{1}{\beta}\ln(pB^\beta+(1-p)C^\beta).$$
 Hence, the first order derivative of $g(\beta)$ 
 $$g'(\beta)=\frac{1}{A(\beta)}A'(\beta),$$
 which implies
 \begin{align}
 A'&(\beta)=g'(\beta)A(\beta)\notag\\
 =&(pB^\beta + (1 - p)C^\beta)^{\frac{1}{\beta}}[-\frac{1}{\beta^2}\ln(pB^\beta + (1 - p)C^\beta) \notag\\
 &+ \frac{1}{\beta}\frac{1}{pB^\beta + (1 - p)C^\beta}(pB^\beta\ln B + (1 - p)C^\beta\ln C)]\notag\\
 =&\frac{(pB^\beta + (1 - p)C^\beta)^{\frac{1}{\beta}}}{\beta^2(pB^\beta + (1 - p)C^\beta)}[\beta(pB^\beta\ln B + (1 - p)C^\beta\ln C)\notag\\
 &-(pB^\beta + (1 - p)C^\beta)\ln(pB^\beta + (1 - p)C^\beta)]\notag\\
  =&\frac{(pB^\beta + (1 - p)C^\beta)^{\frac{1}{\beta}}}{\beta^2(pB^\beta + (1 - p)C^\beta)}[(pB^\beta\ln B^\beta + (1 - p)C^\beta\ln C^\beta)\notag\\
  &- (pB^\beta + (1 - p)C^\beta)\ln(pB^\beta + (1 - p)C^\beta)].
 \end{align}
 
 Then we study the second order derivative of the function $h(x)=x\ln x$:
 $$h''(x)>0,$$
 which implies
 $$h(px_1+(1-p)x_2)<ph(x_1)+(1-p)h(x_2).$$
 Hence, we know that 
 \begin{align}
     &- (pB^\beta + (1 - p)C^\beta)\ln(pB^\beta + (1 - p)C^\beta) \notag\\
     &+ (pB^\beta\ln B^\beta + (1 - p)C^\beta\ln C^\beta)
     >0,
 \end{align}
 and then 
 $$A'(\beta)>0.$$
 Since $A(\beta)$ is monotone in $\beta$, we prove Theorem 5. 
 
 \end{proof}

From Theorem 5, we can find the unique $\beta$ value by solving the indifference equations. After we find the $\beta$ value, we substitute this value into (26) only to find the value of $\lambda$, because the $\lambda$ in both sides of (25) can be cancelled, and then (25) is not a function in $\lambda$. Since (26) is a linear equation in $\lambda$, we can find the unique solution of $\lambda$.

\begin{IEEEbiography}
[{\includegraphics[width=1in,height=1.25in,clip,keepaspectratio]{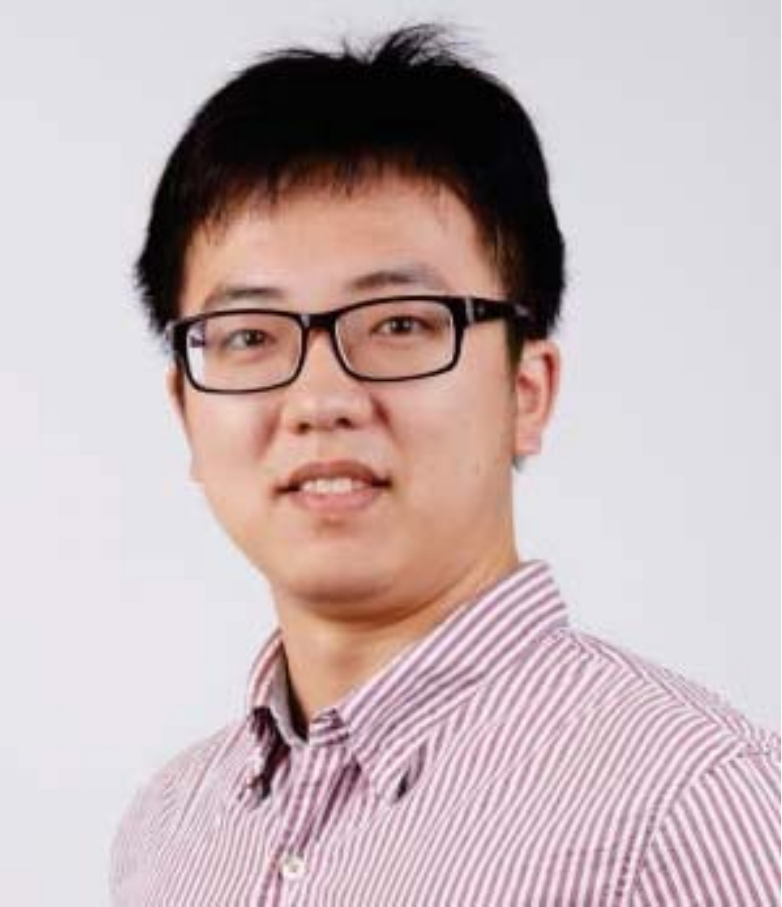}}]
{Junlin Yu} (S'14) is working towards his Ph.D. degree in the Department of Information Engineering at the Chinese University of Hong Kong. His research interests include behavioral economical studies in wireless communication networks, and optimization in mobile data trading. He is a student member of IEEE.
\end{IEEEbiography}

\begin{IEEEbiography}
[{\includegraphics[width=1in,height=1.25in,clip,keepaspectratio]{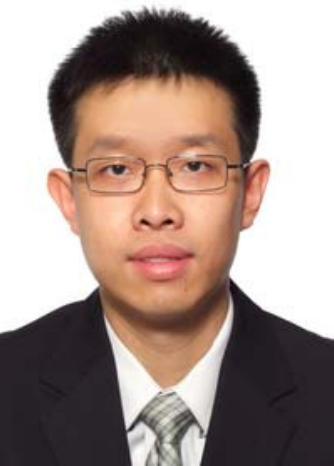}}]
{Man Hon Cheung} received the B.Eng. and M.Phil. degrees in Information Engineering from the Chinese University of Hong Kong (CUHK) in 2005 and 2007, respectively, and the Ph.D. degree in Electrical and Computer Engineering from the University of British Columbia (UBC) in 2012.
 Currently, he is a postdoctoral fellow in the Department of Information Engineering in CUHK.
 He received the IEEE Student Travel Grant for attending {\it IEEE ICC 2009}. He was awarded the Graduate Student International Research Mobility Award by UBC, and the Global Scholarship Programme for Research Excellence by CUHK.
 He serves as a Technical Program Committee member in {\it IEEE ICC}, {\it Globecom}, and {\it WCNC}.
 His research interests include the design and analysis of wireless network protocols using optimization theory, game theory, and dynamic programming, with current focus on mobile data offloading, mobile crowd sensing, and network economics.
\end{IEEEbiography}

\begin{IEEEbiography}[{\includegraphics[width=1in,height=1.25in,clip,keepaspectratio]{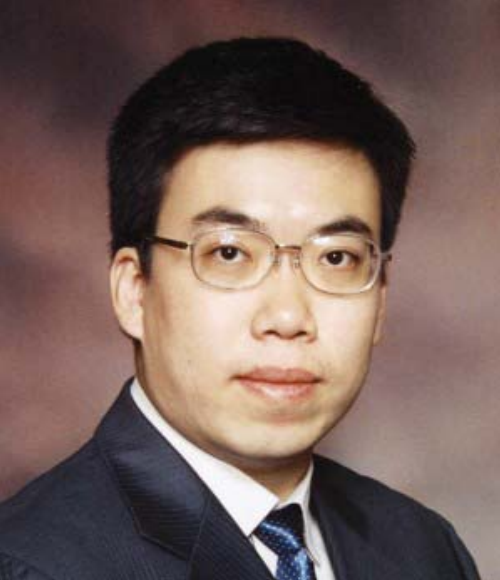}}]
{Jianwei Huang} (S'01-M'06-SM'11-F'16) is an Associate Professor and Director of the Network Communications and Economics Lab (ncel.ie.cuhk.edu.hk), in the Department of Information Engineering at the Chinese University of Hong Kong. He received the Ph.D. degree from Northwestern University in 2005, and worked as a Postdoc Research Associate at Princeton University during 2005-2007. Dr. Huang is the co-recipient of 8 Best Paper Awards, including IEEE Marconi Prize Paper Award in Wireless Communications in 2011. He has co-authored six books, including the textbook on ``Wireless Network Pricing.'' He received the CUHK Young Researcher Award in 2014 and IEEE ComSoc Asia-Pacific Outstanding Young Researcher Award in 2009. Dr. Huang has served as an Associate Editor of IEEE/ACM Transactions on Networking, IEEE Transactions on Cognitive Communications and Networking, IEEE Transactions on Wireless Communications, and IEEE Journal on Selected Areas in Communications - Cognitive Radio Series. He has served as the Chair of IEEE ComSoc Cognitive Network Technical Committee and Multimedia Communications Technical Committee. He is an IEEE Fellow, a Distinguished Lecturer of IEEE Communications Society, and a Thomson Reuters Highly Cited Researcher in Computer Science.
\end{IEEEbiography}

\begin{IEEEbiography}[{\includegraphics[width=1in,height=1.25in,clip,keepaspectratio]{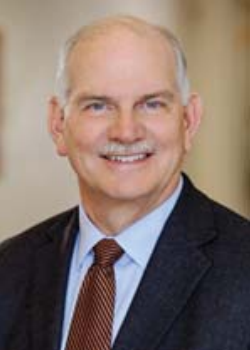}}]
{H. Vincent Poor} (S'72, M'77, SM'82, F'87) received the Ph.D. degree in EECS from Princeton University in 1977.  From 1977 until 1990, he was on the faculty of the University of Illinois at Urbana-Champaign. Since 1990 he has been on the faculty at Princeton, where he is currently the Michael Henry Strater University Professor of Electrical Engineering. During 2006 to 2016, he served as Dean of Princeton's School of Engineering and Applied Science. His research interests are in the areas of information theory, statistical signal processing and stochastic analysis, and their applications in wireless networks and related fields such as smart grid and social networks. Among his publications in these areas is the book \emph{Mechanisms and Games for Dynamic Spectrum Allocation} (Cambridge University Press, 2014).

Dr. Poor is a member of the National Academy of Engineering, the National Academy of Sciences, and is a foreign member of the Royal Society. He is also a fellow of the American Academy of Arts and Sciences, the National Academy of Inventors, and other national and international academies. He received the Marconi and Armstrong Awards of the IEEE Communications Society in 2007 and 2009, respectively. Recent recognition of his work includes the 2016 John Fritz Medal, the 2017 IEEE Alexander Graham Bell Medal, Honorary Professorships at Peking University and Tsinghua University, both conferred in 2016, and a D.Sc. \emph{honoris causa} from Syracuse University awarded in 2017.

\end{IEEEbiography}

\end{document}